\documentclass[a4paper]{article}

\usepackage{graphicx}

\usepackage[total={190mm,257mm},left=10mm,top=20mm]{geometry}
\usepackage{comment}
\usepackage{fancyhdr}
\usepackage{amsmath, adjustbox}
\usepackage{amssymb}
\usepackage{amsthm}
\usepackage{xcolor}
\usepackage{pgfplots}
\usepackage{nameref}
\usepackage{enumerate}
\pgfplotsset{compat=1.11}
\usepackage{dsfont}
\usetikzlibrary{decorations.markings,decorations.pathreplacing, arrows.meta}
\usepackage{subfigure}
\usepackage{bm}
\usepackage{tikz}
\usepackage{float}
\usepackage{tkz-euclide}
\usetikzlibrary{angles}
\usetikzlibrary{quotes}
\usepackage{soul}
\usepackage[hidelinks]{hyperref}
\usepackage{csquotes}
\usepackage{enumitem}
\usepackage{multicol}
\usepackage{currfile}

\usepackage{caption}

\newcommand{\R}{\mathbb{R}}
\newcommand{\N}{\mathbb{N}}

\newcommand{\dom}{\mathsf{D}}
\newcommand{\acos}{\operatorname{arccos}}

\newcommand{\C}[1]{\mathbf{C^{\pmb #1}}}
\newcommand{\Cc}[1]{\mathbf{C_c^{#1}}}
\renewcommand{\L}[1]{\mathbf{L^{\pmb #1}}}

\newcommand{\rss}{\mathcal{RS}_{\rm s}}
\newcommand{\rsc}{\mathcal{RS}_{\rm c}}
\renewcommand{\geq}{\geqslant}
\renewcommand{\ge}{\geqslant}
\renewcommand{\leq}{\leqslant}
\renewcommand{\le}{\leqslant}
\newcommand{\BV}{{\mathbf{BV}}}
\newcommand{\BVloc}{{\mathbf{BV}_{\rm loc}}}
\newcommand{\us}{u_{\rm s}}

\newcommand{\uc}{u_{\rm c}}

\newcommand{\qc}{q_{\rm c}^{0}}
\newcommand{\hc}{h_{\rm c}}
\newcommand{\ucm}{u_{\rm c}^-}
\newcommand{\rcm}{\rho_{\rm c}^-}

\newcommand{\ucp}{u_{\rm c}^+}
\newcommand{\rcp}{\rho_{\rm c}^+}

\newcommand{\qsm}{q_{\rm s}^-}

\newcommand{\qsp}{q_{\rm s}^+}
\newcommand{\qs}{q_{\rm s}^{0}}
\newcommand{\uspm}{u_{\rm s}^\pm}
\newcommand{\rspm}{\rho_{\rm s}^\pm}
\newcommand{\qspm}{q_{\rm s}^\pm}
\newcommand{\ucpm}{u_{\rm c}^\pm}
\newcommand{\rcpm}{\rho_{\rm c}^\pm}

\renewcommand{\d}{{\rm{d}}}

\newtheorem{teo}{Theorem}[section]
\newtheorem{prop}[teo]{Proposition}

\newtheorem{lem}[teo]{Lemma}

\newtheorem{defi}[teo]{Definition}

\theoremstyle{remark}
\newtheorem{rmk}[teo]{Remark}
\newtheorem{ex}[teo]{Example}

\usepackage{multirow}

\title{One-way pressure-independent valves\\with hysteresis and chattering}
\author{
Fabio Bagagiolo\footnote{Department of Mathematics, University of Trento, Italy}
\and
Andrea Corli\footnote{Department of Mathematics and Computer Science, University of Ferrara, Italy}
\and
Stefan Moreti\footnotemark[1]
\and
Massimiliano D. Rosini\footnote{Uniwersytet Marii Curie-Sk\l odowskiej, Plac Marii Curie-Sk\l odowskiej 1 20-031 Lublin, Poland, and Department of Business Administration, University of Chieti-Pescara, 66100 Chieti, Italy}
}

\begin{document}
\maketitle

\begin{multicols}{2}
{\scriptsize\tableofcontents\normalsize}
\end{multicols}

\begin{abstract}
We propose a model for the one-dimensional isothermal gas flow through a valve in a pipe. The valve aims at keeping the outgoing flow in an assigned interval; moreover, its opening and closing depend on a further parameter. The model  is realized by a system of two conservation laws with suitable conditions coupling the ingoing and the outgoing flow. We study the corresponding Riemann problem, where non-classical shocks show up. Although simple, the model allows to study two key oscillatory behaviors of the flow, which are due, respectively, to the chattering of the valve, especially seen in the numerical simulations, and to the stabilizing hysteresis induced by the valve operation. We characterize the states leading to chattering and show how they depend on the parameters of the model. Through several numerical simulations, we study the interplay between the two above phenomena.
\end{abstract}

\vspace{1cm}


\section{Introduction}
This paper deals with the modeling and analysis of a gas flow through a valve in a pipeline. More precisely, the flow is assumed to be isothermal, for simplicity; the case of isentropic flows could be modeled in an analogous way. The valve, which only acts one-way, when it is open selects the outgoing flow by aiming at keeping it in a tolerance interval $I_*^- \doteq [q_*^-,q_*]$, with $0<q_*^-<q_*$, possibly truncating it when the maximal allowed flow is larger than $q_*$; the valve closes when the ingoing flow is not sufficient to maintain the outgoing flow in that interval and opens when the ingoing flow is sufficient to keep the maximal allowed outgoing flow in $(q_*^+,+\infty)$, for some fixed $q_{*}^{+} > q_{*}$. In the remaining cases, the opening and closing of the valve depend on its past history by a hysteretic switching rule.
Our model has a last feature: it only allows subsonic or sonic outgoing flows (ingoing supersonic flows are admissible). In spite of the fact that supersonic outgoing flows may occur in specific applications \cite{Gugat-Herty-Muller, Modesti-Pirozzoli}, in most cases this condition is met and, in fact, desirable, in order to avoid malfunctions of the valve due to extreme situations. From a mathematical viewpoint, the avoidance of supersonic outgoing flows simplifies the analysis and excludes non-uniqueness of the solutions \cite{C1, C2}.

In the applied literature the modeling of a valve is often done through one or more ordinary differential equations, see for instance \cite{Bazso-Champneys-Hos, Hos-Bazso-Champneys} and references contained therein; these equations are deduced from a set of partial differential equations  expressing the conservation of mass and, when is the case, energy and momentum. In that case, the discontinuous opening and closing of a valve is sometimes understood in the framework of Filippov solutions \cite{Filippov}, see \cite{AMP, BHJ, TPN}. The mathematical modeling of hysteresis in a valve is less studied; see however  \cite{Krejci-SPP}. 
The basic feature of our approach is instead in the framework of systems of conservation laws, see \cite{C3,C1,C2,DGRR} (see also \cite{Colombo-Garavello_notions, Colombo-Goatin, Gugat-Herty, Herty_compressors} for related papers), and hysteresis is introduced as in \cite{vis}; we also refer to \cite{bagmor, bagmor2, Corli-Fan_2019, goamor} for other recent studies on relationships between conservation laws and hysteresis.

Two important phenomena occur in this problem: 

\begin{itemize}

\item 
in particular regimes, the valve can rapidly alternate opening and closing, the so called {\em chattering}, which results in the instability of the flow; 

\item 
the introduction of the tolerance interval $I_{*}^{-}$ and the particular nonlocal (in time) response of the valve induces a {\em hysteretic} behavior in the flow, which stabilizes the flow.

\end{itemize}
The two above issues may occur independently, as shown in this paper; however, see for instance \cite{C1, C2, DGRR} for a valve with no hysteresis but showing chattering.
The aim of this paper is to understand both phenomena and in particular the relationships between them. More precisely, we first explicitly characterize the initial states of suitable Riemann problems that lead to chattering; these states are called {\em incoherent}, and {\em coherent} states are those where chattering does not show up. Second, we study the effects of hysteresis on the flow. A rigorous analytic treatment of the general Cauchy problem with hysteresis for our framework would instead require the analysis of wave interactions for long times, because these effects are nonlocal in time; this could be done, but only at the cost of great technical difficulties. For this reason, we tackle this problem numerically. Our numerical scheme perfectly captures the evolution of the flow for coherent initial data, and shows interesting chattering patterns otherwise. At last, by letting the parameters $q_*^\pm$ vary, we explain in which way they \lq\lq regularize\rq\rq\ the flow. About the latter topic, one can intuitively imagine that increasing the interval $I_{*} \doteq [q_*^-,q_*^+]$ leads to more stable flows, and thus reduces chattering; however, the more stable the flow is, the more it deviates from the desired value $q_*$. This leads to a control problem, where one searches for the optimal parameters $q_*^\pm$ with respect to a cost functional. We briefly address this issue in the last part of the paper. 

\smallskip

Here follows how the paper is structured. In Section \ref{s:basic} we provide some general information about the isothermal Euler system, its Lax curves and the standard Riemann solver. In Section \ref{s:flowv} we show a general framework for the modeling of gas flow through valves; in particular we introduce here the notion of {\em coupling Riemann solver}. Such a general framework is specialized to the valve under consideration in Section \ref{s:valve-relay}; there, we also show in detail how the valve works and give some preliminary properties of the hysteresis relation. The main analytical results are contained in Section \ref{s:main}: they concern the explicit characterization of the coherent states and give precise information about the size of the incoherent region. The numerical scheme is introduced in the first part of Section \ref{s:num}: it is a Random Choice Method, to better follow the behavior of the flow around the valve. Then we test our scheme on coherent initial data, and verify that it perfectly captures the exact solution. Moreover, we investigate how chattering arises for incoherent initial data. In the subsequent part of Section \ref{s:num} we consider a strongly oscillatory ingoing flow, and check the response of the valve by letting the parameters $q_*^\pm$ vary. A brief mention to a control problem ends this section. The final Section \ref{s:final} draws conclusions and briefly shows how our results can be generalized to slightly different types of valves. 

\section{The free flow along a pipe}
\label{s:basic}
In this section we collect some basic facts about the modeling of an isothermal gas through a pipe.

\subsection{The system}

The flow along a pipe is described by the isothermal Euler equations \cite{LeVeque-book}
\begin{equation}
\label{e:system}
\left\{\begin{array}{@{}l@{}}
\rho_t + q_x = 0, \\
q_t + P_x = 0,
\end{array}\right.
\end{equation}
where $\rho = \rho(t, x)$ represents the \emph{density}, and $q = q(t, x)$ the \emph{momentum} at time $t \geq 0$ and position $x \in \R$. The term
\begin{equation*}
P(\rho, q) \doteq \frac{q^2}{\rho} + p(\rho)
\end{equation*}
is the \emph{dynamic pressure} and $p(\rho) = a^2 \, \rho$ is the \emph{pressure}, where $a > 0$ is the \emph{sound speed}.
The \emph{velocity} is
\begin{equation}
\label{e:v}
v(\rho, q) \doteq q/\rho.
\end{equation}
A state $u \doteq (\rho, q)$ is \emph{subsonic} if $|v(u)| < a$, \emph{sonic} if $|v(u)| = a$, and \emph{supersonic} if $|v(u)| > a$.
The functions $\rho \mapsto \pm a \, \rho$ describe the \emph{sonic curves} in the $(\rho, q)$-plane.
The eigenvalues of the Jacobian matrix corresponding to the flux function $f(\rho, q) \doteq (q, P(\rho, q))^T$ are $\lambda_1(u) \doteq v(u) - a$ and $\lambda_2(u) \doteq v(u) + a$.
System \eqref{e:system} is \emph{strictly hyperbolic} in $\Omega \doteq \left\{ (\rho, q) \in \R^2 : \rho > 0 \right\}$ because the inequality $\lambda_1(u) < \lambda_2(u)$ holds for all $u \in \Omega$. Both characteristic fields are genuinely nonlinear.

For completeness, we give the definitions of weak and entropy solution to the Cauchy problem for \eqref{e:system} with initial condition
\begin{equation}
\label{e:ini}
u(0,x) = u_o(x).
\end{equation}

\begin{defi}
\label{d:weak}
A function $u \in \C0([0, +\infty); \BV_{\rm loc}(\R;\Omega))$ is a weak solution of the Cauchy problem \eqref{e:system}, \eqref{e:ini} in $[0, +\infty) \times \R$ if it satisfies the following conditions, for any test function $\varphi \in \Cc\infty([0, +\infty) \times \R; \R)$:
\begin{equation}
\label{e:wcond}
\begin{aligned}
\int_0^{\infty} \int_{\R} \left(\rho \, \varphi_t + q \, \varphi_x\right) \, \d x \, \d t + \int_{\R} \rho_o(x) \, \varphi(0,x) \d x = 0,
\\
\int_0^{\infty} \int_{\R} \left(q \, \varphi_t + P(\rho,q) \, \varphi_x\right) \, \d x \, \d t + \int_{\R} q_o(x) \, \varphi(0,x) \d x = 0.
\end{aligned}
\end{equation}
\end{defi}

Recall that any piecewise regular weak solution $u = (\rho,q)$ satisfies the Rankine-Hugoniot conditions: if $x=\gamma(t)$ is a smooth curve along which $u$ is discontinuous then 
\begin{align}
\label{e:RH1}
&(\rho^+ - \rho^- ) \, \dot{\gamma} = q^+ - q^- ,\\
\label{e:RH2}
&(q^+ - q^- ) \, \dot{\gamma} = P(\rho^+,q^+) - P(\rho^-,q^-),
\end{align}
where $u^\pm(t) \doteq u(t,\gamma(t)^\pm)$ are the traces of $u$ on the sides $\pm(x-\gamma(t))>0$ along the discontinuity.

Weak solutions are not unique, in general. Here follows the definition of entropy solution; existence and uniqueness of such solutions holds for a large class of systems of conservation laws \cite{Bressan-book}.

\begin{defi}
\label{d:entro}
An entropy solution $u \in \C0([0, +\infty); \BV_{\rm loc}(\R;\Omega))$ of the Cauchy problem \eqref{e:system}, \eqref{e:ini} in $[0, +\infty) \times \R$ is a weak solution such that for any convex entropy pair $(\mathcal{E},\mathcal{Q})$ and test function $\varphi \in \Cc\infty((0, +\infty) \times \R; [0, +\infty))$:
\begin{equation}
\label{e:econd}
\int_0^{\infty} \int_{\R} \left(\mathcal{E}(u) \, \varphi_t + \mathcal{Q}(u) \, \varphi_x\right) \, \d x \, \d t \geq 0.
\end{equation}
\end{defi}
\noindent
Recall that $(\mathcal{E},\mathcal{Q})$ is an entropy pair provided that $D_u\mathcal{Q}=D_u\mathcal{E}\,D_uf$, where $f$ is the flux function introduced above.

\subsection{The Lax curves}

We refer to \cite{C3,C1,C2} for proofs and more details about this section. 
For fixed $u_o \doteq (\rho_o,q_o) \in \Omega$, we define the functions $\mathcal{FL}_i^{u_o}, \mathcal{BL}_i^{u_o} \colon (0, +\infty) \to \R$, $i\in\{1,2\}$ as 
\begin{align*}
\mathcal{FL}_1^{u_o}(\rho) &\doteq
\begin{cases}
\mathcal{R}_1^{u_o}(\rho)&\hbox{if }\rho \in (0,\rho_o],
\\
\mathcal{S}_1^{u_o}(\rho)&\hbox{if }\rho \in (\rho_o, +\infty),
\end{cases}&
\mathcal{FL}_2^{u_o}(\rho) &\doteq
\begin{cases}
\mathcal{S}_2^{u_o}(\rho)&\hbox{if }\rho \in (0,\rho_o),
\\
\mathcal{R}_2^{u_o}(\rho)&\hbox{if }\rho \in [\rho_o, +\infty),
\end{cases}
\\
\mathcal{BL}_1^{u_o}(\rho) &\doteq
\begin{cases}
\mathcal{S}_1^{u_o}(\rho)&\hbox{if }\rho \in (0,\rho_o),
\\
\mathcal{R}_1^{u_o}(\rho)&\hbox{if }\rho \in [\rho_o, +\infty),
\end{cases}&
\mathcal{BL}_2^{u_o}(\rho) &\doteq
\begin{cases}
\mathcal{R}_2^{u_o}(\rho)&\hbox{if }\rho \in (0,\rho_o],
\\
\mathcal{S}_2^{u_o}(\rho)&\hbox{if }\rho \in (\rho_o, +\infty),
\end{cases}
\end{align*}
with
\begin{align}
\mathcal{R}_i^{u_o}(\rho) &\doteq
\left(\frac{q_o}{\rho_o} + (-1)^i a \ln\left(\frac{\rho}{\rho_o}\right) \right) \rho,&
\mathcal{S}_i^{u_o}(\rho) 
&\doteq \left( \frac{q_o}{\rho_o} + (-1)^i a \left(\sqrt{\dfrac{\rho}{\rho_o}} - \sqrt{\dfrac{\rho_o}{\rho}}\right) \right) \rho.
\label{e:laxRS}
\end{align}
The graphs of the functions $\mathcal{FL}_i^{u_o}$ and $\mathcal{BL}_i^{u_o}$ are the \emph{forward} $\mathsf{FL}_i^{u_o}$ and \emph{backward} $\mathsf{BL}_i^{u_o}$ \emph{Lax curves} of the $i$-th family through $u_o$, see \figurename~\ref{f:RVLaxcurves}.
We point out that all figures representing Lax and related curves are obtained numerically.
Analogously, the \emph{shock} $\mathsf{S}_i^{u_o}$ and \emph{rarefaction} $\mathsf{R}_i^{u_o}$ \emph{curves} through $u_o$ are the graphs of the functions $\mathcal{S}_i^{u_o}$ and $\mathcal{R}_i^{u_o}$.
The $i$-shock speed between $(\rho_o,q_o)$ and $(\rho,q)$ is $s_{i}^{u_o}(\rho) \doteq v(\rho_o,q_o) + (-1)^i \, a \, \sqrt{\rho/\rho_o}$, where $v$ is given in \eqref{e:v}.

\begin{figure}[!htb]\centering
\begin{tikzpicture}[every node/.style={anchor=south west,inner sep=0pt},x=1mm, y=1mm]
\node at (4,4) {\includegraphics[width=50mm]{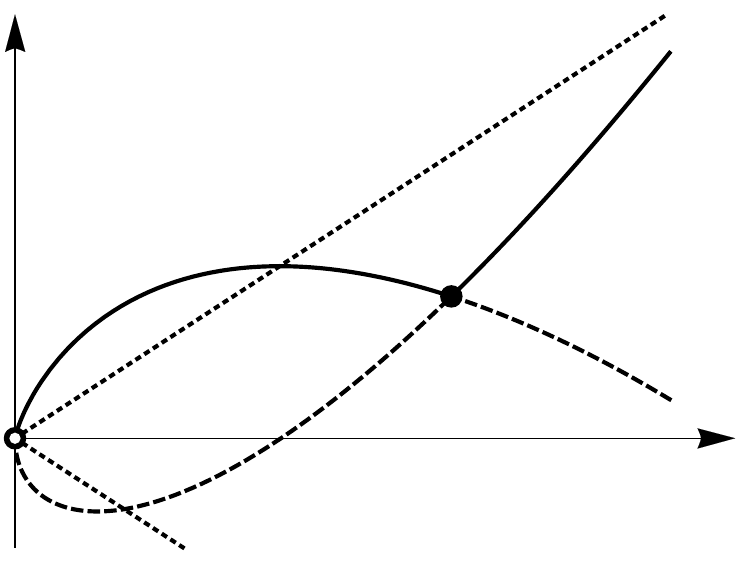}};
\node at (52,5) {\strut $\rho$};
\node at (0,38) {\strut $q$};
\node at (20,38) {\strut $\mathsf{FL}_1^{u_o} \cup \mathsf{FL}_2^{u_o}$};
\node at (15,24) {\strut $\mathsf{R}_1^{u_o}$};
\node at (45,28) {\strut $\mathsf{R}_2^{u_o}$};
\node at (21,6) {\strut $\mathsf{S}_2^{u_o}$};
\node at (42,18) {\strut $\mathsf{S}_1^{u_o}$};
\node at (31,22.5) {\strut $u_o$};

\begin{scope}[shift={(60,0)}]
\node at (4,4) {\includegraphics[width=50mm]{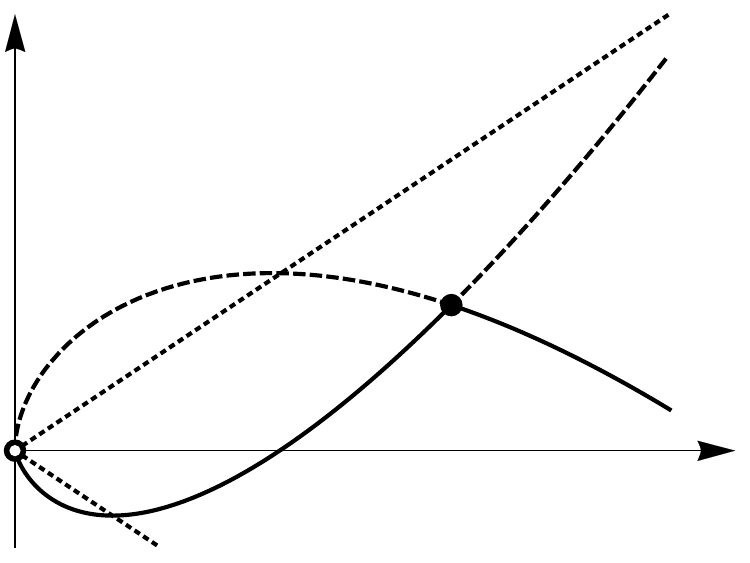}};
\node at (52,5) {\strut $\rho$};
\node at (0,38) {\strut $q$};
\node at (20,38) {\strut $\mathsf{BL}_1^{u_o} \cup \mathsf{BL}_2^{u_o}$};
\node at (15,24) {\strut $\mathsf{S}_1^{u_o}$};
\node at (47,30) {\strut $\mathsf{S}_2^{u_o}$};
\node at (21,5) {\strut $\mathsf{R}_2^{u_o}$};
\node at (41,18) {\strut $\mathsf{R}_1^{u_o}$};
\node at (31,22) {\strut $u_o$};
\end{scope}
\end{tikzpicture}
\caption{Forward (on the left) and backward (on the right) Lax curves through $u_o$.
The solid lines refer to rarefaction curves $\mathsf{R}_i^{u_o}$, the dashed lines to shock curves $\mathsf{S}_i^{u_o}$, $i\in\{1,2\}$, and the dotted lines to the sonic curves.}
\label{f:RVLaxcurves}
\end{figure}
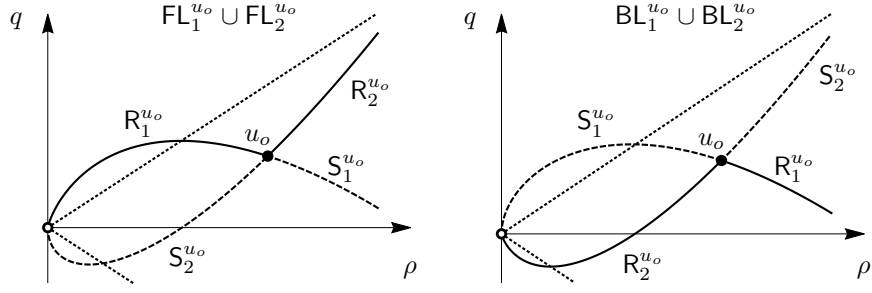

The following lemma collects the main properties of the above functions.

\begin{lem}
\label{l:JSatriani}
Let $u_o = (\rho_o,q_o), u^o = (\rho^o,q^o)\in\Omega$ be two distinct states and $i\in\{1,2\}$.
Then we have:
\begin{enumerate}[leftmargin=*, label={\bf(\alph*)}]\setlength{\itemsep}{0cm}\setlength\itemsep{0em}

\item
$\mathcal{R}_1^{(\rho_o, - q_o)} \equiv - \mathcal{R}_2^{(\rho_o,q_o)}$ and $\mathcal{S}_1^{(\rho_o, - q_o)} \equiv - \mathcal{S}_2^{(\rho_o,q_o)}$.

\item
$\mathcal{R}_i^{u_o}$, $\mathcal{S}_i^{u_o}$, $\mathcal{FL}_i^{u_o}$, $\mathcal{BL}_i^{u_o}$ are $\C2$ functions in $(0,+\infty)$ and
\begin{align*}
&\mathcal{R}_i^{u_o}(0^+) = \mathcal{S}_i^{u_o}(0^+) = 0,&
&\frac{\d\mathcal{R}_i^{u_o}}{\d\rho}(0^+) = \frac{\d\mathcal{S}_i^{u_o}}{\d\rho}(0^+) = (-1)^{i+1} \cdot \infty,
\\
&\mathcal{R}_i^{u_o}(\infty) = \mathcal{S}_i^{u_o}(\infty) = (-1)^i \cdot \infty,&
&\frac{\d\mathcal{R}_i^{u_o}}{\d\rho}(\rho_o) = \frac{\d\mathcal{S}_i^{u_o}}{\d\rho}(\rho_o) = \lambda_i(u_o).
\end{align*}

\item
$\mathcal{R}_1^{u_o}$, $\mathcal{S}_1^{u_o}$, $\mathcal{FL}_1^{u_o}$, $\mathcal{BL}_1^{u_o}$ are strictly concave, while $\mathcal{R}_2^{u_o}$, $\mathcal{S}_2^{u_o}$, $\mathcal{FL}_2^{u_o}$, $\mathcal{BL}_2^{u_o}$ are strictly convex.

\item
\label{i:d}
$\mathcal{S}_2^{u_o} = \mathcal{FL}_2^{u_o} < \mathcal{R}_2^{u_o} = \mathcal{BL}_2^{u_o} < \mathcal{R}_1^{u_o} = \mathcal{FL}_1^{u_o} < \mathcal{S}_1^{u_o} = \mathcal{BL}_1^{u_o}$ in $(0,\rho_o)$, and $\mathcal{S}_1^{u_o} = \mathcal{FL}_1^{u_o} < \mathcal{R}_1^{u_o} = \mathcal{BL}_1^{u_o} < \mathcal{R}_2^{u_o} = \mathcal{FL}_2^{u_o} < \mathcal{S}_2^{u_o} = \mathcal{BL}_2^{u_o}$ in $(\rho_o,+\infty)$.

\item
$\mathsf{R}_i^{u_o} \cap \mathsf{R}_i^{u^o} \neq \emptyset$ if and only if $\mathsf{R}_i^{u_o} = \mathsf{R}_i^{u^o}$, while if $u^o \in \mathsf{S}_i^{u_o}$ then $\mathsf{S}_i^{u_o} \cap \mathsf{S}_i^{u^o} = \{u_o,u^o\}$.
\end{enumerate}
\end{lem}

We introduce some notation; see \figurename~\ref{f:TheMarsVolta}.

\begin{defi}\label{d:notation} 
For $u_\ell,u_r\in\Omega$ we define:
\begin{itemize}[leftmargin=*]\setlength{\itemsep}{0cm}\setlength\itemsep{0em}
\item 
$\Bar{u}(u_\ell) \doteq (\Bar{\rho}(u_\ell),\Bar{q}(u_\ell))$ is the state on $\mathsf{FL}_1^{u_\ell}$ with the maximum $q$-coordinate;
\item 
$\tilde{u}(u_\ell,u_r) \doteq (\tilde{\rho}(u_\ell,u_r),\tilde{q}(u_\ell,u_r))$ is the (unique) state of $\mathsf{FL}_1^{u_\ell} \cap \mathsf{BL}_2^{u_r}$;
\item
$\hat{u}(q_o,u_\ell) \doteq (\hat{\rho}(q_o,u_\ell),q_o)$, for $q_o \leq \Bar{q}(u_\ell)$, is the state with the largest $\rho$-coordinate in the intersection of $\mathsf{FL}_1^{u_\ell}$ and $q=q_o$;
\item
$\check{u}(q_o,u_r) \doteq (\check{\rho}(q_o,u_r),q_o)$, for $q_o \geq 0$, is the state  with the largest $\rho$-coordinate in the intersection of $\mathsf{BL}_2^{u_r}$ and $q=q_o$;
\item
$u^{a}(u_r) \doteq (\rho^{a}(u_r),q^{a}(u_r))$ is the (unique) intersection of $\mathsf{BL}_{2}^{u_r}$ with the sonic line $q=a \, \rho$ in $\Omega$.
\end{itemize}
\end{defi}

We recall that $v\left(\Bar{u}(u_\ell)\right) > a$ if and only if $v(u_\ell) > a$ and $v\left(\Bar{u}(u_\ell)\right) = a$ if and only if $v(u_\ell) \leq a$; furthermore, $\Bar{u}(u_\ell) = u_\ell$ if and only if $v(u_\ell) = a$, see \cite[Lemma~2.6]{C2}.

\begin{figure}[!htb]\centering
\begin{tikzpicture}[every node/.style={anchor=south west,inner sep=0pt},x=1mm, y=1mm]
\node at (4,4) {\includegraphics[width=50mm]{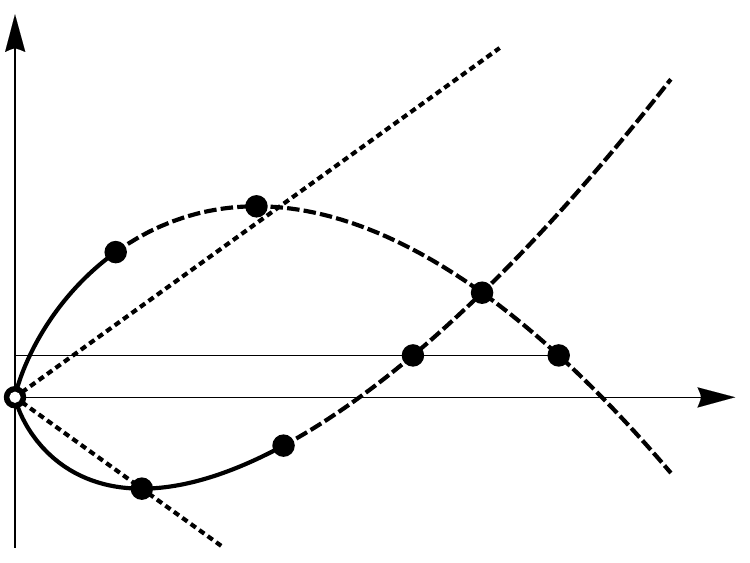}};
\node at (52,10) {\strut $\rho$};
\node at (0,38) {\strut $q$};
\node at (0,16) {\strut $q_o$};
\node at (48,36) {\strut $\mathsf{BL}_2^{u_r}$};
\node at (49,5) {\strut $\mathsf{FL}_1^{u_\ell}$};
\node at (43,16) {\strut $\hat{u}(q_o,u_\ell)$};
\draw[latex-] (37.5,22.2) - - + + (5,2) node[right] {\strut $\tilde{u}(u_\ell,u_r)$};
\node at (17.5,17.5) {\strut $\check{u}(q_o,u_r)$};
\node at (24,8) {\strut $u_r$};
\node at (15,28) {\strut $\Bar{u}(u_\ell)$};
\node at (7.5,24.5) {\strut $u_\ell$};

\begin{scope}[shift={(60,0)}]
\node at (4,4) {\includegraphics[width=50mm]{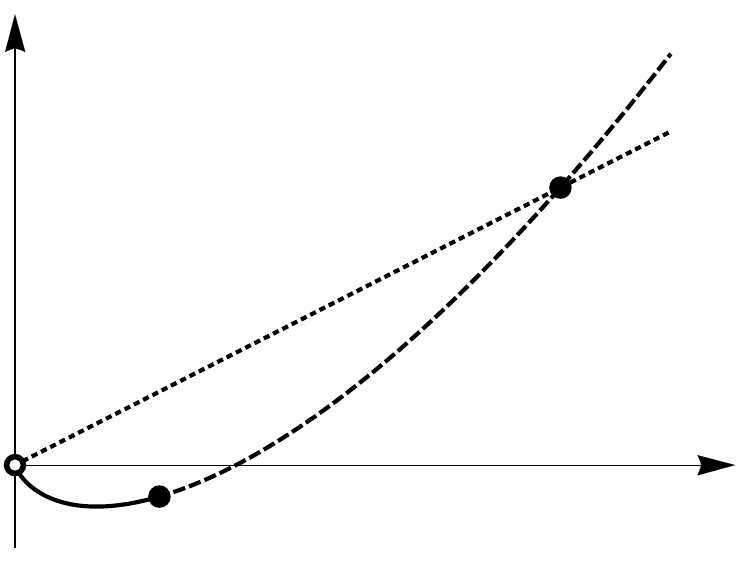}};
\node at (52,3) {\strut $\rho$};
\node at (0,38) {\strut $q$};
\node[above right, fill=white] at (45,37) {\strut $\mathsf{BL}_{2}^{u_r}$};
\node at (15,4) {\strut $u_r$};
\node at (42,25) {\strut $u^a(u_r)$};
\node[rotate=27] at (20,18) {\strut $q= a\rho$};
\end{scope}

\end{tikzpicture}
\caption{Notation in Definition~\ref{d:notation}. Rarefaction, shock and sonic curves are depicted as in \figurename~\ref{f:RVLaxcurves}.}
\label{f:TheMarsVolta}
\end{figure}
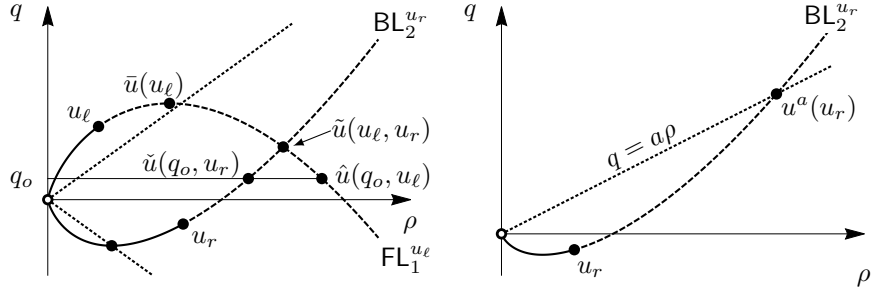

\subsection{The Riemann solver \texorpdfstring{$\rss$}{}}

For any pair of constant states $(u_\ell,u_r)\in\dom \doteq \Omega \times \Omega$, the corresponding Riemann problem is the Cauchy problem for system \eqref{e:system} with initial condition
\begin{equation}\label{e:iniRie}
u(0,x) = 
\begin{cases}
u_\ell &\hbox{if } x<0,
\\
u_r &\hbox{if } x\geq0.
\end{cases}
\end{equation}
Let $\rss \colon \dom \to\BV(\R;\Omega)$ be the {\em standard Riemann solver} of \eqref{e:system}, see \cite{Bressan-book}.
Denote then
\begin{align}
\label{e:BMTH}
&\us = (\rho_{\rm s},q_{\rm s}) \doteq \rss[u_\ell,u_r] \in \BV(\R;\Omega),
&\uspm = (\rspm,\qspm) \doteq \us(0^\pm) \in \Omega.
\end{align}
The function $(t,x) \mapsto \us(x/t)$  is the \emph{unique} entropy solution in $[0, +\infty)\times\R$ of the Riemann problem \eqref{e:system}, \eqref{e:iniRie} in the sense of Definition~\ref{d:entro}.
Since functions with bounded variations are defined up to a countable set of points, where the side limits are finite, we agree that $\xi\mapsto\us(\xi)$ is right continuous.
Observe that by \eqref{e:RH1} we have $\qsm = \qsp$ and therefore
\begin{equation}
\label{e:RHp}
\qs \doteq \qspm \in \R
\end{equation}
is well defined.

\subsubsection{Properties of the speed of the waves in the case \texorpdfstring{$\qs \geq 0$}{}}

As we mentioned in the Introduction, in the following we shall deal with one-way valves, where the flow through the valve is positive.
For this reason, and concerning the free flow, in Proposition~\ref{p:Korn} we require that $v(u_r) \ge 0$, and in Proposition~\ref{p:Loathe} that $v(u_\ell) \ge 0$. 

\begin{prop}\label{p:Korn}
Fix $u_\ell\in\Omega$.
A state $u_r\in\Omega$ is such that 
\begin{align*}
&\rss[u_\ell,u_r](0^-) = u_r,&
&v(u_{r}) \geq 0,
\end{align*}
if and only if $u_r \in \mathsf{FL}_1^{u_\ell}$ and one of the following conditions is met:
\begin{enumerate}[label={\bf (\Roman*)}, leftmargin=*]\setlength{\itemsep}{0cm}\setlength\itemsep{0em}
\item\label{i:Korn1}
$v(u_\ell) > a$ and either $u_r=u_\ell$ or $\rho_r \in (\hat{\rho}(q_\ell,u_\ell),\hat{\rho}(0,u_\ell)]$,
\item\label{i:Korn2}
$v(u_\ell) \leq a$ and $\rho_r \in [\Bar{\rho}(u_\ell),\hat{\rho}(0,u_\ell)]$.
\end{enumerate}
In both cases, $\rss[u_\ell,u_r]$ consists at most of one wave (a $1$-wave). If $u_r$ is supersonic, then also $u_\ell$ is supersonic and $u_r=u_\ell$.
\end{prop}

\begin{prop}\label{p:Loathe}
Fix $u_r\in\Omega$ and consider $u^a=u^a(u_r)$.
A state $u_\ell\in\Omega$ satisfies
\begin{align*}
&\rss[u_\ell,u_r](0^+) = u_\ell,&
&v(u_{\ell}) \in [0,a],
\end{align*}
if and only if one of the following conditions is met:
\begin{enumerate}[label={\bf (\Roman*)}, leftmargin=*]\setlength{\itemsep}{0cm}\setlength\itemsep{0em}
\item\label{i:Loathe1}
$\rho_{\ell} \in (0,\rho^{a}]$ and $q_\ell \in \mathcal{BL}_{2}^{u_r}(\rho_\ell)$,
\item\label{i:Loathe2}
$\rho_{\ell} > \rho^{a}$ and $q_\ell =a\rho_{\ell}$.
\end{enumerate}
In case~\ref{i:Loathe1}, $\rss[u_\ell,u_r]$ consists  at most of one wave (a $2$-wave); when~\ref{i:Loathe2} is satisfied, then $\rss[u_\ell,u_r]$ consists at most of two waves (a $1$-rarefaction and a $2$-wave).
If $u_\ell$ is subsonic, then it satisfies~\ref{i:Loathe1}.
\end{prop}

For $u_\ell \in \Omega$ we define, see \figurename~\ref{f:BHK0},
\begin{align}
\label{e:Gammaum}
\Gamma^-(u_\ell) &\doteq \left\{
u_r \in \mathsf{FL}_1^{u_\ell} : u_r \hbox{ satisfies~\ref{i:Korn1} or~\ref{i:Korn2} in Proposition~\ref{p:Korn}}
\right\},
\intertext{and for any $u_{r} \in \Omega$ we define, see \figurename~\ref{f:BHK1},}
\label{e:Gammaup}
\Gamma^+(u_r) &\doteq \left\{
u_\ell \in \Omega : u_\ell \hbox{ satisfies \ref{i:Loathe1} or \ref{i:Loathe2} in Proposition~\ref{p:Loathe}}
\right\}.
\end{align}
Proposition~\ref{p:Korn} coincides with \cite[Corollary~2.10]{CRR-MMAS-2025}; hence, its proof is omitted.
The conditions on $u_{\ell}$ in Proposition~\ref{p:Loathe} are obtained by augmenting those of \cite[Corollary~2.11]{CRR-MMAS-2025} with the additional requirement that $u_{\ell}$ is not supersonic.
Accordingly, items \ref{i:Loathe1} and \ref{i:Loathe2} follow from items {\bf (I)}-{\bf (III)} in \cite[Corollary~2.11]{CRR-MMAS-2025} upon further imposing the constraint $v(u_{\ell}) \leq a$. 
Graphically, $\Gamma^+(u_r)$ is obtained by intersecting the region depicted in \cite[\figurename~7]{CRR-MMAS-2025} with the half-plane $\{\,u \in \Omega : v(u) \leq a\,\}$.
The proof of Proposition~\ref{p:Loathe} is therefore omitted.

\begin{figure}[!htb]\centering
\begin{tikzpicture}[every node/.style={anchor=south west,inner sep=0pt},x=1mm, y=1mm]

\node at (4,4) {\includegraphics[width=50mm]{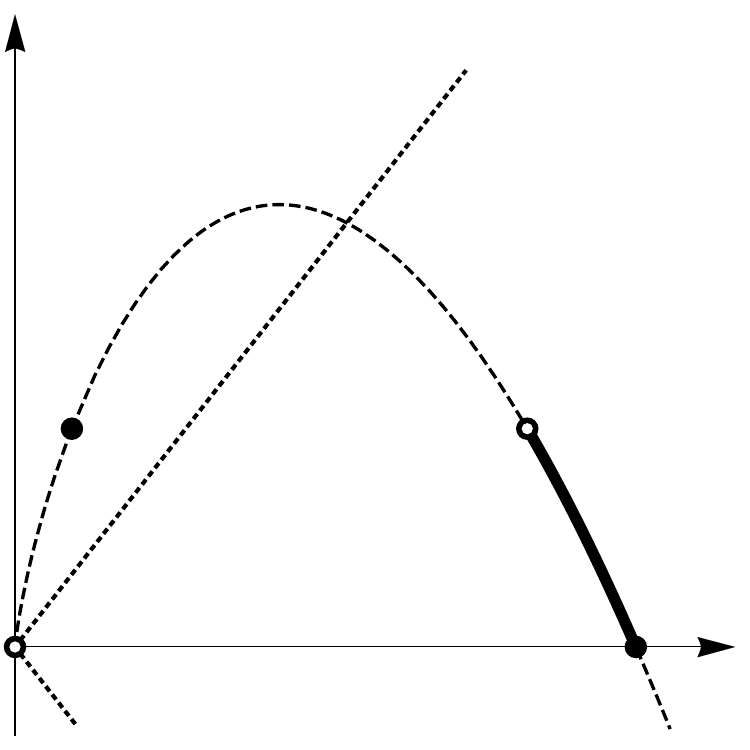}};
\node at (55,9) {\strut $\rho$};
\node at (0,50) {\strut $q$};
\node[right] at (50,4) {\strut $\mathsf{FL}_1^{u_\ell}$};
\node at (10,23) {\strut $u_\ell$};
\node at (41,23) {\strut $\hat{u}(q_\ell,u_\ell)$};
\node[right] at (37,50) {\strut $q=a\rho$};
\node at (34,5) {\strut $\hat{u}(0,u_\ell)$};

\begin{scope}[shift={(70,0)}]
\node at (4,4) {\includegraphics[width=50mm]{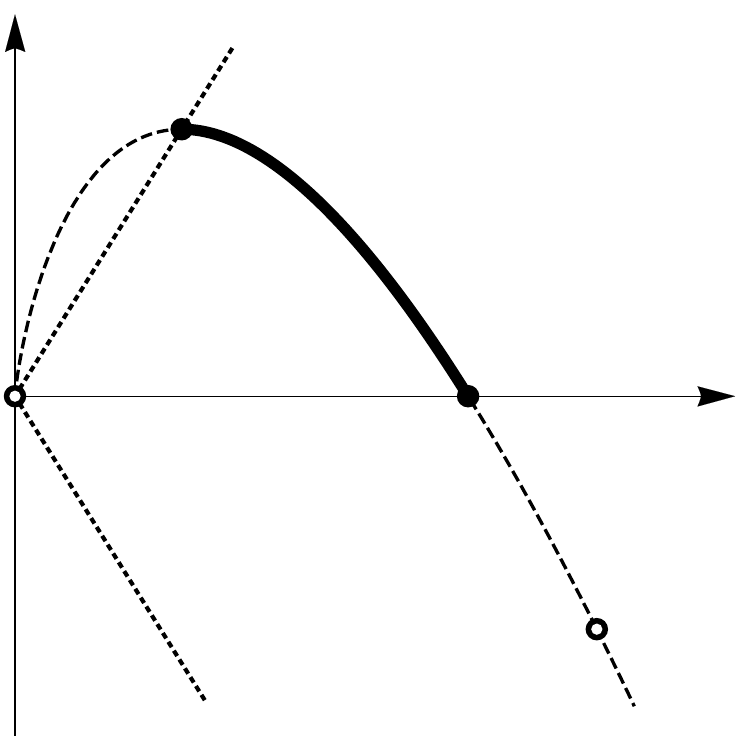}};
\node at (55,25) {\strut $\rho$};
\node at (0,50) {\strut $q$};
\node[right] at (48,4) {\strut $\mathsf{FL}_1^{u_\ell}$};
\node at (8,46) {\strut $\Bar{u}(u_\ell)$};
\node at (46,11) {\strut $u_\ell$};
\node at (21,50) {\strut $q=a\rho$};
\node at (23,23) {\strut $\hat{u}(0,u_\ell)$};
\end{scope}

\end{tikzpicture}
\caption{
The set $\Gamma^-(u_\ell)$ (solid lines and the full points) in \eqref{e:Gammaum}.
On the left, case~\ref{i:Korn1} in Proposition~\ref{p:Korn}; on the right, case~\ref{i:Korn2} in Proposition~\ref{p:Korn}. The state $\hat{u}(q_\ell,u_\ell)$ on the left does not belong to the solid line.
In the case under consideration on the right, the state $u_{\ell}$ does not belong to $\Gamma^-(u_\ell)$ because $v(u_{\ell})<0$.
}
\label{f:BHK0}
\end{figure}
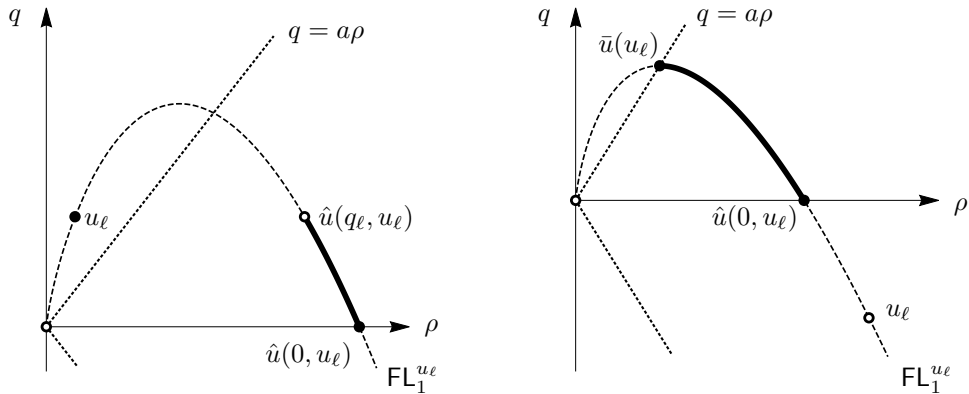

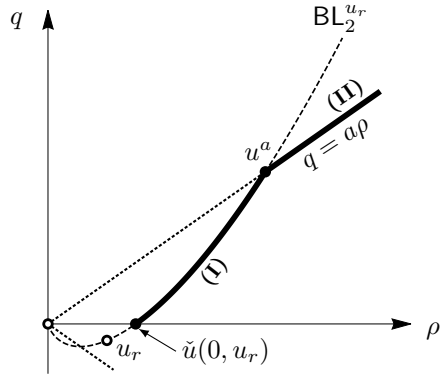
\begin{figure}[!htb]\centering
\begin{tikzpicture}[every node/.style={anchor=south west,inner sep=0pt},x=1mm, y=1mm]
\node at (4,4) {\includegraphics[width=50mm]{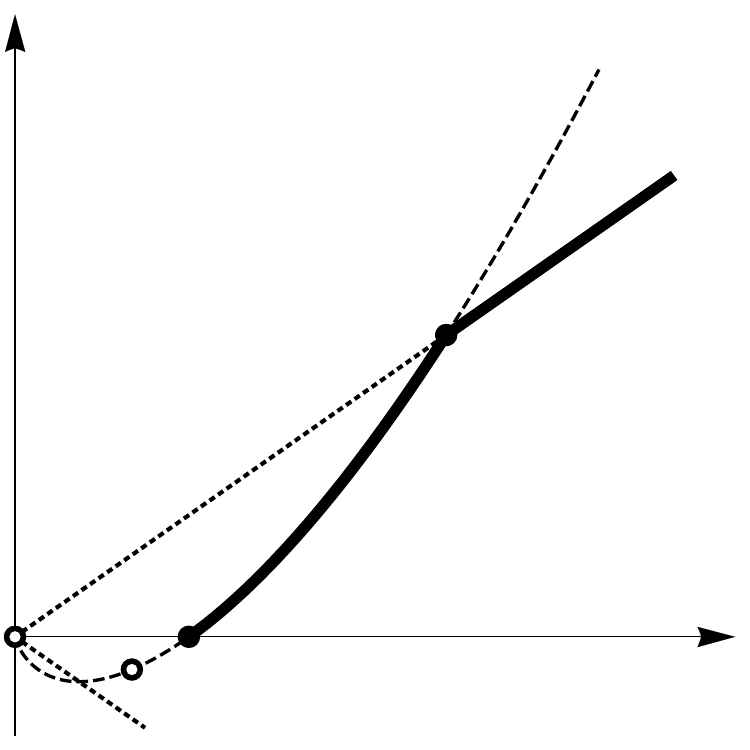}};
\node at (55,9) {\strut $\rho$};
\node at (0,50) {\strut $q$};
\node[right] at (40,52) {\strut $\mathsf{BL}_{2}^{u_r}$};
\node at (31,32) {\strut $u^a$};
\draw[latex-] (17,11) --++ (5,-3) node[right, inner sep=2pt] {\strut $\check{u}(0,u_r)$};
\node at (14,6) {\strut $u_r$};
\node[rotate=33] at (40,31) {\strut $q=a\rho$};

\node[rotate=35] at (43,39) {\footnotesize\bf (II)};
\node[rotate=52] at (27,16) {\footnotesize\bf (I)};
\end{tikzpicture}
\caption{The set $\Gamma^+(u_r)$ (solid lines and the full points) in \eqref{e:Gammaup}.
In the case under consideration, the state $u_{r}$ does not belong to $\Gamma^+(u_r)$ because $v(u_{r})<0$.}
\label{f:BHK1}
\end{figure}

\subsubsection{Demand function}

The \emph{demand} function $\overline{Q} \colon \Omega\to\R$ is defined as follows:
\begin{equation}
\label{e:SleepToken}
\overline{Q}(\rho,q) \doteq
\begin{cases}
q&\text{if } v(u) \geq a,
\\
\overline{q}(\rho,q) = \frac{a\,\rho}{e} \, \exp\left(\frac{v}{a}\right)&\text{if } v(u) < a,
\end{cases}
\end{equation}
where $e \approx 2.71828$ is Euler's number.

\begin{lem}
We have $\overline{Q}\in\C1(\Omega)$ and $0 < \overline{Q}(u) \in [q, \overline{q}(u)]$ for any $u \in \Omega$.
\end{lem}

At last we provide a motivation for the introduction of $\overline{Q}(u_\ell)$.

\begin{prop}
\label{p:Leprous}
Fix $u_\ell \in \Omega$. 
Then, for any $u_r\in\Omega$, the flow attained by $\uspm = \rss[u_\ell,u_r]$ at $\xi=0$ is bounded from above by $\overline{Q}(u_\ell)$.
More precisely, we have
\[
\overline{Q}(u_\ell)=\max\left\{ \qs : u_r\in\Omega \right\},
\]
where $\qs$ is the common value of the traces $\qspm$ of the flow at $\xi=0$, see \eqref{e:RHp}.
\end{prop}

\begin{rmk}
\label{r:Leprous}
Fix $u_\ell \in \Omega$. Proposition \ref{p:Leprous} shows that
$\overline{Q}(u_\ell)$ represents the {\em maximum achievable flow} at $x=0$ of a solution to the Riemann problem \eqref{e:system}, \eqref{e:iniRie} with $u_r \in \Omega$.
This means that, for any $u_r \in \Omega$, the flow at $x=0$ of the solution to the Riemann problem \eqref{e:system}, \eqref{e:iniRie} with initial state $(u_\ell,u_r)$ is less or equal to $\overline{Q}(u_\ell)$, i.e., $q_{\rm s}^0 \leq \overline{Q}(u_\ell)$, where $q_{\rm s}^0$ is defined in \eqref{e:RHp}.
Moreover, there exists $u_r \in \Omega$ such that $q_{\rm s}^0 = \overline{Q}(u_\ell)$.    
\end{rmk}

We observe that
\begin{equation}
\label{e:Periphery}
u \in \Gamma^-(u_\ell) \Longrightarrow q \leq \overline{Q}(u_\ell).
\end{equation}

\section{The flow through a valve}
\label{s:flowv}
\setcounter{equation}{0}

In this section we model the flow along a pipe-valve-pipe system;
for definiteness, the valve is placed at $x=0$. 
The valves under consideration are assumed to be either closed or \lq\lq fully\rq\rq\ open as time passes;
their opening or closing is encoded by the function $h \colon [0, +\infty) \to \{0,1\}$, so that
\[h(t) \doteq \begin{cases}
1&\hbox{if the valve is open at time }t,
\\
0&\hbox{if the valve is closed at time }t.
\end{cases}\]
We first generalize Definition~\ref{d:entro} of entropy solution to a Cauchy problem to the case a valve is present at $x=0$; then we give the definition of c-solution (\lq\lq c\rq\rq\ stands for coupling).
Subsequently, we consider the Riemann problem and generalize the Riemann solver $\rss$ to the case a valve is present at $x=0$ and give the definition of c-Riemann solver $\rsc$.

\subsection{The Cauchy problem}
We now make explicit what we mean by a solution of a Cauchy problem. In the general Definition \ref{d:csol} below, the function $h$ plays a different role according to the model under consideration. On the one hand, $h$ can be barely computed from the solution $u$, which then does not depend on it; this is the case of the non-hysteretic valves examined in \cite{C3, C1, C2}. On the other hand, the determination of a solution $u$ may depend on the initial value $h_o$ of $h$; this is the case of the hysteretic valves considered below. The same remark holds for the Riemann problem in Subsection~\ref{ss:RP}.

\begin{defi}
\label{d:csol}
For a given initial datum $(u_o,h_o) \in \BV_{\rm loc}(\R;\Omega) \times \{0,1\}$, a function $(u,h)\in\BV_{\rm loc}([0, +\infty) \times \R;\Omega) \times \L\infty([0, +\infty);\{0,1\})$ is a global c-solution of the Cauchy problem
\begin{equation}
\label{e:csystem}
\left\{\begin{array}{@{}l@{}}
\rho_t + q_x = 0, \\
q_t + P_x = 0, \\
\left(u(0,x),h(0)\right) = (u_o,h_o),
\end{array}\right.
\end{equation}
if away from $x=0$ the function $u$ is an entropy solution in the sense of Definition~\ref{d:entro}, and at $x=0$ it satisfies the following coupling conditions for a.e.\ $t>0$:
\begin{gather}
\label{e:csol-q}
q(t,0^-) = q(t,0^+) = q^0(t),
\\
\label{e:csol-h}
h(t) = 0 \Longrightarrow q^0(t) = 0.
\end{gather}
\end{defi}

\begin{rmk}
Some comments on Definition~\ref{d:csol} are in order.
\begin{itemize}[leftmargin=*]
\item 
Condition \eqref{e:csol-q} ensures that a c-solution satisfies the first Rankine-Hugoniot condition \eqref{e:RH1} at $x=0$, which expresses the conservation of gas through the valve.  
Condition \eqref{e:csol-h} ensures that if the valve is closed, then no flow occurs through $x=0$.

\item 
Away from $x=0$, the $u$-component of a c-solution is an entropy solution of the Cauchy problem \eqref{e:csystem} in the sense of Definition~\ref{d:entro}; namely, it satisfies condition \eqref{e:wcond} for any test function $\varphi \in \Cc\infty([0, +\infty) \times (\R\setminus\{0\}); \R)$, and condition \eqref{e:econd} for any entropy pair $(\mathcal{E},\mathcal{Q})$ and test function $\varphi \in \Cc\infty((0, +\infty) \times (\R\setminus\{0\}); [0, +\infty))$.
In particular, along the two pipes $\pm x>0$ the gas flow is described by \eqref{e:system}.
Across $x=0$, non-entropic stationary waves may arise; we refer to such waves as c-waves.
\end{itemize}
\end{rmk}

\subsection{The Riemann problem}
\label{ss:RP}
In this section we introduce the general definition of the one-way c-Riemann solver with non-supersonic outflow considered in this paper.
In fact, the simplest way to model the gas flow through a valve is to describe the associated c-Riemann solver.

\begin{defi}\label{d:RSV} 
A function
\begin{equation}
\label{e:C}
\begin{array}{cccc@{\;}l}
C \colon& \dom \times \{0,1\}& \to &\dom \times \{0,1\}
\\
&(u_\ell,u_r,h_o)&\mapsto& \left(\ucm(u_\ell,h_o), \ucp(u_\ell,u_r,h_o), \hc(u_\ell,h_o) \right)
\end{array}
\end{equation}
is a \emph{c-function} if there exist $\rcm\colon \Omega \times \{0,1\} \to (0, +\infty)$, $\rcp \colon \dom \times \{0,1\} \to (0, +\infty)$ and $\qc \colon \Omega \times \{0,1\} \to [0, +\infty)$ such that $\ucm =(\rcm ,\qc )$ and $\ucp = (\rcp ,\qc )$ satisfy, for any $(u_\ell,u_r,h_o)\in\dom\times\{0,1\}$,
\begin{align}
\label{e:RSV1}
&\ucm(u_\ell,h_o) \in \Gamma^-(u_\ell),&
&\qc(u_\ell,h_o) \in \left[ 0 , \overline{Q}\left(u_\ell\right) \right],
\\
\label{e:RSV2}
&\ucp(u_\ell,u_r,h_o) \in \Gamma^+(u_r),&
&\hc(u_\ell,h_o) = 0 \Longrightarrow
\qc(u_\ell,h_o) = 0.
\end{align}
The corresponding \emph{c-Riemann solver} $\rsc \colon \dom \times \{0,1\} \to \BVloc(\R;\Omega) \times \{0,1\}$ is defined as
\begin{align*}
\rsc[u_\ell,u_r,h_o](\xi) &\doteq
\begin{cases}
\left(\rss\left[u_\ell,\ucm \right](\xi),\hc\right) &\hbox{if }\xi<0,\\[5pt]
\left(\rss\left[\ucp ,u_r\right](\xi),\hc\right) &\hbox{if }\xi\geq 0,
\end{cases}&
(\ucm ,\ucp,\hc) &= C(u_\ell,u_r,h_o).
\end{align*}
\end{defi}

In analogy with \eqref{e:BMTH}, we denote
\begin{align*}
&(\uc,\hc) = (\rho_{\rm c},q_{\rm c},\hc) \doteq \rsc[u_\ell,u_r,h_o],
&(\ucpm,\hc) = (\rcpm,\qc,\hc) \doteq (\uc(0^\pm),\hc).
\end{align*}

\begin{rmk}
We now comment on Definition~\ref{d:RSV}.

\begin{enumerate}[label={(\roman*)},leftmargin=*,nolistsep]\setlength{\itemsep}{0cm}\setlength\itemsep{0em}

\item 
The map $(t,x) \mapsto \uc(x/t)$ satisfies the first Rankine-Hugoniot condition \eqref{e:RH1} in the whole of $\R$, the second Rankine-Hugoniot condition \eqref{e:RH2} in $\R\setminus\{0\}$, but it does not necessarily satisfy \eqref{e:RH2} at $x=0$.
Hence, mass is conserved, whereas momentum may fail to be conserved through $x=0$.
As a consequence, $\uc$ \emph{may fail} to be a weak solution of \eqref{e:system} in the whole of $\R$ in the sense of Definition~\ref{d:weak}.
On the other hand, $(\uc,\hc)$ is a c-solution in the sense of Definition~\ref{d:csol}.

\item 
A c-Riemann solver $\rsc$ is uniquely characterized by the function $C$, which associates to any initial value $(u_\ell,u_r,h_o)$ the traces $\ucpm$ of $\uc$ at $\xi=0$ and the corresponding value of $\hc$. 
For brevity we omitted the dependence of $\rsc$ on $C$.

\item
By the definitions \eqref{e:Gammaum}, \eqref{e:Gammaup} of $\Gamma^-(u_\ell)$ and $\Gamma^+(u_r)$, conditions \eqref{e:RSV1}\textsubscript{$1$} and \eqref{e:RSV2}\textsubscript{$1$} imply that for any $(u_\ell,u_r,h_o) \in \dom \times \{0,1\}$ the following conditions are satisfied by $\ucm=\ucm(u_{\ell},h_o)$ and $\ucp=\ucp(u_{\ell},u_{r},h_o)$:
\begin{enumerate}[label={\bf\alph*)}]\setlength{\itemsep}{0cm}\setlength\itemsep{0em}
\item
$\rss\left[u_\ell,\ucm\right](0^-) = \ucm$, hence $\rss\left[u_\ell,\ucm\right]$ has no waves with positive speed;
\item
$\rss\left[\ucp,u_r\right](0^+) = \ucp$, hence $\rss\left[\ucp,u_r\right]$ has no waves with negative speed.
\end{enumerate}
\item
Again by \eqref{e:Gammaum}, \eqref{e:Gammaup}, for any $\qc \in \left[ 0 , \overline{Q}\left(u_\ell\right) \right]$ we have that $\ucm=\ucm(u_{\ell},h_o)$ and $\ucp=\ucp(u_{\ell},u_{r},h_o)$ are given by:
\begin{align}
\label{e:ucpm}
\ucm &= \begin{cases}
u_{\ell}&\hbox{if } v(u_{\ell})>a \hbox{ and } \qc=q_{\ell},
\\
\hat{u}(\qc,u_{\ell})&\hbox{otherwise},
\end{cases}&
\ucp &= \begin{cases}
\check{u}(\qc,u_r)&\hbox{if } \qc \in [0,q^{a}(u_{r})],
\\
(\qc/a,\qc)&\hbox{if } \qc > q^{a}(u_{r}).
\end{cases}
\end{align}
As a consequence, $v(\ucm)>a$ if and only if $v(u_{\ell})>a$ and $\qc=q_{\ell}$, whereas $v(\ucp) \leq a$ holds in any case.
\end{enumerate}
\end{rmk}

We now give the definition of a coherent c-Riemann solver.

\begin{defi}
\label{d:coherent}
A c-Riemann solver $\rsc \colon \dom\times\{0,1\} \to \BV(\R;\Omega)\times\{0,1\}$ is \emph{coherent} at $(u_\ell,u_r,h_o)\in\dom\times\{0,1\}$ if the traces $(\ucpm,\hc) \doteq \rsc[u_\ell,u_r,h_o](0^\pm)$ satisfy
\begin{align}
(\ucm,\ucp) \in \dom
\qquad 
\hbox{ and }
\qquad
C(\ucm,\ucp,\hc) = (\ucm,\ucp,\hc).
\label{e:coherent}
\end{align}
The \emph{coherence domain} $\mathsf{CH}$ of $\rsc$ is the set of all triples $(u_\ell,u_r,h_o)\in\dom\times\{0,1\}$ where $\rsc$ is coherent. The set $\mathsf{CH}^{\scriptscriptstyle\complement}\doteq(\dom\times\{0,1\}) \setminus \mathsf{CH}$ is the \emph{incoherence domain}.
\end{defi}

Notice that a c-Riemann solver $\rsc$ is coherent at an initial datum $(u_\ell, u_r,h_o) \in \dom\times\{0,1\}$ if the ordered pair of the traces of the corresponding solution $(\ucpm,\hc)$ are such that $(\ucm,\ucp) \in \dom$ and $(\ucm,\ucp,\hc)$ is a fixed point of $C$; in this case
\[\rsc[\ucm,\ucp,\hc](\xi) = 
\begin{cases}
(\ucm,\hc)&\hbox{if } \xi<0,\\
(\ucp,\hc)&\hbox{if } \xi\geq0.
\end{cases}
\]
Hence, coherence may be thought as a stability property. On the contrary, for instance if the coupling represents a valve, the incoherence of a c-Riemann solver is understood as modeling the chattering and may yield analytical and numerical instabilities; see \cite{C1,C2,DGRR} and references there.

We extend the notion of coherence to solutions of the Cauchy problem as follows.  
A c-solution $(u,h)$ of \eqref{e:csystem} in the sense of Definition~\ref{d:csol} is said to be coherent if
\begin{equation}
\label{e:coCau}
\left(u(t,0^-),u(t,0^+)\right) \in \dom
\qquad 
\hbox{ and }
\qquad
C\left(u(t,0^-),u(t,0^+),h(t)\right) = \left(u(t,0^-),u(t,0^+),h(t)\right) \qquad \text{for a.e.\ } t>0.
\end{equation}

\section{A pressure-independent valve with relay}
\label{s:valve-relay}

In this section we first describe the operating behavior of the valve under consideration and then give the corresponding c-Riemann solver.

\subsection{Description of the valve}
\label{s:descr}
\setcounter{equation}{0}

In this section we model the dynamic behavior of a \emph{pressure-independent one-way valve with relay}. 
The valve is designed to ensure a subsonic outgoing flow and to maintain a fixed outgoing flow $q_*$, independently from the pressure; if it cannot, then it closes. Since the valve is one-way, the outgoing flow has a fixed sign: for definiteness we assume $q_*>0$. 
This behavior is made more realistic by a relay-type behavior, namely, by the presence of two thresholds $q_*^\pm$ satisfying
\[
0 < q_*^-<q_*<q_*^+.
\]

Roughly speaking, if the maximum achievable {\em free} flow at $x=0$ is strictly less than $q_*^-$, then the valve is closed and the flow at $x=0$ is $0$; if it is strictly larger than $q_*^+$, then the valve is open and the flow at $x=0$ is $q_*$. On the other hand, if the maximum achievable free flow at $x=0$ is in the intermediate region $[q_*^-,q_*^+]$, then the flow at $x=0$ is in $\{0\} \cup [q_*^-,q_*]$ and depends on the past history of the flow and on $h_{o}$.
The last behavior introduces an hysteresis process in the modeling.

Observe that $q_{*}^{+}$ and $q_{*}^{-}$ do not play symmetric roles; this is because we deal with a one-way valve and admit outflows less but not greater than $q_*$.

\begin{rmk}
An analogous valve is considered in \cite[Definition~4.2]{C2}. 
The key difference is that in this paper we introduce thresholds, $q_*^\pm$, along with hysteresis processes.  
Let us also recall that the c-Riemann solver in \cite[Definition~4.2]{C2} was modified in \cite[Section~3]{C1} to obtain a coherent c-Riemann solver.  
\end{rmk}

We now rigorously describe the operating behavior of our valve by steps. 
Let $u(t,x)=\left(\rho(t,x),q(t,x)\right) \in \Omega$ be the state of the gas in the pipe-valve-pipe system at time $t \geq 0$ and position $x \in \R$, and $h(t)$ be the operating mode of the valve at time $t \geq 0$.
We assume that
\begin{align}\label{e:uhreg}
&u\in\BV_{\rm loc}([0, +\infty) \times \R;\Omega),&
&h\in\L\infty([0, +\infty);\{0,1\}),
\end{align}
so that the spatial traces $u(t,x^\pm)$ at fixed $x\in\R$ are well defined for a.e.\ $t\ge0$; the time traces $u(t^\pm,x)$ at fixed $t>0$ exist for a.e.\ $x\in\R$, while if $t=0$ then the same occurs for the right trace $u(0^+,x)$. 
In particular, the spatial traces at $x=0$, 
\begin{align*}
&u_\ell^0(t) = \left(\rho_\ell^0(t),q_\ell^0(t)\right) \doteq u(t,0^-),&
&u_r^0(t) = \left(\rho_r^0(t),q_r^0(t)\right) \doteq u(t,0^+),
\end{align*}
are well defined for a.e.\ $t \geq 0$ and belong to $\Omega$. 
Moreover, by imposing the first Rankine-Hugoniot condition \eqref{e:RH1}, we have $q_\ell^0(t)=q_r^0(t)$. As a consequence, the flow at $x=0$
\[
q^0(t) \doteq q(t,0)
\]
is well defined for a.e.\ $t\ge0$. 
Since the flow at $x=0$ is in one direction only, then $q^0(t)\ge0$; indeed, we shall prove that $q^0(t) \in \{0\} \cup [q_*^-,q_*]$. Anyway, inside the pipes the flow may well be negative.

\smallskip

In order to give an idea of how the valve works, we assume for the moment
\begin{equation}\label{e:VOLA}
t \mapsto u_\ell^0(t) \hbox{ is {\em continuous}.}
\end{equation}
Then also $t \mapsto \overline{Q}(u_\ell^0(t))$ is continuous, see \eqref{e:SleepToken}. We recall that, for a fixed $t\geq0$, $\overline{Q}(u_\ell^0(t))$ represents the {\em maximum achievable free flow} at $x=0$ of a solution to the \lq\lq free\rq\rq\ (i.e., without the presence of the valve at $x=0$) Riemann problem with initial state $\left(u_\ell,u_r\right)$, for $u_\ell = u_\ell^0(t)$ and $u_r \in \Omega$, see Proposition~\ref{p:Leprous} and Remark~\ref{r:Leprous}.

For fixed $t\geq0$, we define the sets, see \figurename~\ref{f:Xtpm},
\begin{align}
\mathcal{T}_t^-\left[\overline{Q}\left(u_\ell^0\right)\right] &\doteq \left\{\tau\in[0,t] : \overline{Q}\left(u_\ell^0(\tau)\right) < q_*^-\right\},&
\mathcal{T}_t^+\left[\overline{Q}\left(u_\ell^0\right)\right] &\doteq \left\{\tau\in[0,t] : \overline{Q}\left(u_\ell^0(\tau)\right) > q_*^+ \right\}.
\label{e:TT}
\end{align}
\begin{figure}[!htb]
\centering
\begin{tikzpicture}[x=10mm, y=7mm]
\coordinate (A) at (0,1);
\coordinate (B) at (1,.8);
\coordinate (C) at (2,2);
\coordinate (D) at (2.7,3);
\coordinate (E) at (4.2,4);
\coordinate (F) at (5,3);
\coordinate (G) at (6,2);
\coordinate (H) at (7,.5);
\coordinate (I) at (8,1);
\draw[-latex] (0,0) -- (8.1,0) node[below] {\strut$t$}; 
\draw[-latex] (0,0) -- (0,4.1) node[left] {$\overline{Q}\left(u_\ell^0\right)$};
\draw[dashed] (0,2) node[left] {\strut$q_*^-$} --(8,2);
\draw[dashed] (0,2.5) node[left] {\strut$q_*$} --(8,2.5);
\draw[dashed] (0,3) node[left] {\strut$q_*^+$} --(8,3);
\draw[thick] plot [smooth, tension=.6] coordinates { (A) (B) (C) (D) (E) (F) (G) (H) (I) };

\filldraw [fill=black, thick] (0,0) circle (2pt);

\draw[dotted] (2,0)--(C);
\filldraw [fill=black, thick] (C) circle (2pt);
\draw[very thick] (0,0)--(2,0) node[midway,below]{\strut$\mathcal{T}_{t_0}^-\left[\overline{Q}\left(u_\ell^0\right)\right]$};
\filldraw [fill=white, thick] (2,0) circle (2pt);

\draw[dotted] (2.7,0)--(D);
\filldraw [fill=black, thick] (D) circle (2pt);
\draw[dotted] (5,0)--(F);
\filldraw [fill=black, thick] (F) circle (2pt);
\draw[very thick] (2.7,0)--(5,0) node[midway,below]{\strut$\mathcal{T}_{t_0}^+\left[\overline{Q}\left(u_\ell^0\right)\right]$};
\filldraw [fill=white, thick] (2.7,0) circle (2pt);
\filldraw [fill=white, thick] (5,0) circle (2pt);

\draw[dotted] (5.7,0) node[below]{\strut$t_0$} --(5.7,4);

\end{tikzpicture}
\caption{The sets $\mathcal{T}_{t_0}^\pm\left[\overline{Q}\left(u_\ell^0\right)\right]$ corresponding to an hypothetical maximum achievable free flow $\overline{Q}\left(u_\ell^0(t)\right)$, in the case \eqref{e:VOLA}, for some $t_0>0$.}
\label{f:Xtpm}
\end{figure}
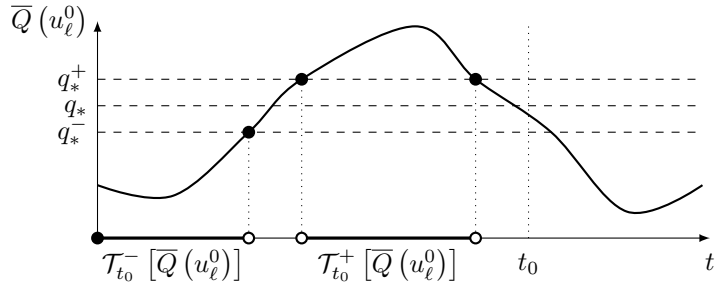
We used square brackets, the usual operator notation, to emphasize that $\mathcal{T}^\pm$ are operators mapping the function $\overline{Q}\left(u_\ell^0\right)$ to the $t$-dependent time sets $\mathcal{T}_t^\pm\left[\overline{Q}\left(u_\ell^0\right)\right]$.
The sets $\mathcal{T}_t^\pm\left[\overline{Q}\left(u_\ell^0\right)\right]$ do not solely depend on the value of $u_\ell^0$ at time $t$, but they keep trace of its past history in the interval $[0, t]$. Under \eqref{e:VOLA} the sets $\mathcal{T}_t^\pm\left[\overline{Q}\left(u_\ell^0\right)\right]$ are separated, i.e., $\overline{\mathcal{T}_t^-\left[\overline{Q}\left(u_\ell^0\right)\right]}\cap \overline{\mathcal{T}_t^+\left[\overline{Q}\left(u_\ell^0\right)\right]}=\emptyset$.

The operating mode of the valve now follows, for any $t_0\geq0$. Notice that in the cases below we emphasize the consequences of the assumptions characterizing each of them by involving the function $h$, which shall be made explicit only in \eqref{e:h=H}.

\begin{enumerate}[label={\alph*)}]
\item 
\label{i:v-i}
$\overline{Q}\left(u_\ell^0(t_0)\right) \in (0,q_*^-)$. The valve is closed for small times $t \geq t_0$ and so 
\[
h(t)=0, \qquad q^0(t) = 0.
\]

\item
\label{i:v-ii} 
$\overline{Q}\left(u_\ell^0(t_0)\right) \in (q_*^+, +\infty)$. The valve is open for small times $t \geq t_0$ and
\[
h(t)=1, \qquad q^0(t) = q_*.
\]

\item
\label{i:v-iii} 
$\overline{Q}\left(u_\ell^0(t_0)\right) \in [q_*^-,q_*^+]$ and $\sup \mathcal{T}_{t_0}^-\left[\overline{Q}\left(u_\ell^0\right)\right] < \sup \mathcal{T}_{t_0}^+\left[\overline{Q}\left(u_\ell^0\right)\right] \doteq \tau^+$ for some $\tau^+\in[0,t_0]$, where we mean $\sup\emptyset=-\infty$. The valve is open for small times $t \geq t_0$; for such times we have
\[
h(t)=1, \qquad  q^0(t) = \min\left\{q_*,\overline{Q}\left(u_\ell^0(t)\right) \right\}.
\] 
Observe that $h(s) = 1$ and $q^0(s) = q_*$ for large enough $s<\tau^+$, by item \ref{i:v-ii}.
Moreover, for any $t\in (\tau^+,t_0]$ we still have $\overline{Q}\left(u_\ell^0(t)\right) \in [q_*^-,q_*^+]$, $h(t) = 1$ and $q^0(t) = \min\left\{q_*,\overline{Q}\left(u_\ell^0(t)\right) \right\}$, see \figurename~\ref{f:Vola}.
\begin{figure}[!htb]
\centering
\begin{tikzpicture}[x=10mm, y=7mm]
\coordinate (A) at (0,1);
\coordinate (B) at (1,.8);
\coordinate (C) at (2,2);
\coordinate (D) at (2.7,3);
\coordinate (E) at (4.2,4);
\coordinate (F) at (4.7,3.1);
\coordinate (G) at (5,3);
\coordinate (H) at (6,3);
\coordinate (I) at (6.3,2.9);
\coordinate (J) at (6.8,2.1);
\coordinate (K) at (7,2);
\coordinate (L) at (8,2);
\coordinate (M) at (8.8,2.1);
\coordinate (N) at (9.5,2.8);
\coordinate (O) at (10,3.5);
\draw[-latex] (0,0) -- (10.1,0) node[below] {\strut$t$}; 
\draw[-latex] (0,0) -- (0,4.1) node[left] {$\overline{Q}\left(u_\ell^0\right)$};
\draw[dashed] (0,2) node[left] {\strut$q_*^-$} --(10,2);
\draw[dashed] (0,2.5) node[left] {\strut$q_*$} --(10,2.5);
\draw[dashed] (0,3) node[left] {\strut$q_*^+$} --(10,3);
\draw[thick] plot [smooth, tension=.6] coordinates { (A) (B) (C) (D) (E) (F) (G) };
\draw[thick] (G) -- (H);
\draw[thick] plot [smooth, tension=.6] coordinates { (H) (I) (J) (K) };
\draw[thick] (K) -- (L);
\draw[thick] plot [smooth, tension=.6] coordinates { (L) (M) (N) (O) };

\filldraw [fill=black, thick] (0,0) circle (2pt);

\draw[dotted] (2,0)--(C);
\filldraw [fill=black, thick] (C) circle (2pt);
\draw[very thick] (0,0)--(2,0) node[midway,below]{\strut$\mathcal{T}_{t_0}^-\left[\overline{Q}\left(u_\ell^0\right)\right]$};
\filldraw [fill=white, thick] (2,0) circle (2pt);

\draw[dotted] (2.7,0)--(D);
\filldraw [fill=black, thick] (D) circle (2pt);
\draw[dotted] (5,0)--(G);
\filldraw [fill=black, thick] (G) circle (2pt);
\draw[very thick] (2.7,0)--(5,0) node[midway,below]{\strut$\mathcal{T}_{t_0}^+\left[\overline{Q}\left(u_\ell^0\right)\right]$};
\node[below] at (5,0) {\strut$\tau^{+}$};
\filldraw [fill=white, thick] (2.7,0) circle (2pt);
\filldraw [fill=white, thick] (5,0) circle (2pt);

\draw[dotted] (9,0) node[below]{\strut$t_0$} --(9,4);

\end{tikzpicture}
\caption{The sets $\mathcal{T}_{t_0}^\pm\left[\overline{Q}\left(u_\ell^0\right)\right]$ corresponding to an hypothetical maximum achievable free flow $\overline{Q}\left(u_\ell^0(t)\right)$, in the case \eqref{e:VOLA}, for some $t_0>0$, see item \ref{i:v-iii}.}
\label{f:Vola}
\end{figure}
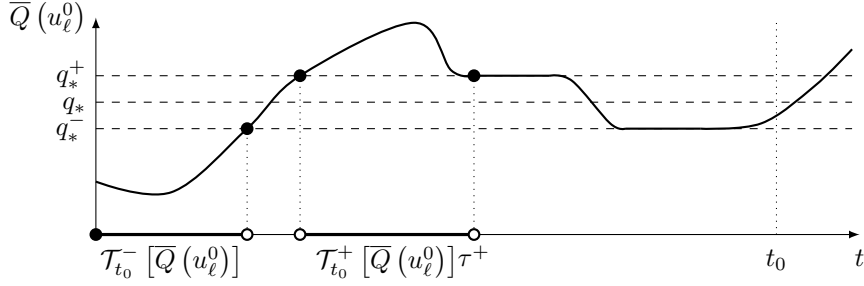

\item
\label{i:v-iv} 
$\overline{Q}\left(u_\ell^0(t_0)\right) \in [q_*^-,q_*^+]$ and $\sup \mathcal{T}_{t_0}^-\left[\overline{Q}\left(u_\ell^0\right)\right] > \sup \mathcal{T}_{t_0}^+\left[\overline{Q}\left(u_\ell^0\right)\right]$. In this case the valve is closed for small times $t \geq t_0$ and so 
\[
h(t)=0, \qquad q^0(t) = 0.
\] 
This case can be discussed as the previous one. 

\item 
\label{i:v-v}
$\overline{Q}\left(u_\ell^0(t_0)\right) \in [q_*^-,q_*^+]$ and $\mathcal{T}_{t_0}^-\left[\overline{Q}\left(u_\ell^0\right)\right] = \mathcal{T}_{t_0}^+\left[\overline{Q}\left(u_\ell^0\right)\right] = \emptyset$.
In this case, for small enough $t > t_0$, we have
\[
h(t) = h_o,\qquad  q^0(t) = h_o \, \min\left\{q_*,\overline{Q}\left(u_\ell^0(t)\right) \right\}.
\]
\end{enumerate}

The opening and closing of the valve described above in items~\ref{i:v-i}-\ref{i:v-v} can be resumed by following the approach in \cite{vis}. We define the set
\begin{equation}\label{e:L}
\mathcal{L} \doteq\left([0,q_*^+]\times\{0\}\right)\cup\left([q_*^-, +\infty)\times\{1\}\right)\subseteq\R^2,
\end{equation}
see \figurename~\ref{f:LacunaCoil0}.
\begin{figure}[!htb]
\centering
\begin{tikzpicture}[x=15mm, y=15mm]
\draw[-latex] (0,0) -- (4,0) node[below] {\strut$q$}; 
\draw[-latex] (0,0) -- (0,1.5) node[left] {$h$};
\draw[dashed] (1,0) node[below] {\strut$q_*^-$} --(1,1);
\draw[dashed] (2,0) node[below] {\strut$q_*$}   --(2,1);
\draw[dashed] (3,0) node[below] {\strut$q_*^+$} --(3,1);

\draw[->] (.9,.9)--(.9,.1);
\draw[<-] (3.1,.9)--(3.1,.1);

\draw[ultra thick] (0,0)--(3,0);
\draw[ultra thick] (1,1)--(4,1);
\draw[dashed] (0,1) node[left] {\strut$1$} -- (1,1);
\filldraw [fill=black, ultra thick] (3,0) circle (2pt);
\filldraw [fill=black, ultra thick] (1,1) circle (2pt);
\filldraw [fill=black, ultra thick] (0,0) circle (2pt);
\end{tikzpicture}
\caption{The set $\mathcal{L}$ in \eqref{e:L}. If the maximum achievable free flow $\overline{Q}\left(u_\ell^0\right)$ increases from $0$, then the valve remains closed until $\overline{Q}\left(u_\ell^0\right)$ reaches the value $q_*^+$, when the valve opens. Analogously, if $\overline{Q}\left(u_\ell^0\right)$ decreases from values larger than $q_*^+$, then the valve remains open until the threshold $q_*^-$ is reached, when it closes.}
\label{f:LacunaCoil0}
\end{figure}
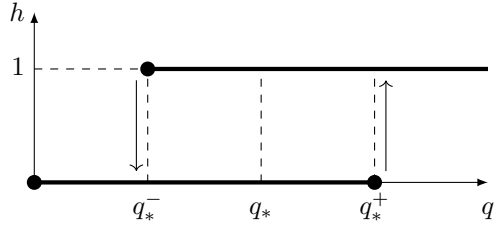
Assume $\left(\overline Q(u_\ell^0(0)),h_o\right)\in{\cal L}$. 
For $t>0$ we define
\begin{equation}
\label{e:hrelay}
H\left[\overline{Q}\left(u_\ell^0\right),h_o\right](t)\doteq
\left\{
\begin{array}{@{}l@{\quad}l@{}}
\displaystyle
h_o&\mbox{ if } \mathcal{T}_t^-\left[\overline{Q}\left(u_\ell^0\right)\right] = \emptyset = \mathcal{T}_t^+\left[\overline{Q}\left(u_\ell^0\right)\right],
\\[7pt]
\displaystyle
1&\mbox{ if } \sup \mathcal{T}_t^-\left[\overline{Q}\left(u_\ell^0\right)\right]<\sup \mathcal{T}_t^+\left[\overline{Q}\left(u_\ell^0\right)\right],
\\[7pt]
\displaystyle
0&\mbox{ if } \sup \mathcal{T}_t^-\left[\overline{Q}\left(u_\ell^0\right)\right]>\sup \mathcal{T}_t^+\left[\overline{Q}\left(u_\ell^0\right)\right].
\end{array}
\right.
\end{equation}
Notice that the cases where only one of $\mathcal{T}_t^-\left[\overline{Q}\left(u_\ell^0\right)\right]$ and $\mathcal{T}_t^+\left[\overline{Q}\left(u_\ell^0\right)\right]$ is the empty set enter in the last two items of \eqref{e:hrelay} because of the convention $\sup\emptyset=-\infty$.

Consequently, we have
\begin{equation}\label{e:h=H}
h(t) = H\left[\overline{Q}\left(u_\ell^0\right),h_o\right](t).
\end{equation}
As a consequence, we refer to $\overline{Q}\left(u_\ell^0\right)$ and $h$ as the {\em input} and {\em output} functions, respectively. 
Recalling \eqref{e:TT} and \eqref{e:hrelay}, if $\mathcal{T}_t^-\left[\overline{Q}\left(u_\ell^0\right)\right]$ and $\mathcal{T}_t^+\left[\overline{Q}\left(u_\ell^0\right)\right]$ are not both empty, then $h(t)$ given in \eqref{e:h=H} is determined solely by $\overline{Q}\left(u_\ell^0\right)$.
It is straightforward to check that 
\[
\left(\overline{Q}\left(u_\ell^0(t)\right),h(t)\right)\in\mathcal{L}, \quad \hbox{for every $t>0$.}
\]

\begin{rmk}\label{rmk:leftcont}
Let us note that to assume \eqref{e:hrelay} as definition for the case of left-continuous inputs is anyway coherent with the  standard way of performing such extension by a continuous approximation which \lq\lq monotonically fills the discontinuity jump\rq\rq, see \cite{bagmor, brospre, Recupero}.
\end{rmk}

\begin{rmk}
\label{r:caus}
The evolution described by \eqref{e:hrelay}, \eqref{e:h=H} is an example of causal-memory relation between the input $\overline{Q}\left(u_\ell^0\right)$ and the output $h$: $H\left[\overline{Q}\left(u_\ell^0\right),h_o\right](t)$ depends on the whole past history in $[0,t]$ of $\overline{Q}\left(u_\ell^0\right)$ (but not on the future).
\end{rmk}

We now show that relation \eqref{e:h=H} is independent from positive time-scaling (rate-independence) and satisfies a semi-group property, and hence it is an hysteresis relation in the sense of \cite[pages 59-60]{vis}.
\begin{prop}
The hysteresis relation \eqref{e:h=H} satisfies
\begin{equation}
\label{e:Landmvrks}
H\left[\overline{Q}\left(u_\ell^0\circ\varphi\right),h_o\right] = H\left[\overline{Q}\left(u_\ell^0\right),h_o\right]\circ\varphi,
\end{equation}
for all increasing homeomorphism $\varphi \colon [0, +\infty)\to[0, +\infty)$. Moreover, for all $\tau\in[0,t]$ we have
\begin{equation}
\label{e:Zytecki}
H\left[\overline{Q}\left(u_\ell^0\right), h_o\right](t) = H\left[\overline{Q}\left(u_\ell^0(\,\cdot\,+\tau)\right), H\left[\overline{Q}\left(u_\ell^0\right), h_o\right](\tau)\right](t-\tau).
\end{equation}
\end{prop}

\begin{proof}
Condition \eqref{e:Landmvrks} holds because $\varphi(0)=0$ and $\mathcal{T}_t^\pm\left[\overline{Q}\left(u_\ell^0\circ\varphi\right)\right] = \mathcal{T}_{\varphi(t)}^\pm\left[\overline{Q}\left(u_\ell^0\right)\right]$, see \eqref{e:TT}.
We prove now \eqref{e:Zytecki}.
Fix $\tau\in[0,t]$.
By \eqref{e:TT} we have
\[\mathcal{T}_{t-\tau}^\pm\left[\overline{Q}\left(u_\ell^0(\,\cdot\,+\tau)\right)\right] = 
\left\{s\in[\tau,t] : \pm\left(\overline{Q}\left(u_\ell^0(s)\right) - q_*^-\right) > 0 \right\},\]
and therefore
\begin{equation}\label{e:aldd}
\mathcal{T}_{t}^\pm\left[\overline{Q}\left(u_\ell^0\right)\right] = \mathcal{T}_{\tau}^\pm\left[\overline{Q}\left(u_\ell^0\right)\right] \cup \mathcal{T}_{t-\tau}^\pm\left[\overline{Q}\left(u_\ell^0(\,\cdot\,+\tau)\right)\right].
\end{equation}
Observe that if $s_{1} \in \mathcal{T}_{\tau}^\pm\left[\overline{Q}\left(u_\ell^0\right)\right]$ and $ s_{2} \in \mathcal{T}_{t-\tau}^\pm\left[\overline{Q}\left(u_\ell^0(\,\cdot\,+\tau)\right)\right]$, then $s_{1} \leq \tau \leq s_{2}$.
This implies 
\begin{align*}
\mathcal{T}_{t-\tau}^\pm\left[\overline{Q}\left(u_\ell^0(\,\cdot\,+\tau)\right)\right] \neq \emptyset
&\Longrightarrow
\sup \mathcal{T}_{\tau}^\pm\left[\overline{Q}\left(u_\ell^0\right)\right] \leq \tau \leq\sup \mathcal{T}_{t}^\pm\left[\overline{Q}\left(u_\ell^0\right)\right] = \sup \mathcal{T}_{t-\tau}^\pm\left[\overline{Q}\left(u_\ell^0(\,\cdot\,+\tau)\right)\right] \leq t,
\\
\mathcal{T}_{t-\tau}^\pm\left[\overline{Q}\left(u_\ell^0(\,\cdot\,+\tau)\right)\right] = \emptyset
&\Longrightarrow
\sup \mathcal{T}_{t-\tau}^\pm\left[\overline{Q}\left(u_\ell^0(\,\cdot\,+\tau)\right)\right] \leq \sup \mathcal{T}_{t}^\pm\left[\overline{Q}\left(u_\ell^0\right)\right] = \sup \mathcal{T}_{\tau}^\pm\left[\overline{Q}\left(u_\ell^0\right)\right] \leq \tau.
\end{align*}

Based on  \eqref{e:hrelay}, we distinguish the following cases.
\begin{itemize}[leftmargin=*]

\item
If $\mathcal{T}_{t-\tau}^\pm\left[\overline{Q}\left(u_\ell^0(\,\cdot\,+\tau)\right)\right] = \emptyset$, then by \eqref{e:aldd} we deduce $\mathcal{T}_{t}^\pm\left[\overline{Q}\left(u_\ell^0\right)\right] = \mathcal{T}_{\tau}^\pm\left[\overline{Q}\left(u_\ell^0\right)\right]$. Then \eqref{e:Zytecki} holds true because of \eqref{e:hrelay}: indeed
\[H\left[\overline{Q}\left(u_\ell^0(\,\cdot\,+\tau)\right), H\left[\overline{Q}\left(u_\ell^0\right), h_o\right](\tau)\right](t-\tau) = H\left[\overline{Q}\left(u_\ell^0\right), h_o\right](\tau) = H\left[\overline{Q}\left(u_\ell^0\right), h_o\right](t).\]

\item
Assume that $\sup \mathcal{T}_{t-\tau}^-\left[\overline{Q}\left(u_\ell^0(\,\cdot\,+\tau)\right)\right] < \sup \mathcal{T}_{t-\tau}^+\left[\overline{Q}\left(u_\ell^0(\,\cdot\,+\tau)\right)\right]$.
In this case we have $\mathcal{T}_{t-\tau}^+\left[\overline{Q}\left(u_\ell^0(\,\cdot\,+\tau)\right)\right] \neq \emptyset$ and $\sup \mathcal{T}_t^+\left[\overline{Q}\left(u_\ell^0\right)\right] = \sup\mathcal{T}_{t-\tau}^+\left[\overline{Q}\left(u_\ell^0(\,\cdot\,+\tau)\right)\right] \geq \tau$.
Moreover
\begin{align*}
\mathcal{T}_{t-\tau}^-\left[\overline{Q}\left(u_\ell^0(\,\cdot\,+\tau)\right)\right] \neq \emptyset
&\Longrightarrow
\sup \mathcal{T}_t^-\left[\overline{Q}\left(u_\ell^0\right)\right] = \sup \mathcal{T}_{t-\tau}^-\left[\overline{Q}\left(u_\ell^0(\,\cdot\,+\tau)\right)\right] < \sup \mathcal{T}_{t-\tau}^+\left[\overline{Q}\left(u_\ell^0(\,\cdot\,+\tau)\right)\right],
\\
\mathcal{T}_{t-\tau}^-\left[\overline{Q}\left(u_\ell^0(\,\cdot\,+\tau)\right)\right] = \emptyset
&\Longrightarrow
\sup \mathcal{T}_t^-\left[\overline{Q}\left(u_\ell^0\right)\right] = \sup \mathcal{T}_\tau^-\left[\overline{Q}\left(u_\ell^0\right)\right] \leq \tau \leq \sup\mathcal{T}_{t-\tau}^+\left[\overline{Q}\left(u_\ell^0(\,\cdot\,+\tau)\right)\right].
\end{align*}
As a consequence, we have
\[\sup \mathcal{T}_t^-\left[\overline{Q}\left(u_\ell^0\right)\right] < \sup \mathcal{T}_{t-\tau}^+\left[\overline{Q}\left(u_\ell^0(\,\cdot\,+\tau)\right)\right] = \sup \mathcal{T}_t^+\left[\overline{Q}\left(u_\ell^0\right)\right].
\]
Therefore \eqref{e:Zytecki} holds true because by \eqref{e:hrelay}
\[H\left[\overline{Q}\left(u_\ell^0(\,\cdot\,+\tau)\right), H\left[\overline{Q}\left(u_\ell^0\right), h_o\right](\tau)\right](t-\tau) = 1 = H\left[\overline{Q}\left(u_\ell^0\right), h_o\right](t).\]

\item
Similarly to the previous case, if $\sup \mathcal{T}_{t-\tau}^+\left[\overline{Q}\left(u_\ell^0(\,\cdot\,+\tau)\right)\right] < \sup \mathcal{T}_{t-\tau}^-\left[\overline{Q}\left(u_\ell^0(\,\cdot\,+\tau)\right)\right]$, then
\[\sup \mathcal{T}_t^+\left[\overline{Q}\left(u_\ell^0\right)\right] < \sup \mathcal{T}_{t-\tau}^-\left[\overline{Q}\left(u_\ell^0(\,\cdot\,+\tau)\right)\right] = \sup \mathcal{T}_t^-\left[\overline{Q}\left(u_\ell^0\right)\right]\]
and \eqref{e:Zytecki} holds true because by \eqref{e:hrelay}
\[H\left[\overline{Q}\left(u_\ell^0(\,\cdot\,+\tau)\right), H\left[\overline{Q}\left(u_\ell^0\right), h_o\right](\tau)\right](t-\tau) = 0 = H\left[\overline{Q}\left(u_\ell^0\right), h_o\right](t).\]
\end{itemize}
This concludes the proof.
\end{proof}

\begin{rmk}
As stated in \cite{vis}, the causality (see Remark~\ref{r:caus}) and the rate-independence \eqref{e:Landmvrks} are necessary conditions to have hysteresis; on the other hand, the semi-group property \eqref{e:Zytecki} is not, but it is indeed a natural property to be required, especially in the context of evolutionary equations. Moreover, several analytical models representing real-life situations of hysteresis satisfy the semi-group property (see for instance the Relay, Play, Stop, Preisach, Duhem models in \cite{vis}). 
\end{rmk}

\begin{prop}
The dynamic behavior of $h(t)$ in \eqref{e:hrelay}-\eqref{e:h=H} coincides with that described in items \ref{i:v-i}-\ref{i:v-v}.
\end{prop}
\begin{proof}
If both $\mathcal{T}_{t}^-\left[\overline{Q}\left(u_\ell^0\right)\right]$ and $\mathcal{T}_{t}^+\left[\overline{Q}\left(u_\ell^0\right)\right]$ are empty, then $\overline{Q}\left(u_\ell^0(s)\right) \in [q_*^-,q_*^+]$ for any $s\in[0,t]$, and by \eqref{e:hrelay} we have $h(t)=0$ if $h_o=0$, while $h(t)=1$ if $h_o=1$, and this corresponds to the case described in \ref{i:v-v}.
If either $\mathcal{T}_{t}^-\left[\overline{Q}\left(u_\ell^0\right)\right]$ or $\mathcal{T}_{t}^+\left[\overline{Q}\left(u_\ell^0\right)\right]$ is a non-empty set, then either $\sup \mathcal{T}_{t}^-\left[\overline{Q}\left(u_\ell^0\right)\right] < \sup \mathcal{T}_{t}^+\left[\overline{Q}\left(u_\ell^0\right)\right]$ or $\sup \mathcal{T}_{t}^-\left[\overline{Q}\left(u_\ell^0\right)\right] > \sup \mathcal{T}_{t}^+\left[\overline{Q}\left(u_\ell^0\right)\right]$. In the former case $\overline{Q}\left(u_\ell^0(t)\right)\geq q_*^-$ and by \eqref{e:hrelay} we have $h(t) = 1$, and this corresponds to one of the cases described in \ref{i:v-ii} or \ref{i:v-iii}. At last, in the latter case $\overline{Q}\left(u_\ell^0(t)\right) \in (0,q_*^+]$ and by \eqref{e:hrelay} we have $h(t) = 0$, and this corresponds to one of the cases described in \ref{i:v-i} or \ref{i:v-iv}.
\end{proof}

If \eqref{e:VOLA} is not satisfied, but $u_\ell^0\in\BV_{\rm loc}([0, +\infty);\Omega)$ and is left-continuous, then also $\overline{Q}\left(u_\ell^0\right)\in\BV_{\rm loc}([0, +\infty);[0, +\infty))$ and is left-continuous by \eqref{e:SleepToken}, see \figurename~\ref{f:oQdisco}, and $H\left[\overline{Q}\left(u_\ell^0\right),h_o\right](t)$ given in \eqref{e:hrelay} is still well defined, see also Remark~\ref{rmk:leftcont}. We also note that the left-continuous feature is somehow suitable to treat evolutionary behavior with causal-memory relation. 
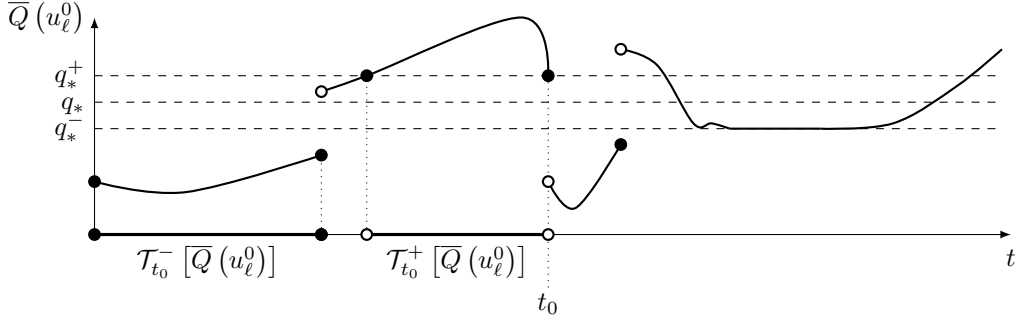
\begin{figure}[!htb]
\centering
\begin{tikzpicture}[x=12mm, y=7mm]
\coordinate (A) at (0,1);
\coordinate (B) at (1,.8);
\coordinate (C) at (2.5,1.5);
\coordinate (D) at (2.5,2.7);
\coordinate (E) at (3,3);
\coordinate (F) at (4.7,4.1);
\coordinate (G) at (5,3);
\coordinate (H) at (5,1);
\coordinate (H1) at (5.3,.5);
\coordinate (H2) at (5.8,1.7);
\coordinate (I0) at (5.8,3.5);
\coordinate (I) at (6.2,3.2);
\coordinate (J) at (6.6,2.1);
\coordinate (K0) at (6.8,2.1);
\coordinate (K) at (7,2);
\coordinate (L) at (8,2);
\coordinate (M) at (8.8,2.1);
\coordinate (N) at (9.5,2.8);
\coordinate (O) at (10,3.5);
\draw[-latex] (0,0) -- (10.1,0) node[below] {\strut$t$}; 
\draw[-latex] (0,0) -- (0,4.1) node[left] {$\overline{Q}\left(u_\ell^0\right)$};
\draw[dashed] (0,2) node[left] {\strut$q_*^-$} --(10,2);
\draw[dashed] (0,2.5) node[left] {\strut$q_*$} --(10,2.5);
\draw[dashed] (0,3) node[left] {\strut$q_*^+$} --(10,3);
\draw[thick] plot [smooth, tension=.6] coordinates { (A) (B) (C) };
\draw[thick] plot [smooth, tension=.6] coordinates { (D) (E) (F) (G) };
\draw[thick] plot [smooth, tension=.6] coordinates { (H) (H1) (H2) };
\draw[thick] plot [smooth, tension=.6] coordinates { (I0) (I) (J) (K0) (K) };
\draw[thick] (K) -- (L);
\draw[thick] plot [smooth, tension=.6] coordinates { (L) (M) (N) (O) };

\draw[dotted] (2.5,0)--(C);
\draw[very thick] (0,0)--(2.5,0) node[midway,below]{\strut$\mathcal{T}_{t_0}^-\left[\overline{Q}\left(u_\ell^0\right)\right]$};
\filldraw [fill=white, thick] (2.5,0) circle (2pt);

\draw[dotted] (3,0)--(E);
\draw[very thick] (3,0)--(5,0) node[midway,below]{\strut$\mathcal{T}_{t_0}^+\left[\overline{Q}\left(u_\ell^0\right)\right]$};

\draw[dotted] (5,-1) node[below, inner sep=0pt]{\strut$t_0$} --(G);

\filldraw [fill=black, thick] (0,0) circle (2pt);
\filldraw [fill=black, thick] (A) circle (2pt);
\filldraw [fill=black, thick] (C) circle (2pt);
\filldraw [fill=white, thick] (D) circle (2pt);
\filldraw [fill=black, thick] (E) circle (2pt);
\filldraw [fill=white, thick] (3,0) circle (2pt);
\filldraw [fill=black, thick] (G) circle (2pt);
\filldraw [fill=white, thick] (H) circle (2pt);
\filldraw [fill=black, thick] (H2) circle (2pt);
\filldraw [fill=white, thick] (I0) circle (2pt);
\filldraw [fill=black, thick] (2.5,0) circle (2pt);
\filldraw [fill=white, thick] (5,0) circle (2pt);

\end{tikzpicture}
\caption{The sets $\mathcal{T}_{t_0}^\pm\left[\overline{Q}\left(u_\ell^0\right)\right]$ corresponding to an hypothetical maximum achievable free flow $\overline{Q}\left(u_\ell^0\right)$ in $\BV_{\rm loc}([0, +\infty);[0, +\infty))$ and left-continuous, for some $t_0>0$.}
\label{f:oQdisco}
\end{figure}
At last, the flow through the valve is
\begin{equation}
\label{e:valveA02}
q^0(t) = h(t) \min\left\{q_*,\overline{Q}\left(u(t,0^-)\right)\right\},
\end{equation}
where $h$ is defined in \eqref{e:h=H}.

Finally, the case when \eqref{e:uhreg} does not hold true corresponds to the chattering of the valve.
We will discuss this case later, see Remark~\ref{r:PorcupineTree}.

\subsection{The c-Riemann solver associated to the valve}\label{subsec:cRiemannValve}

In this section we introduce the c-Riemann solver corresponding to the valve described in Section~\ref{s:descr} in the case \eqref{e:uhreg} holds true, namely, in the case without the chattering of the valve or, equivalently, for coherent Riemann initial data.
By this we mean that $\hc(u_\ell,h_o)$ and $\qc(u_\ell,h_o)$, given below in \eqref{e:hc} and \eqref{e:qc}, coincide with $h(t)$ and $q^{0}(t)$, $t>0$, with $h$ defined in \eqref{e:h=H} and $q^{0}$ in \eqref{e:valveA02}, and correspond to the Riemann initial datum $(u_{\ell},u_{r},h_{o})\in \mathsf{CH}$.

The c-function $C \colon \dom \times \{0,1\} \to \dom \times \{0,1\}$ is defined as follows. For any $(u_\ell,u_r,h_o) \in \dom \times \{0,1\}$, we define
\begin{align}
\label{e:hc}
\hc(u_\ell,h_o) &\doteq 
\left\{
\begin{array}{@{}l@{\quad}l@{}}
h_o &\mbox{ if } \overline{Q}\left(u_\ell\right) \in [q_*^-,q_*^+],
\\[2pt]
1&\mbox{ if } \overline{Q}\left(u_\ell\right) > q_*^+,
\\[2pt]
0&\mbox{ if } \overline{Q}\left(u_\ell\right) < q_*^-,
\end{array}
\right.
\\
\label{e:qc}
\qc(u_\ell,h_o) &\doteq \hc(u_\ell,h_o) \, \min\left\{ q_* , \overline{Q}(u_\ell) \right\}.
\end{align}
Then, the functions $\ucm(u_\ell,h_o) \in \Gamma^{-}(u_{\ell})$, $\ucp(u_\ell,u_r,h_o) \in \Gamma^{+}(u_{r})$ are uniquely determined by \eqref{e:ucpm} and so $C$, as in Definition~\ref{d:RSV}, see \eqref{e:C}, is well defined.
We then introduce the corresponding c-Riemann solver
\begin{equation}
\label{e:InFlames}
\rsc[u_\ell,u_r,h_o](\xi) =
\begin{cases}
\left(\rss[u_\ell,\ucm](\xi),\hc\right) &\hbox{if } \xi<0,
\\
\left(\rss[\ucp,u_r](\xi),\hc\right) &\hbox{if } \xi\geq0.
\end{cases}
\end{equation}
Observe that $\qc(u_\ell,h_o)\in \left\{0\right\}\cup[q_*^-,q_*]$ by \eqref{e:hc} and \eqref{e:qc}.
Moreover, $\qc(u_\ell,h_o) = 0$ if and only if $\hc(u_\ell,h_o) = 0$, that is, if and only if either $\overline{Q}\left(u_\ell\right) < q_*^-$ or $\overline{Q}\left(u_\ell\right) \in [q_*^-,q_*^+]$ and $h_{o}=0$.

In the following proposition we investigate the link between the valve described in Section~\ref{s:descr}, in the case \eqref{e:uhreg} holds true, and condition \eqref{e:coherent}\textsubscript{$2$}, which is necessary for the coherence of the c-Riemann solver defined in \eqref{e:InFlames}, see Definition~\ref{d:coherent}.
\begin{prop}
\label{p:equi}
Let $(u,h)$ be a global c-solution of the Cauchy problem \eqref{e:csystem}, in the sense of Definition~\ref{d:csol}, corresponding to the initial datum $(u_o,h_o)$.
For a fixed $t>0$, the following conditions are equivalent:
\begin{gather}
\label{e:Moodring1}
h(t) = H\left[\overline{Q}\left(u(\,\cdot\,,0^{-})\right),h_o\right](t),
\qquad
q(t,0) = h(t) \min\left\{q_*,\overline{Q}\left(u(t,0^-)\right)\right\},
\qquad
v\left(u(t,0^+)\right) \leq a,
\\
\label{e:Moodring2}
C\left(u(t,0^-),u(t,0^+),h(t)\right) = \left(u(t,0^-),u(t,0^+),h(t)\right),
\end{gather}
where $H$ is defined in \eqref{e:hrelay}.
\end{prop}

\begin{proof}
Assume that $h(t)$ and $u(t,0^\pm)$ satisfy \eqref{e:Moodring1}.
By \eqref{e:Zytecki} with $\tau=t$, \eqref{e:hrelay}, \eqref{e:TT} and \eqref{e:hc} we have
\begin{align*}
h(t) &= 
H\left[\overline{Q}\left(u(\,\cdot\,,0^{-})\right),h_o\right](t) =
H\left[\overline{Q}\left(u(t,0^{-})\right),H\left[\overline{Q}\left(u(\,\cdot\,,0^{-})\right),h_o\right](t)\right](0) 
\\&=
\left\{
\begin{array}{@{}l@{\quad}l@{}}
\displaystyle
H\left[\overline{Q}\left(u(\,\cdot\,,0^{-})\right),h_o\right](t)&\mbox{ if } \mathcal{T}_0^-\left[\overline{Q}\left(u(t,0^{-})\right)\right] = \emptyset = \mathcal{T}_0^+\left[\overline{Q}\left(u(t,0^{-})\right)\right]
\\[7pt]
\displaystyle
1&\mbox{ if } \sup \mathcal{T}_0^-\left[\overline{Q}\left(u(t,0^{-})\right)\right]<\sup \mathcal{T}_0^+\left[\overline{Q}\left(u(t,0^{-})\right)\right]
\\[7pt]
\displaystyle
0&\mbox{ if } \sup \mathcal{T}_0^-\left[\overline{Q}\left(u(t,0^{-})\right)\right]>\sup \mathcal{T}_0^+\left[\overline{Q}\left(u(t,0^{-})\right)\right]
\end{array}
\right.
\\&=
\left\{
\begin{array}{@{}l@{\quad}l@{}}
h(t) &\mbox{ if } \overline{Q}\left(u(t,0^{-})\right) \in [q_*^-,q_*^+]
\\[5pt]
1&\mbox{ if } \overline{Q}\left(u(t,0^{-})\right) > q_*^+
\\[5pt]
0&\mbox{ if } \overline{Q}\left(u(t,0^{-})\right) < q_*^-
\end{array}
\right.
\\&=
\hc(u(t,0^{-}),h(t)).
\end{align*}
We have then
\[q(t,0) = h(t) \min\left\{q_*,\overline{Q}\left(u(t,0^-)\right)\right\} = \hc(u(t,0^{-}),h(t))  \min\left\{q_*,\overline{Q}\left(u(t,0^-)\right)\right\} = \qc\left(u(t,0^-),h(t)\right).\]
By \eqref{e:Moodring1}\textsubscript{$3$}, Propositions~\ref{p:Korn}, \ref{p:Loathe}, \eqref{e:Gammaum} and \eqref{e:Gammaup}, this ensures that \eqref{e:Moodring2} holds, see \eqref{e:ucpm}.
The converse statement follows from analogous arguments; this completes the proof.
\end{proof}

\section{Main results}
\label{s:main}
\setcounter{equation}{0}

In this section we collect our main results.
For any $(u_\ell,u_r,h_o) \in \dom \times \{0,1\}$, the map $(u(t,x),h(t)) = \rsc[u_\ell,u_r,h_o](x/t)$ is by construction a c-solution in the sense of Definition~\ref{d:csol}.
However, in general it may fail to be a coherent c-solution since condition
\eqref{e:coCau} may not be satisfied.
This is a consequence of the lack of coherence of $\rsc$ in $\dom\times\{0,1\}$, as shown in the following proposition.
In fact, we show that for some initial data $(u_\ell,u_r,h_o)$ the valve closes, i.e.\ $\qc(u_\ell,h_o) = 0$, but the resulting traces at the valve position, $\left(\ucm(u_\ell,h_o),\ucp(u_\ell,u_r,h_o)\right) = \left(\hat{u}(0,u_{\ell}),\check{u}(0,u_{r})\right)$, in turn force the valve to open because $\overline{Q}\left(\hat{u}(0,u_{\ell})\right) > q_{*}^{+}$, so that the flow through the valve is $\qc(\hat{u}(0,u_{\ell}),0) = q_{*}$.

We need to introduce some notation, see \figurename~\ref{f:CH}.
First we denote:
\begin{itemize}
\item 
$u_{*a}^{+} \doteq \left(q_{*}^{+}/a,q_{*}^{+}\right)$ as the intersection of the sonic line $q=a\,\rho$ with the horizontal line $q=q_{*}^+$;
\item 
$u_{*0}^+$ as the unique element of $\mathsf{BL}_{1}^{u_{*a}^+}$ with $q=0$;
\item 
$u_{*a0}^{+}$ and $u_{*a0}^{-}$ as the states with the smallest $\rho$-coordinates in the intersections of $\mathsf{BL}_{1}^{u_{*0}^+}$ with $q=q_{*}^{+}$ and $q=q_{*}^{-}$, respectively.
\end{itemize}
At last we define
\begin{equation}
\label{e:NIN}
\begin{aligned}
\mathsf{CH}_{\ell,0}^{\scriptscriptstyle\complement}(q_*^+) &\doteq 
\left\{ (\rho,q) \in \Omega : \rho < \rho_{*a0}^{+},\ \mathcal{BL}_1^{u_{*0}^+}(\rho) < q \leq q_*^+ \right\},
\\
\mathsf{CH}_{\ell,1}^{\scriptscriptstyle\complement}(q_*^-,q_*^+) &\doteq 
\left\{ (\rho,q) \in \Omega : \rho < \rho_{*a0}^{-},\ \mathcal{BL}_1^{u_{*0}^+}(\rho) < q < q_*^- \right\}.
\end{aligned}
\end{equation}
The index $\ell$ recalls that the sets defined in \eqref{e:NIN} will be the $u_\ell$-components of $(u_\ell,u_r,h)$ in $\mathsf{CH}^{\scriptscriptstyle\complement}$, see Definition~\ref{d:coherent} and Theorem~\ref{t:coherence}.
The indices $0$ and $1$ refer to $h_o=0$ and $h_o=1$, respectively, see \figurename~\ref{f:CH} and \eqref{e:incodom2}.
\begin{figure}[!htb]\centering\maxsizebox{.99\linewidth}{!}{
\begin{tikzpicture}[every node/.style={anchor=south west,inner sep=0pt},x=1mm, y=1mm]
\node at (4,4) {\includegraphics[width=50mm]{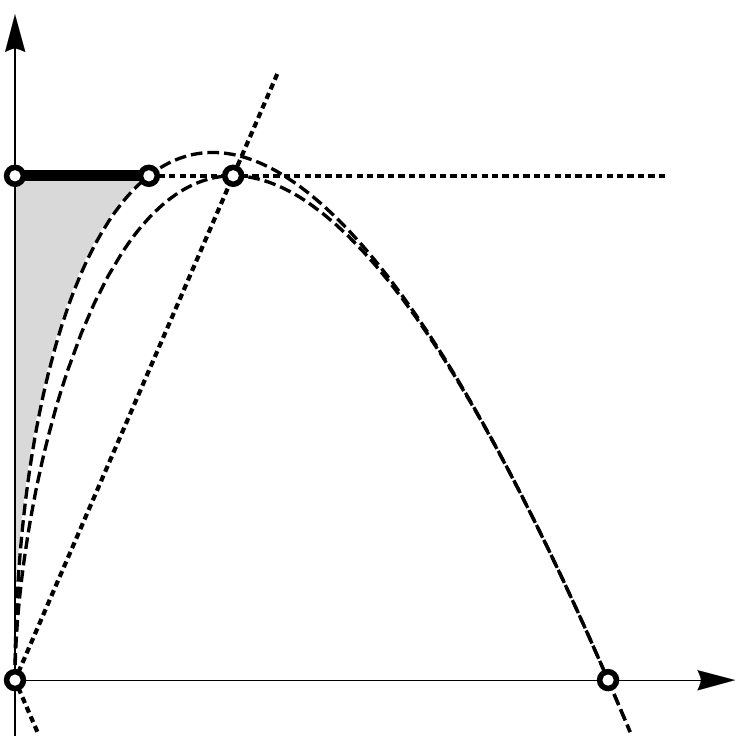}};
\node[style={anchor=center}] at (25,55) {\strut $h_{o}=0$};
\node at (55,7) {\strut $\rho$};
\node at (0,50) {\strut $q$};
\draw[latex-] (11,35) --++ (7,-2) node[right] {\strut $\mathsf{BL}_1^{u_{*a}^{+}}$};
\draw[latex-] (9,36) --++ (-7,2) node[left] {\strut $\mathsf{BL}_1^{u_{*0}^{+}}$};
\node at (19,37.5) {\strut $u_{*a}^{+}$};
\node at (40,4) {\strut $u_{*0}^{+}$};
\node at (0,41) {\strut $q_{*}^{+}$};
\node at (9,43.5) {\strut $u_{*a0}^{+}$};
\begin{scope}[shift={(70,0)}]
\node at (4,4) {\includegraphics[width=50mm]{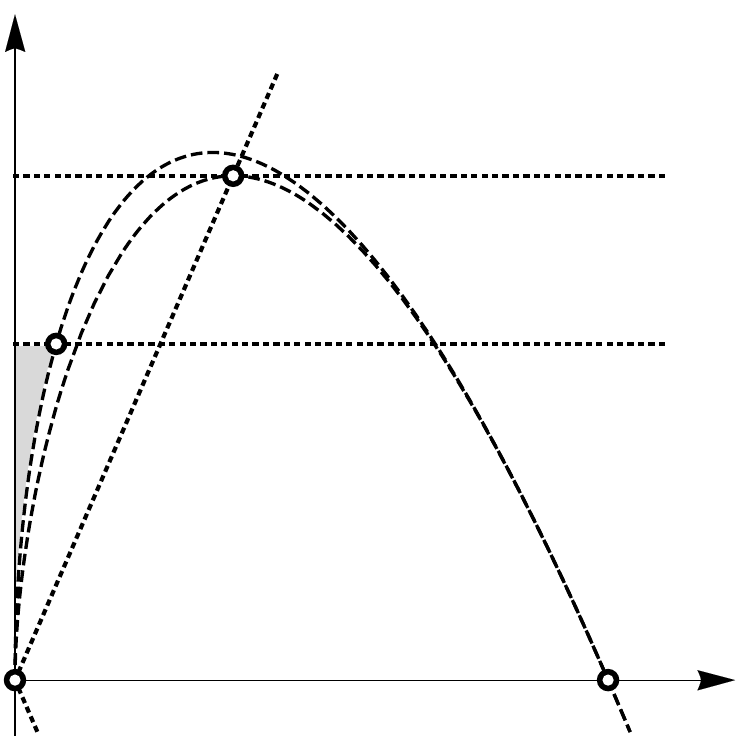}};
\node[style={anchor=center}] at (25,55) {\strut $h_{o}=1$};
\node at (55,7) {\strut $\rho$};
\node at (0,50) {\strut $q$};
\draw[latex-] (14,39) --++ (6,-4) node[right] {\strut $\mathsf{BL}_1^{u_{*a}^{+}}$};
\draw[latex-] (12,41) --++ (-9,6) node[left] {\strut $\mathsf{BL}_1^{u_{*0}^{+}}$};
\node at (19,37.5) {\strut $u_{*a}^{+}$};
\node at (40,4) {\strut $u_{*0}^{+}$};
\node at (0,29) {\strut $q_{*}^{-}$};
\node at (0,41) {\strut $q_{*}^{+}$};
\draw[latex-] (7,32) --++ (-5,4) node[left] {\strut $u_{*a0}^{-}$};
\end{scope}
\end{tikzpicture}}
\caption{Representation of $\mathsf{CH}_{\ell,0}^{\scriptscriptstyle\complement}(q_*^+)$ and $\mathsf{CH}_{\ell,1}^{\scriptscriptstyle\complement}(q_*^-,q_*^+)$, see \eqref{e:NIN}, corresponding to the cases $h_{o}=0$, left, and $h_{o}=1$, right.}
\label{f:CH}
\end{figure}
\begin{teo}
\label{t:coherence}
We have
\begin{subequations}
\label{e:incodom}
\begin{align}
\label{e:incodom1}
\mathsf{CH}^{\scriptscriptstyle\complement} &= \left\{
(u_\ell,u_r,h_o)\in\dom\times\{0,1\} : \qc(u_\ell,h_o) = 0, \ \overline{Q}\left(\hat{u}(0,u_{\ell})\right) > q_{*}^{+} \right\}
\\&=
\left(\mathsf{CH}_{\ell,0}^{\scriptscriptstyle\complement}(q_*^+) \times \Omega \times \{0\}\right)
\cup
\left(\mathsf{CH}_{\ell,1}^{\scriptscriptstyle\complement}(q_*^-,q_*^+) \times \Omega \times \{1\}\right).
\label{e:incodom2}
\end{align}
\end{subequations}
Furthermore, if $(u_{\ell},u_r,h_o) \in \mathsf{CH}^{\scriptscriptstyle\complement}$ then $v(u_\ell)>a$.
\end{teo}

\begin{proof}
We first prove \eqref{e:incodom1}.
Observe that $(\ucm,\ucp)$ given in \eqref{e:ucpm} is well defined because $\qc \in [0,\overline{Q}(u_\ell)] \subseteq [0,\Bar{q}(u_\ell)]$, see Definition~\ref{d:notation}, and $(\ucm,\ucp) \in \dom$.
Hence, by Definition~\ref{d:coherent}, $\rsc$ is coherent at $(u_\ell,u_r,h_o)$ if and only if
\begin{equation*}
C\left(\ucm(u_\ell,h_o),\ucp(u_\ell,u_r,h_o),\hc(u_\ell,h_o)\right) = \left(\ucm(u_\ell,h_o),\ucp(u_\ell,u_r,h_o),\hc(u_\ell,h_o)\right),
\end{equation*}
which, by \eqref{e:ucpm}, is equivalent to require that
\begin{align*}
&\hc\left(\ucm(u_\ell,h_o),\hc(u_\ell,h_o)\right) = \hc(u_\ell,h_o),&
&\qc\left(\ucm(u_\ell,h_o),\hc(u_\ell,h_o)\right) = \qc(u_\ell,h_o).
\end{align*}
The proof is carried out by the following case-by-case analysis.

\begin{itemize}[leftmargin=*]

\item
Assume $\qc(u_\ell,h_o)=q_{*}$.
In this case $\hc(u_\ell,h_o) = 1$ and $\overline{Q}\left(\ucm(u_\ell,h_o)\right) \geq \qc(u_\ell,h_o) = q_{*}$.
By \eqref{e:hc} and \eqref{e:qc} we have then
\begin{align*}
&\hc\left(\ucm(u_\ell,h_o),\hc(u_\ell,h_o)\right) = 1 = \hc(u_\ell,h_o),&
&\qc\left(\ucm(u_\ell,h_o),\hc(u_\ell,h_o)\right) = q_{*} = \qc(u_\ell,h_o).
\end{align*}
Therefore, $\rsc$ is coherent at $(u_\ell,u_r,h_o)$.

\item
Assume $\qc(u_\ell,h_o) \in [q_{*}^{-},q_{*})$.
In this case $\hc(u_\ell,h_o) = 1$ and $\qc(u_\ell,h_o) = \overline{Q}(u_\ell)$.
By \eqref{e:hc} and \eqref{e:qc} we have $\overline{Q}\left(\ucm(u_\ell,h_o)\right) \geq \qc(u_\ell,h_o) \geq q_{*}^{-}$ and then
\begin{align*}
&\hc\left(\ucm(u_\ell,h_o),\hc(u_\ell,h_o)\right) = 1 = \hc(u_\ell,h_o),&
&\qc\left(\ucm(u_\ell,h_o),\hc(u_\ell,h_o)\right) = \min\left\{ q_* , \overline{Q}\left(\ucm(u_\ell,h_o)\right) \right\}.
\end{align*}
We distinguish two subcases:
\begin{itemize}[leftmargin=*]
\item
Assume $v(u_{\ell}) \geq a$.
In this case $\qc(u_{\ell},h_{o}) = \overline{Q}(u_\ell) = q_{\ell} \in [q_{*}^{-},q_{*})$, see \eqref{e:SleepToken}; by \eqref{e:ucpm}\textsubscript{$1$} we have $\ucm(u_{\ell},h_{o}) = u_{\ell}$.
As a consequence, we deduce
\[\qc\left(\ucm(u_\ell,h_o),h(u_\ell,h_o)\right) = \min\left\{ q_* , \overline{Q}(u_\ell) \right\} = q_{\ell} = \qc(u_{\ell},h_{o}).\]

\item
Assume $v(u_{\ell}) < a$.
In this case $\qc(u_{\ell},h_{o}) = \overline{Q}(u_\ell) = \Bar{q}(u_\ell) \in [q_{*}^{-},q_{*})$, hence $\ucm(u_{\ell},h_{o}) = \Bar{u}(u_\ell)$.
This implies $\overline{Q}\left(\ucm(u_{\ell},h_{o})\right) = \Bar{q}(u_\ell) \in [q_{*}^{-},q_{*})$ and therefore
\[\qc\left(\ucm(u_\ell,h_o),\hc(u_\ell,h_o)\right) = \Bar{q}(u_\ell) = \qc(u_{\ell},h_{o}).\]
\end{itemize}
Therefore, in both cases, $\rsc$ is coherent at $(u_\ell,u_r,h_o)$.

\item
Assume $\qc(u_\ell,h_o) = 0$.
In this case $\hc(u_\ell,h_o) = 0$ and $\overline{Q}(u_{\ell}) \leq q_{*}^{+}$.
By \eqref{e:ucpm}\textsubscript{$1$} we have then $\ucm(u_\ell,h_o) = \hat{u}(0,u_{\ell})$.
We distinguish two cases:
\begin{itemize}[leftmargin=*]
\item
Assume $\overline{Q}\left(\hat{u}(0,u_{\ell})\right) > q_{*}^{+}$.
By \eqref{e:hc} we have $\hc\left(\ucm(u_\ell,h_o),\hc(u_\ell,h_o)\right) = 1 \neq 0 = \hc(u_\ell,h_o)$.
Therefore, $\rsc$ is not coherent at $(u_\ell,u_r,h_o)$, see \figurename~\ref{f:coherence}.
Furthermore, in this case $v(u_{\ell}) > a$.
Indeed, if by contradiction $v(u_{\ell}) \leq a$, then by item~\ref{i:d} of Lemma~\ref{l:JSatriani} we have $q_{*}^{+} \geq \overline{Q}(u_{\ell}) = \Bar{q}(u_\ell) \geq \Bar{q}(\hat{\rho}(0,u_{\ell})) = \overline{Q}\left(\hat{u}(0,u_{\ell})\right) > q_{*}^{+}$, which gives a contradiction.

\item
Assume $\overline{Q}\left(\hat{u}(0,u_{\ell})\right) \leq q_{*}^{+}$.
In this case, by \eqref{e:hc} we have $\hc\left(\ucm(u_\ell,h_o),\hc(u_\ell,h_o)\right) = \hc(u_\ell,h_o) = 0$ and therefore $\qc\left(\ucm(u_\ell,h_o),\hc(u_\ell,h_o)\right) = 0 = \qc(u_\ell,h_o)$.
Therefore, $\rsc$ is coherent at $(u_\ell,u_r,h_o)$.
\end{itemize}
\end{itemize}
This concludes the proof of \eqref{e:incodom1}.

\begin{figure}[!htb]\centering\resizebox{.7\linewidth}{!}{
\begin{tikzpicture}[every node/.style={anchor=south west,inner sep=0pt},x=1mm, y=1mm]

\node at (4,4) {\includegraphics[width=50mm]{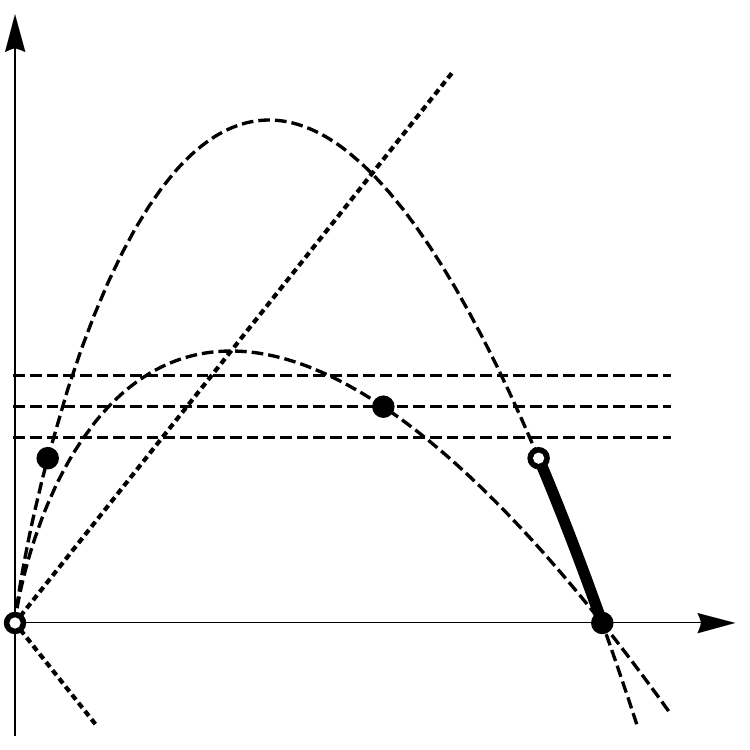}};
\node[style={anchor=center}] at (25,55) {\strut $h_{o}\in\{0,1\}$};
\node at (55,10) {\strut $\rho$};
\node at (0,50) {\strut $q$};
\node at (0,21) {\strut $q_*^-$};
\draw[latex-]  (4,27) --++ (-5,0) node[left] {\strut $q_*$};
\node at (0,28) {\strut $q_*^+$};
\node[right] at (39,4) {\strut $\mathsf{FL}_1^{u_\ell}$};
\node[right] at (49,7) {\strut $\mathsf{FL}_1^{\ucm}$};
\node at (9,21) {\strut $u_\ell$};
\node at (42,20) {\strut $\hat{u}(q_\ell,u_\ell)$};
\node[right] at (35.5,50) {\strut $q=a\rho$};
\node at (45,12.5) {\strut $\ucm$};
\draw[latex-] (30,27) --++ (10,5) node[right] {\strut $(\ucm)_{\rm c}^-$};

\begin{scope}[shift={(70,0)}]
\node at (4,4) {\includegraphics[width=50mm]{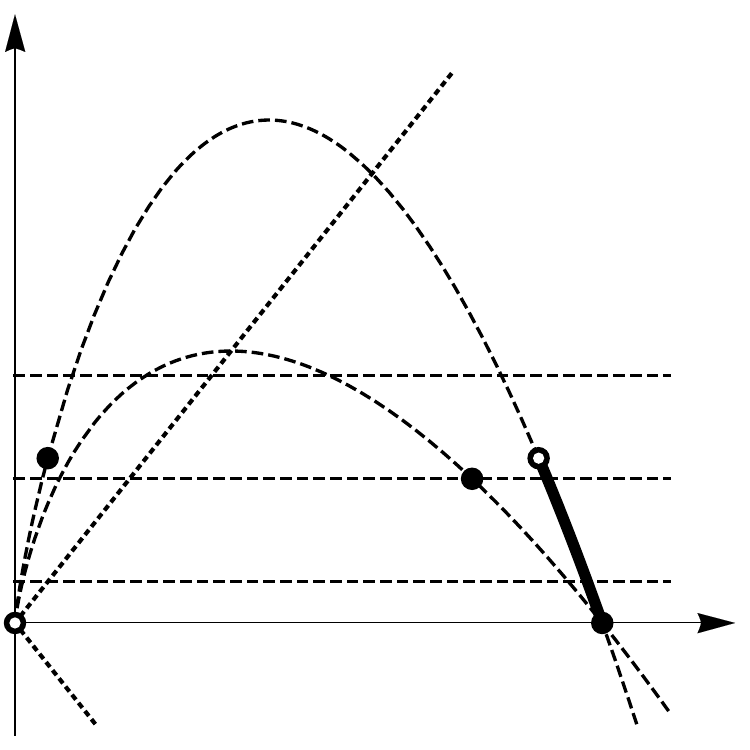}};
\node[style={anchor=center}] at (25,55) {\strut $h_{o}=0$};
\node at (55,10) {\strut $\rho$};
\node at (0,50) {\strut $q$};
\node at (0,14) {\strut $q_*^-$};
\node at (0,20) {\strut $q_*$};
\node at (0,27) {\strut $q_*^+$};
\node[right] at (39,4) {\strut $\mathsf{FL}_1^{u_\ell}$};
\node[right] at (49,7) {\strut $\mathsf{FL}_1^{\ucm}$};
\node at (9.5,22) {\strut $u_\ell$};
\node at (41,22) {\strut $\hat{u}(q_\ell,u_\ell)$};
\node[right] at (35.5,50) {\strut $q=a\rho$};
\node at (39.5,8.5) {\strut $\ucm$};
\node at (25,18) {\strut $(\ucm)_{\rm c}^-$};

\end{scope}

\end{tikzpicture}}
\caption{The c-Riemann solver given in \eqref{e:InFlames} is not coherent, see Theorem~\ref{t:coherence}.
On the left $h_{o}\in\{0,1\}$ and $q_\ell<q_*^-$, whereas on the right $h_{o}=0$ and $q_\ell \in (q_*,q_*^+)$.
In both cases depicted above we take $v(u_{\ell}) > a$ and $\hc(u_\ell,h_o) = 0$.
For brevity, we write $\ucm$ for $\ucm(u_\ell,h_o)$ and $(\ucm)_{\rm c}^-$ for $\ucm(\ucm,0)$.
We also represent $\Gamma^-(u_\ell)$ to emphasize that $(\ucm)_{\rm c}^- \notin\Gamma^-(u_\ell)$ in both cases.}
\label{f:coherence}
\end{figure}

We now prove \eqref{e:incodom2}; we also discuss how the set $\mathsf{CH}^{\scriptscriptstyle\complement}$ depends on the thresholds $q_*^-$ and $q_*^+$.
First, by \eqref{e:incodom1}, \eqref{e:hc} and \eqref{e:qc} we notice that 
\[\mathsf{CH}^{\scriptscriptstyle\complement}(q_*^-,q_*^+) = \left( \mathsf{CH}^{\scriptscriptstyle\complement}_{1}(q_*^-,q_*^+) \cup \mathsf{CH}^{\scriptscriptstyle\complement}_{2}(q_*^-) \right) \cap \mathsf{CH}^{\scriptscriptstyle\complement}_{3}(q_*^+),\]
where
\begin{subequations}
\label{e:CH123}
\begin{align}
\label{e:CH1}
\mathsf{CH}^{\scriptscriptstyle\complement}_{1}(q_*^-,q_*^+) & \doteq \left\{ (u_\ell,u_r,h_o)\in\dom\times\{0\} : \overline{Q}\left(u_\ell\right) \in [q_*^-,q_*^+] \right\},
\\
\label{e:CH2}
\mathsf{CH}^{\scriptscriptstyle\complement}_{2}(q_*^-) & \doteq \left\{ (u_\ell,u_r,h_o)\in\dom\times\{0,1\} : \overline{Q}\left(u_\ell\right) < q_*^- \right\},
\\
\label{e:CH3}
\mathsf{CH}^{\scriptscriptstyle\complement}_{3}(q_{*}^{+}) & \doteq \left\{ (u_\ell,u_r,h_o)\in\dom\times\{0,1\} : \overline{Q}\left(\hat{u}(0,u_{\ell})\right) > q_{*}^{+} \right\}.
\end{align}
\end{subequations}
Observe that $\mathsf{CH}^{\scriptscriptstyle\complement}_{1}(q_*^-,q_*^+) \cup \mathsf{CH}^{\scriptscriptstyle\complement}_{2}(q_*^-) = \left\{ (u_\ell,u_r,h_o)\in\dom\times\{0,1\} : \qc(u_\ell,h_o) = 0 \right\}$.
We also have
\begin{align*}
[q_{*\alpha}^-,q_{*\alpha}^+] \subseteq [q_{*\beta}^-,q_{*\beta}^+]
&\Longrightarrow
\mathsf{CH}^{\scriptscriptstyle\complement}_{1}(q_{*\alpha}^-,q_{*\alpha}^+) \subseteq \mathsf{CH}^{\scriptscriptstyle\complement}_{1}(q_{*\beta}^-,q_{*\beta}^+),
\\
q_{*\alpha}^{-} \leq q_{*\beta}^{-}
&\Longrightarrow
\mathsf{CH}^{\scriptscriptstyle\complement}_{2}(q_{*\alpha}^-) \subseteq \mathsf{CH}^{\scriptscriptstyle\complement}_{2}(q_{*\beta}^-).
\end{align*}

Roughly speaking, $\mathsf{CH}^{\scriptscriptstyle\complement}_{1}(q_*^-,q_*^+)$, $\mathsf{CH}^{\scriptscriptstyle\complement}_{2}(q_*^-)$ and $\mathsf{CH}^{\scriptscriptstyle\complement}_{3}(q_*^+)$ can be represented as in \figurename~\ref{f:CH123}, and $\mathsf{CH}^{\scriptscriptstyle\complement}(q_*^-,q_*^+)$ as in \figurename~\ref{f:CH}.
More precisely, let $\pi_{\ell} \colon \dom \times \{0,1\} \to \Omega$ denote the projection onto the first $u$-component, namely
\[
\pi_{\ell}\left(u_{\ell},u_{r},h_{o}\right) = u_{\ell}.
\]
Then in \figurename~\ref{f:CH123} we represent $\pi_{\ell}\left(\mathsf{CH}^{\scriptscriptstyle\complement}_{1}(q_*^-,q_*^+)\right)$, $\pi_{\ell}\left(\mathsf{CH}^{\scriptscriptstyle\complement}_{2}(q_*^-)\right)$ and $\pi_{\ell}\left(\mathsf{CH}^{\scriptscriptstyle\complement}_{3}(q_*^+)\right)$, while for $\pi_{\ell}\left(\mathsf{CH}^{\scriptscriptstyle\complement}(q_*^-,q_*^+)\right)$ see \figurename~\ref{f:CH}.
\begin{figure}[!htb]\centering\maxsizebox{.99\linewidth}{!}{
\begin{tikzpicture}[every node/.style={anchor=south west,inner sep=0pt},x=1mm, y=1mm]
\node at (4,4) {\includegraphics[width=50mm]{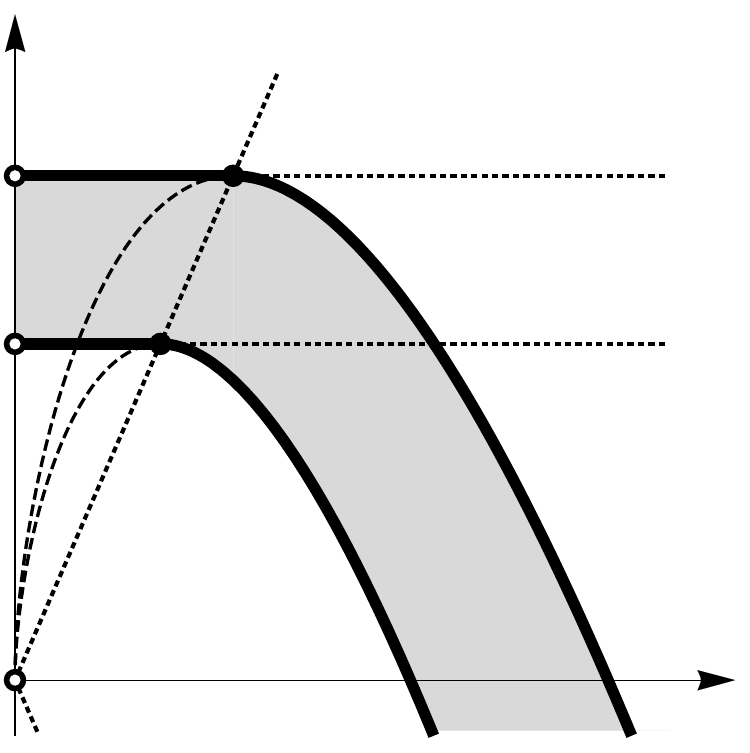}};
\node[style={anchor=center}] at (25,55) {\strut $h_{o}=0$};
\node at (55,7) {\strut $\rho$};
\node at (0,50) {\strut $q$};
\node at (23,2) {\strut $\mathsf{BL}_1^{u_{*a}^{-}}$};
\node at (47,2) {\strut $\mathsf{BL}_1^{u_{*a}^{+}}$};
\node at (14,26) {\strut $u_{*a}^{-}$};
\node at (21.5,42.5) {\strut $u_{*a}^{+}$};
\node at (0,29) {\strut $q_{*}^{-}$};
\node at (0,41) {\strut $q_{*}^{+}$};
\begin{scope}[shift={(70,0)}]
\node at (4,4) {\includegraphics[width=50mm]{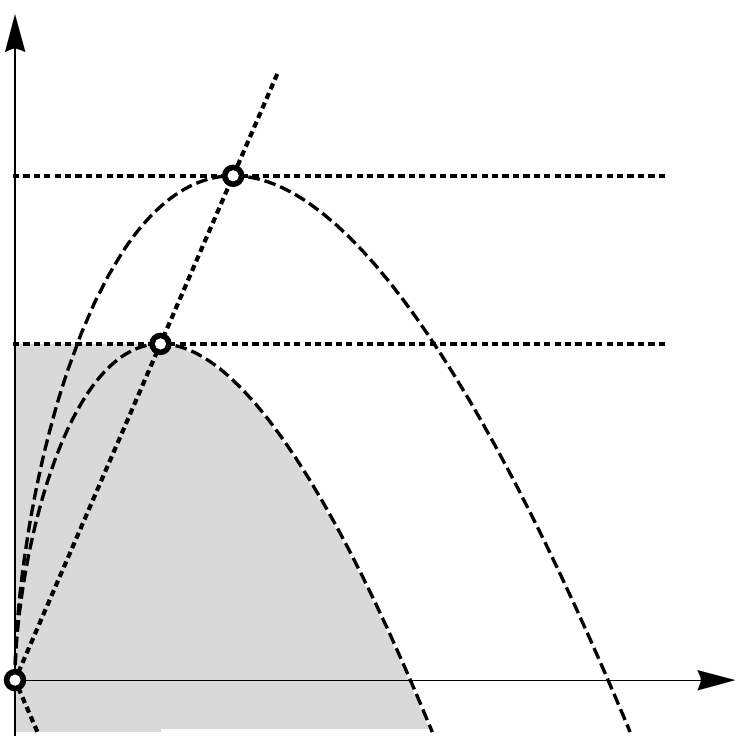}};
\node[style={anchor=center}] at (25,55) {\strut $h_{o}\in\{0,1\}$};
\node at (55,7) {\strut $\rho$};
\node at (0,50) {\strut $q$};
\node at (34,2) {\strut $\mathsf{BL}_1^{u_{*a}^{-}}$};
\node at (47,2) {\strut $\mathsf{BL}_1^{u_{*a}^{+}}$};
\node at (17,31.5) {\strut $u_{*a}^{-}$};
\node at (21.5,42.5) {\strut $u_{*a}^{+}$};
\node at (0,29) {\strut $q_{*}^{-}$};
\node at (0,41) {\strut $q_{*}^{+}$};
\end{scope}
\begin{scope}[shift={(140,0)}]
\node at (4,4) {\includegraphics[width=50mm]{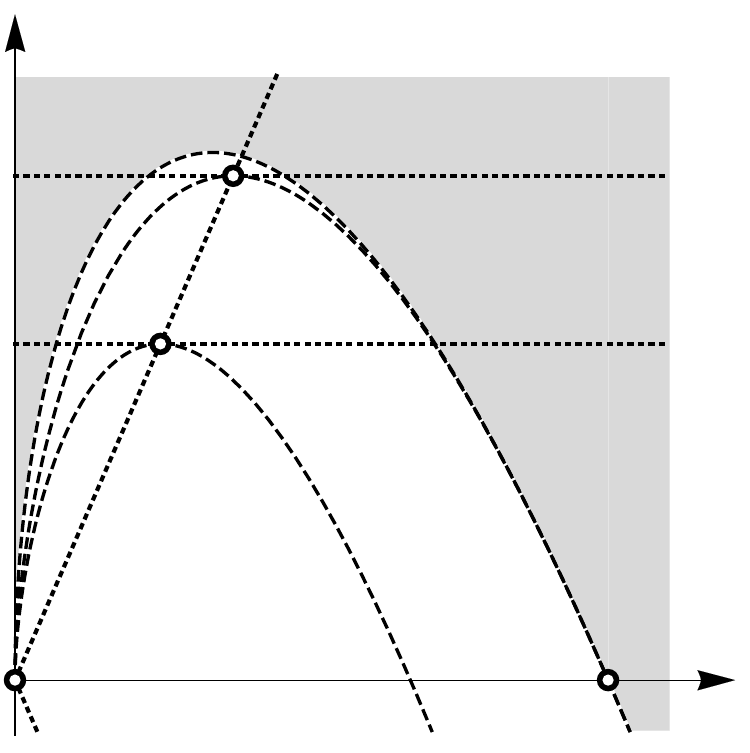}};
\node[style={anchor=center}] at (25,55) {\strut $h_{o}\in\{0,1\}$};
\node at (55,7) {\strut $\rho$};
\node at (0,50) {\strut $q$};
\node at (23,2) {\strut $\mathsf{BL}_1^{u_{*a}^{-}}$};
\draw[latex-] (14,39) --++ (6,-4) node[right] {\strut $\mathsf{BL}_1^{u_{*a}^{+}}$};
\draw[latex-] (12,41) --++ (-9,6) node[left] {\strut $\mathsf{BL}_1^{u_{*0}^{+}}$};
\node at (14,26) {\strut $u_{*a}^{-}$};
\node at (19,37.5) {\strut $u_{*a}^{+}$};
\node at (40,4) {\strut $u_{*0}^{+}$};
\node at (0,29) {\strut $q_{*}^{-}$};
\node at (0,41) {\strut $q_{*}^{+}$};
\end{scope}
\end{tikzpicture}}
\caption{The sets $\pi_{\ell}\left(\mathsf{CH}^{\scriptscriptstyle\complement}_{1}(q_*^-,q_*^+)\right)$, $\pi_{\ell}\left(\mathsf{CH}^{\scriptscriptstyle\complement}_{2}(q_*^-)\right)$, and $\pi_{\ell}\left(\mathsf{CH}^{\scriptscriptstyle\complement}_{3}(q_*^+)\right)$, respectively, given in \eqref{e:CH123}.}
\label{f:CH123}
\end{figure}

Now, we briefly motivate the representations in \figurename s~\ref{f:CH123} and~\ref{f:CH}.
Consider first $\mathsf{CH}^{\scriptscriptstyle\complement}_{1}(q_*^-,q_*^+)$.
By \eqref{e:CH1} we have $\mathsf{CH}^{\scriptscriptstyle\complement}_{1}(q_*^-,q_*^+) \subseteq \dom\times\{0\}$.
By \eqref{e:SleepToken}, we obtain
\[
u_{\ell} \in \pi_{\ell}\left(\mathsf{CH}^{\scriptscriptstyle\complement}_{1}(q_*^-,q_*^+)\right)
\iff
\overline{Q}\left(u_\ell\right) \in [q_*^-,q_*^+]
\iff
\left\{\begin{array}{@{}l@{\quad\hbox{if }}l@{}}
q_{\ell} \in [q_*^-,q_*^+] & v(u_{\ell}) \geq a,
\\
\overline{q}(u_{\ell}) \in [q_*^-,q_*^+] & v(u_{\ell}) < a.
\end{array}\right.
\]
As a consequence, on the one hand we have
\[
\pi_{\ell}\left(\mathsf{CH}^{\scriptscriptstyle\complement}_{1}(q_*^-,q_*^+)\right) \cap \{u \in \Omega : v(u) \geq a\} =
\{u \in \Omega : v(u) \geq a,\ q \in  [q_*^-,q_*^+]\}.
\]
On the other hand, observe that if $v(u_{\ell}) < a$ then the state on $\mathsf{FL}_1^{u_{\ell}}$ with maximum $q$-coordinate lies on the sonic line $q=a\,\rho$.
Hence, if $u_{\ell} \in \pi_{\ell}\left(\mathsf{CH}^{\scriptscriptstyle\complement}_{1}(q_*^-,q_*^+)\right)$ and $v(u_{\ell}) < a$, then $\overline{u}(u_{\ell})$ lies on the portion of the sonic line $q=a\,\rho \in [q_*^-,q_*^+]$, and therefore
\[\pi_{\ell}\left(\mathsf{CH}^{\scriptscriptstyle\complement}_{1}(q_*^-,q_*^+)\right) \cap \{u \in \Omega : v(u) < a\} = \bigcup_{q\in[q_*^-,q_*^+]} \left\{ u\in\mathsf{BL}_{1}^{(q/a,q)} : v(u)<a\right\}.
\]
We deduce that $\pi_{\ell}\left(\mathsf{CH}^{\scriptscriptstyle\complement}_{1}(q_*^-,q_*^+)\right)$ is as in \figurename~\ref{f:CH123}, left, where $u_{*a}^{\pm}$ are the intersections of the sonic line $q=a\,\rho$ with the horizontal lines $q=q_{*}^{\pm}$, respectively, namely $u_{*a}^{\pm} \doteq \left(q_{*}^{\pm}/a,q_{*}^{\pm}\right)$.
The representation of $\mathsf{CH}^{\scriptscriptstyle\complement}_{2}(q_*^-)$ follows from analogous arguments, see \eqref{e:CH2}.

We consider now $\mathsf{CH}^{\scriptscriptstyle\complement}_{3}(q_*^+)$ given in \eqref{e:CH3}.
Clearly, we have $v\left(\hat{u}(0,u_{\ell})\right) = 0 < a$; therefore, by \eqref{e:SleepToken}, if $u_{\ell} \in \pi_{\ell}\left(\mathsf{CH}^{\scriptscriptstyle\complement}_{3}(q_*^+)\right)$ then $\overline{u}(\hat{u}(0,u_{\ell}))$ lies on the portion of the sonic line $q=a\,\rho > q_*^+$, namely,
\[u_{\ell} \in \pi_{\ell}\left(\mathsf{CH}^{\scriptscriptstyle\complement}_{3}(q_*^+)\right)
\iff
\hat{u}(0,u_{\ell}) \in \bigcup_{q>q_*^+} \left\{ (\rho,q)\in\mathsf{BL}_{1}^{(q/a,q)} : q=0\right\}
\iff
u_{\ell} \in \bigcup_{\rho>\rho_{*0}^+} \mathsf{BL}_{1}^{(\rho,0)},
\]
where we recall that $u_{*0}^+$ is the unique element of $\mathsf{BL}_{1}^{u_{*a}^+}$ with $q=0$.
At last, the above considerations lead to the representation of $\pi_{\ell}\left(\mathsf{CH}^{\scriptscriptstyle\complement}(q_*^-,q_*^+)\right)$ as in \figurename~\ref{f:CH} and to \eqref{e:incodom2}.
\end{proof}

\begin{rmk}
\label{r:PorcupineTree}
The valve described in Section~\ref{s:descr} is characterized by the three conditions in \eqref{e:Moodring1} which, by Proposition~\ref{p:equi}, are equivalent to \eqref{e:Moodring2} if $(u,h)\in\BV_{\rm loc}([0, +\infty) \times \R;\Omega) \times \L\infty([0, +\infty);\{0,1\})$ is a global c-solution of the Cauchy problem~\eqref{e:csystem}.
Condition \eqref{e:Moodring2} may fail, for instance, in the case of Riemann initial data $(u_\ell,u_r,h_o) \in \mathsf{CH}^{\scriptscriptstyle\complement}$, see \eqref{e:coherent}\textsubscript{$2$}.
As a consequence, $\left(u(t,x),h(t)\right) \doteq \rsc[u_\ell,u_r,h_o](x/t)$ satisfies \eqref{e:Moodring1} globally in time if and only if $(u_\ell,u_r,h_o) \in \mathsf{CH}$.
In other words, $(u,h)$ corresponds to the valve without chattering described in Section~\ref{s:descr} if and only if $(u_\ell,u_r,h_o) \in \mathsf{CH}$.
On the other hand, if $(u_\ell,u_r,h_o) \in \mathsf{CH}^{\scriptscriptstyle\complement}$, then $(u,h)$ corresponds to such a valve with chattering. In fact, for this choice of initial datum the valve undergoes chattering and the regularity required in \eqref{e:uhreg}\textsubscript{$1$} is no longer satisfied.
For completeness, let us stress that $(u,h)$ is a global c-solution of the Cauchy problem \eqref{e:csystem} in the sense of Definition~\ref{d:csol} for any $(u_\ell,u_r,h_o) \in \dom$.
\end{rmk}

\begin{rmk}
Fix $(u_\ell,u_r,h_o)$ in $\mathsf{CH}^{\scriptscriptstyle\complement}$, see \eqref{e:incodom}.
As already observed, in this case we have $\hc\left(\ucm(u_\ell,h_o),\hc(u_\ell,h_o)\right) = 1$ and then by \eqref{e:hc}, \eqref{e:qc} and \eqref{e:ucpm} we deduce
\begin{align*}
&\ucm(u_\ell,h_o) = \hat{u}(0,u_{\ell}),&
&\qc\left(\ucm(u_\ell,h_o),\hc(u_\ell,h_o)\right) = q_{*},&
&\ucm\left(\ucm(u_\ell,h_o),\hc(u_\ell,h_o)\right) = \hat{u}\left(q_{*},\hat{u}(0,u_{\ell})\right).
\end{align*}
In particular, the second condition above ensures that $\left(\ucm(u_\ell,h_o),\ucp(u_\ell,u_r,h_o),\hc(u_\ell,h_o)\right)$ belongs to $\mathsf{CH}$, see \eqref{e:incodom}.
As a consequence, by Theorem~\ref{t:coherence}, this readily implies that $\rsc$ is coherent at $C\left(C(u_\ell,u_{r},h_o)\right)=(C\circ C)\left(u_\ell,u_{r},h_o\right)$.
This may suggest considering the c-Riemann solver associated with $C\circ C$ and conjecturing that it is coherent.
However, $C\circ C$ is not a c-function.
In fact, if $q_{\ell} < q_{*}$, then $\ucm\left(\ucm(u_\ell,h_o),\hc(u_\ell,h_o)\right)$ belongs to $\Gamma^-(\ucm(u_\ell,h_o))$ but not to $\Gamma^-(u_\ell)$ because $\qc\left(\ucm(u_\ell,h_o),\hc(u_\ell,h_o)\right) = q_{*} > q_{\ell} = \overline{Q}(u_{\ell})$, see \figurename~\ref{f:coherence} and \eqref{e:Periphery}.
\end{rmk}

In the following proposition we consider the two disjoined components $\mathsf{CH}_{\ell,0}^{\scriptscriptstyle\complement}(q_*^+) \times \Omega \times \{0\}$ and $\mathsf{CH}_{\ell,1}^{\scriptscriptstyle\complement}(q_*^-,q_*^+)\times \Omega \times \{1\}$ of $\mathsf{CH}^{\scriptscriptstyle\complement}(q_{*}^-,q_{*}^+)$, see \eqref{e:NIN}, and study their areas as functions of $q_*^-$ and $q_*^+$.
For a given subset $\mathsf{C}_\ell$ of $\Omega$, we denote by $|\mathsf{C}_\ell|$ its area. The proof includes the explicit expressions for the constants $C_1$, $C_2$, and the function $A$, which are given below.

\begin{prop}
\label{p:area}
There exist constants $C_1,C_2>0$, independent of $q_*^\pm$, and
a function $A$ such that
\begin{equation}\label{e:keeel}
\left|\mathsf{CH}_{\ell,0}^{\scriptscriptstyle\complement}(q_*^+)\right| = \frac{C_1}{a} \, (q_*^+)^2 \quad \hbox{ and }\quad
\left|\mathsf{CH}_{\ell,1}^{\scriptscriptstyle\complement}(q_*^-,q_*^+)\right| =
\frac{C_2}{a} \, A\left(\frac{q_*^-}{q_*^+}\right)(q_*^+)^2.
\end{equation}
The function $A$ is strictly increasing, strictly convex and $A(0)=0$, see \eqref{e:A} for its explicit expression; its graph is plotted in \figurename~\ref{f:A}.
In particular, the map $[0,q_*] \ni q_{*}^- \mapsto \left|\mathsf{CH}_{\ell,1}^{\scriptscriptstyle\complement}(q_*^-,q_*^+)\right|$ is strictly increasing, the map $[q_*,+\infty) \ni q_{*}^+ \mapsto \left|\mathsf{CH}_{\ell,1}^{\scriptscriptstyle\complement}(q_*^-,q_*^+)\right|$ is strictly decreasing and
\begin{align}
&\lim_{q_{*}^- \to 0} \left|\mathsf{CH}_{\ell,1}^{\scriptscriptstyle\complement}(q_*^-,q_*^+)\right| = 0,&
&\lim_{q_{*}^+ \to +\infty} \left|\mathsf{CH}_{\ell,1}^{\scriptscriptstyle\complement}(q_*^-,q_*^+)\right| = 0.
\label{e:kkwj}
\end{align}
\end{prop}

\begin{figure}[!htb]\centering
\begin{tikzpicture}[every node/.style={anchor=south west,inner sep=0pt},x=1mm, y=1mm]
\node at (4,4) {\includegraphics[width=50mm]{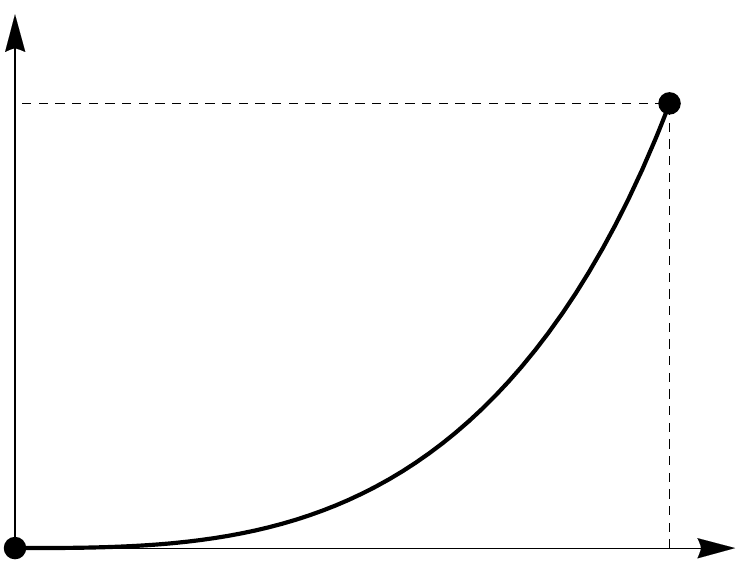}};
\node at (52,0) {\strut $z$};
\node[left] at (3.5,41) {\strut $A$};
\node at (47,0) {\strut $1$};
\node[left] at (3.5,35) {\strut $0.2388$};
\end{tikzpicture}
\caption{The graph of function $A$ in \eqref{e:keeel}\textsubscript{$2$}.}
\label{f:A}
\end{figure}
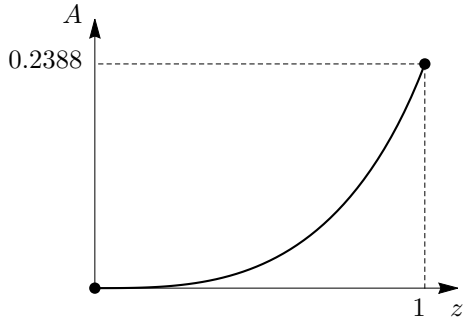

\begin{proof}
We first consider $\mathsf{CH}_{\ell,0}^{\scriptscriptstyle\complement}(q_*^+)$.
By definition of $u_{*a}^{+}$ and \eqref{e:laxRS}\textsubscript{$1$} we have $u_{*0}^+ = (\frac{e}{a}\,q_{*}^{+},0)$ (see \figurename~\ref{f:CH}), because
\begin{equation*}
\{u_{*0}^+\} = \{u\in \mathsf{BL}_{1}^{u_{*a}^+} : q=0\}
\Longrightarrow
a\left(1 - \ln\left(\frac{a\,\rho_{*0}^+}{q_{*}^{+}}\right) \right) \rho_{*0}^+ = 0
\Longrightarrow
\rho_{*0}^+ = \frac{e}{a}\,q_{*}^{+}.
\end{equation*}
Recall that $u_{*a0}^{+}$ is the state with the smallest $\rho$-coordinate in the intersection of $\mathsf{BL}_{1}^{u_{*0}^+}$ and $q=q_{*}^{+}$, see \figurename~\ref{f:CH}, on the left.
By \eqref{e:laxRS}\textsubscript{$2$} we have then
\[
- a \left(\sqrt{\dfrac{\rho_{*a0}^{+}}{\rho_{*0}^+}} - \sqrt{\dfrac{\rho_{*0}^+}{\rho_{*a0}^{+}}}\right) \rho_{*a0}^{+} = q_{*}^{+}
\Longrightarrow
\sqrt{e} = \sqrt{\dfrac{a}{v(u_{*a0}^{+})}}  \left( e - \frac{a}{v(u_{*a0}^{+})} \right),
\]
because $q_{*a0}^{+} = q_{*}^{+}$ and $\rho_{*0}^+ = \frac{e}{a}\,q_{*}^{+}$.
The rightmost formula above, together with the definition of $u_{*a0}^{+}$, implies $v(u_{*a0}^{+})/a \doteq \kappa \approx 1.6294$.
We compute now the area of $\mathsf{CH}_{\ell,0}^{\scriptscriptstyle\complement}(q_*^+)$:
\[
\left| \mathsf{CH}_{\ell,0}^{\scriptscriptstyle\complement}(q_*^+) \right| =
\rho_{*a0}^{+} \, q_*^+ - \int_{0}^{\rho_{*a0}^{+}} \mathcal{BL}_{1}^{u_{*0}^+}(\rho) \,\d \rho =
\rho_{*a0}^{+} \, q_*^+ +  a\int_{0}^{\rho_{*a0}^{+}} \left(\sqrt{\dfrac{\rho \vphantom{\rho_{*0}^+} }{\rho_{*0}^+}} - \sqrt{\dfrac{\rho_{*0}^+}{\rho}}\right) \rho \,\d \rho =
\frac{C_{1}}{a} \, (q_*^+)^2,
\]
where $C_{1} \doteq \frac{1}{\kappa}-\frac{2 \sqrt{e}}{3\kappa^{3/2}}+\frac{2}{5\sqrt{e} \, \kappa^{5/2}} \approx 0.1568$ because $\rho_{*a0}^{+}=\frac{q_*^+}{a \, \kappa}$. Formula \eqref{e:keeel}\textsubscript{$1$} follows at once.

Consider $\mathsf{CH}_{\ell,1}^{\scriptscriptstyle\complement}(q_*^-,q_*^+)$.
First, formula \eqref{e:kkwj}\textsubscript{$1$} is straightforward, see \figurename~\ref{f:CH} on the right.
We now prove \eqref{e:keeel}\textsubscript{$2$}.
As a first step, we compute $\rho_{*a0}^-$.
By definition of $u_{*a0}^-$ and \eqref{e:laxRS}\textsubscript{$2$}, we have that $\rho_{*a0}^-$ is the smallest solution $\rho$ in $[0,\rho_{*0}^+] = [0,e\,q_*^+/a]$ of the equation
\[q_*^- = \mathcal{S}_1^{u_{*0}^+}(\rho) = - a \left(\sqrt{\frac{a\,\rho}{e\,q_*^+}} - \sqrt{\frac{e\,q_*^+}{a\,\rho}}\right) \rho.\]
By introducing $r \doteq \frac{a\,\rho}{e\,q_*^+} \in [0,1]$ and $c \doteq \frac{q_*^-}{e\,q_*^+} \in \left(0,\frac{1}{e}\right)$, the above equation becomes
\[r\left(\sqrt{1/r}-\sqrt{r}\right) = c.\]
Observe that the left-hand side of the above equation is non-negative since $r\in[0,1]$.
By squaring such equation we get
\[r\left(r-1\right)^2=c^2.\]
By studying the map $\R \ni r \mapsto r\left(r-1\right)^2$, we deduce that the above equation admits three solutions in $[0,+\infty)$ because $c^2 \in (0, 4/27)$.
We set $s\doteq r-\frac{2}{3} \in \left[-\frac{2}{3},\frac{1}{3}\right]$ and get the equation
\begin{equation}
\label{e:eqs}
s^3 - \frac{1}{3} \, s + \frac{2}{27} - c^2 = 0.
\end{equation}
We recall that the solutions of the equation $x^3+\alpha\,x+\beta=0$, for $\alpha,\beta\in\R$, are given by Cardano's formula 
\[
x = \sqrt[3]{-\frac{\beta}{2}+\sqrt{\frac{\beta^2}{4} + \frac{\alpha^3}{27}}} + \sqrt[3]{-\frac{\beta}{2}-\sqrt{\frac{\beta^2}{4} + \frac{\alpha^3}{27}}},
\]
where roots are meant in the complex sense. About \eqref{e:eqs}, this yields
\[s = \sqrt[3]{\frac{c^2}{2}-\frac{1}{27} +  {\rm i} \, \frac{c}{2} \sqrt{\frac{4}{27}-c^2}} + \sqrt[3]{\frac{c^2}{2}-\frac{1}{27} -  {\rm i} \, \frac{c}{2} \sqrt{\frac{4}{27}-c^2}}.\]
Observe that the module of the complex numbers under cubic roots above is
\[\sqrt{\left(\frac{c^2}{2}-\frac{1}{27}\right)^2 + \left(\frac{c}{2} \sqrt{\frac{4}{27}-c^2}\right)^2} = \frac{1}{27},
\]
and therefore their arguments are $\pm\theta$, where
\[\theta \doteq \acos\left(\frac{27}{2} \, c^2 - 1\right).\]
The three real solutions of equation \eqref{e:eqs} are then
\begin{align*}
&s_1 \doteq \frac{2}{3} \cos \left(\frac{\theta}{3}\right),&
&s_2 \doteq \frac{2}{3} \cos \left(\frac{\theta}{3}+\frac{2 \pi }{3}\right),&
&s_3 \doteq \frac{2}{3} \cos \left(\frac{\theta}{3}+\frac{4 \pi }{3}\right).
\end{align*}
By observing that $\theta \in \left(\acos\left(\frac{27}{2e^2}-1\right),\pi\right)$ one checks that $s_2 < s_3 < s_1$.
The above arguments show that
\begin{align*}
\rho_{*a0}^- &= \frac{e\,q_*^+}{a} \, r_2,&
r_2 &\doteq \frac{2}{3} \left( 1 + \cos \left(\frac{\theta}{3}+\frac{2 \pi }{3}\right) \right).
\end{align*}
We are now in the position to compute the area of $\mathsf{CH}_{\ell,1}^{\scriptscriptstyle\complement}(q_*^-,q_*^+)$:
\begin{align*}
&\left|\mathsf{CH}_{\ell,1}^{\scriptscriptstyle\complement}(q_*^-,q_*^+)\right| = 
\rho_{*a0}^- \, q_*^- - \int_0^{\rho_{*a0}^-} \mathcal{S}_1^{u_{*0}^+}(\rho) \, \d \rho =
\frac{e \, q_*^+ r_2 \left(15 \, q_*^- + 2 e \, q_*^+ (3 r_2-5) \sqrt{r_2}\right)}{15 a} 
\\={}&
\frac{4 e^2}{45 a} \, (q_*^+)^2 \sin\left(\frac{\acos\left(1-\frac{27 c^2}{2}\right)}{6}\right)^2 \left[\frac{16}{\sqrt{3}} \sin\left(\frac{\acos\left(1-\frac{27 c^2}{2}\right)}{6}\right)^3 - \frac{20}{\sqrt{3}} \sin \left(\frac{\acos\left(1-\frac{27 c^2}{2}\right)}{6}\right) + 15 c\right].
\end{align*}
The second equality above can be directly proved or by using a symbolic software. Then we prove \eqref{e:keeel}\textsubscript{$2$} for $C_2\doteq\frac{4 e^2}{45}$ and 
\begin{align}
\label{e:A}
A(z) &\doteq S(z)^2 \left[\frac{16}{\sqrt{3}} \, S(z)^3 - \frac{20}{\sqrt{3}} \, S(z) + \frac{15}{e} \, z \right],&
S(z) &\doteq \sin\left(\frac{\acos\left(1-\frac{27 z^2}{2e^2}\right)}{6}\right).
\end{align}
The fact that $A$ is strictly increasing, strictly convex and $A(0)=0$ can be proved by tedious calculations or, better, by exploiting a symbolic software.
Again by direct computations, or by using a symbolic software, we deduce
\begin{align*}
\partial_{q_*^-} \left|\mathsf{CH}_{\ell,1}^{\scriptscriptstyle\complement}(q_*^-,q_*^+)\right| &=
\frac{4 e}{3 a} \, q_*^+ \, S\left(\frac{q_*^-}{q_*^+}\right)^2, \quad \partial_{q_*^+} \left|\mathsf{CH}_{\ell,1}^{\scriptscriptstyle\complement}(q_*^-,q_*^+)\right| =
-\frac{4 e}{135 a} \, q_*^+ \, S\left(\frac{q_*^-}{q_*^+}\right)^2 \left( 16 \sqrt{3} \, e \, S\left(\frac{q_*^-}{q_*^+}\right) - 9 \, \frac{q_*^-}{q_*^+} \right),
\end{align*}
that imply the monotonicity properties of $\left|\mathsf{CH}_{\ell,1}^{\scriptscriptstyle\complement}(q_*^-,q_*^+)\right|$ because $S(x)$ is proved to be strictly increasing, $S(0)=0$ and $S'(x) \geq \frac{\sqrt{3}}{2 e} > \frac{9}{16 \sqrt{3} e}$.
Notice that from \eqref{e:keeel}\textsubscript{$2$} we find again \eqref{e:kkwj}\textsubscript{$2$} because $S(0)=0$.
This concludes the proof.
\end{proof}

\begin{rmk}
The total area of the incoherence domain $\mathsf{CH}^{\scriptscriptstyle\complement}$ is 
\[
\mathcal{A}(q_*^-,q_*^+) \doteq \frac{C_1}{a} \, (q_*^+)^2+
\frac{C_2}{a} \, A\left(\frac{q_*^-}{q_*^+}\right)(q_*^+)^2.
\]
The function $A=A(z)$ can be extended to $0$ as a $C^1$ function and $A'(0)=0$. With this remark we have
\[
a \, \nabla\mathcal{A}(q_*^-,q_*^+)= \left(C_2 A'\left(\frac{q_*^-}{q_*^+}\right)q_*^+, 2C_1q_*^+ - C_2A'\left(\frac{q_*^-}{q_*^+}\right)q_*^- + 2C_2A\left(\frac{q_*^-}{q_*^+}\right)q_*^+\right).
\]
The function $\mathcal{A}$ has no critical points in the interior of $[0,q_*] \times [q_*,+\infty)$. Therefore, the problem of minimizing the incoherence domain with respect to $(q_*^-,q_*^+)$ only has solutions lying on the boundary. 
Also the problem of minimizing $[0,q_{*}]\ni q\mapsto\mathcal{A}(q,q_*^+) $, for $q_*^+ \in [q_*,+\infty)$ fixed, only has solutions on the boundary, more precisely at $q=0$. 
At last, also the problem of minimizing $[q_{*},+\infty)\ni q\mapsto\mathcal{A}(q_*^-,q) $ for $q_*^- \in [0,q_*]$ fixed only has solutions on the boundary, more precisely at $q=q_{*}$.
Indeed, it is sufficient to observe that $C_2A'(z)z < 2C_1+2C_2A(z)$ for any $z=q_*^-/q \in (0,1]$.
\end{rmk}

By Proposition~\ref{p:area}, for any fixed $q_*>0$, the pair $(q_*^-,q_*^+)=(0,q_*)$ minimizes the area of the incoherent domain $\mathsf{CH}^{\scriptscriptstyle\complement}$ among all $(q_*^-,q_*^+) \in [0,q_*] \times [q_*,+\infty)$.
However, the choice $q_*^-=0$ implies that the flow through the valve may take any value in the interval $[0,q_*]$. Consequently, the resulting flow at $x=0$ may, in general, differ significantly from the desired one.

\begin{rmk}
Recall that $\kappa \doteq v(u_{*a0}^{+})/a$, see the proof of Proposition~\ref{p:area}. 
This means that the points $u_{*a0}^{+}$ lie on the straight line $q=a\,\kappa\,\rho$ for every $q_*^+$.
As a consequence, see \figurename~\ref{f:CH}, we have
\[\mathsf{CH}_{\ell,1}^{\scriptscriptstyle\complement}(q_*^-,q_*^+) \subset \mathsf{CH}_{\ell,0}^{\scriptscriptstyle\complement}(q_*^+) \subset \{(\rho,q)\in\Omega : v(\rho,q) > a\, \kappa ,\ q \leq q_*^+\}.\]
Moreover, if $q_*^\pm=q_*$ then $\mathsf{CH}_{\ell,0}^{\scriptscriptstyle\complement}(q_*^+) = \mathsf{CH}_{\ell,1}^{\scriptscriptstyle\complement}(q_*^-,q_*^+)$ and they coincide with the incoherent domain studied in \cite[Theorem~2.10]{C1} for the valve without hysteresis.
\end{rmk}

\section{Numerical simulations}
\label{s:num}
\setcounter{equation}{0}

In this section, first we describe, in Subsection~\ref{s:num1}, the numerical scheme we are going to use in the following.
Subsection~\ref{s:num2} concerns some examples to check the convergence of the scheme by comparing the numerical with the exact solution of some Riemann problems.
We also study the convergence rate of the scheme.
In Subsection~\ref{s:num3} we show some numerical simulations for incoherent Riemann data and discuss the onset of chattering.
In the last set of examples, see Subsection~\ref{s:num4}, we consider a strongly oscillatory inflow and discuss the valve behavior depending on $q_{*}^{-}$ and $q_{*}^{+}$.
There, we also discuss how much the flow of the solution at $x=0$ is close to the desired value $q_*$.

\subsection{Description of the numerical scheme}
\label{s:num1}

We describe the numerical scheme used to approximate the solutions provided by the c-Riemann solver $\rsc[u_\ell,u_r,h_o]$, see Subsection~\ref{subsec:cRiemannValve}. The method is based on the Random Choice Method (RCM) and it is inspired by the one introduced in \cite{C1}.

Let $\Delta x$ and $\Delta t_{n+1}$ denote the space and (adaptive) time steps, respectively. We introduce the grid points $x_{j+1/2} \doteq j\,\Delta x$, the cells $K_j \doteq [x_{j-1/2}, x_{j+1/2}),$ and their centers $ x_j \doteq \left(j-1/2\right)\Delta x,$ for $j \in \mathbb{Z}.$ Moreover, we denote by $j_{\rm c}$ the index such that $x_{j_{\rm c}+1/2}$ is the location of the valve. We define the time grid recursively
\begin{align*}
&t^0=0,&
&t^{n+1}=t^n+\Delta t_{n+1}.
\end{align*}
Above, the time step $\Delta t_{n+1}$ is chosen adaptively so that the following CFL condition is satisfied:
\[
\Delta t_{n+1} = \frac{C_{\rm CFL} \Delta x}{\max\limits_{j \in \mathbb{Z}} \max\limits_{i\in\{1,2\}} |\lambda_i(u_j^{n})|},
\]
with $0 < C_{\rm CFL} < 1/2$; in our simulations we take $C_{\rm CFL} = 0.45$.
For a given time horizon $T>0$, set $N \doteq \min\{n\in\N : t^{n} > T\}$.

We now introduce the piecewise constant approximate solutions $(u_\Delta,h_\Delta)$. 
We compute the c-function $C \colon \dom \times \{0,1\} \to \dom \times \{0,1\}$ numerically. Let $
C_\Delta \doteq \left(u_{c \,\Delta}^-, u_{c \,\Delta}^+, \hc\right)$ be the numerical approximation of $C = (\ucm,\ucp,\hc),$ obtained through the following steps: given $(u_\ell,u_r,h_o)$
\begin{enumerate}
\item Compute $\overline{Q}(u_\ell)$ from \eqref{e:SleepToken};
\item Determine $\hc(u_\ell,h_{o}) \in \{0,1\}$ from \eqref{e:hc};
\item Compute $\qc(u_\ell,h_{o})$ from \eqref{e:qc};
\item Use a numerical approximation of \eqref{e:ucpm} to obtain $u_{c \,\Delta}^-$ and $u_{c\, \Delta}^+$.
\end{enumerate}
Denote by $\left(u_j^n,h^n\right)$ the constant value attained by $(u_\Delta,h_\Delta)$ on each time-space cells $K_j^n \doteq [t^n,t^{n+1}) \times K_j$; it will be defined iteratively below.
Given the initial datum $(u_o,h_o)$, we set
\begin{align*}
u_j^0 &\doteq \frac{1}{\Delta x} \int_{K_j} u_o(x)\,\d x,&
h^0 &\doteq h_o.
\end{align*}
Assuming that $u_j^n$ for all $j \in \mathbb{Z}$ and $h^n$ are known, we compute $u_j^{n+1}$ and $h^{n+1}$ as follows:

\begin{itemize}
\item We first generate a (quasi-)random number $\theta^n \in (0,1)$ using the base-2 van der Corput sequence:
\[
\theta^n \doteq \sum_{k=0}^m i_k\,2^{-(k+1)},
\]
where $i_k\in\{0,1\}$ are the coefficients of the binary expansion of $n$, that is $n = \sum_{k=0}^m i_k 2^{k}$.
\item For $j\notin \{j_{\rm c},j_{c+1}\}$ we set
\[
u_j^{n+1}\doteq\begin{cases}
\rss[u_{j-1}^n,u_j^n]\left(\theta^n\dfrac{\Delta x}{\Delta t_{n+1}} \right) &\hbox{if } 0\leq \theta^n \leq 1/2,
\\[7pt]
\rss[u_{j}^n,u_{j+1}^n]\left((\theta^n-1)\dfrac{\Delta x}{\Delta t_{n+1}} \right) &\hbox{if } 1/2< \theta^n \leq 1.
\end{cases}
\]
\item For $j \in \{j_{\rm c},j_{c+1}\}$ we let
\begin{align*}
u_{j_{\rm c}}^{n+1} &\doteq\begin{cases}
\rss[u_{j_{\rm c}-1}^n,u_{j_{\rm c}}^n]\left(\theta^n\dfrac{\Delta x}{\Delta t_{n+1}} \right) &\hbox{if } 0\leq \theta^n \leq 1/2,
\\[7pt]
\rss[u_{j_{\rm c}}^n,u_{c \,\Delta}^-(u_{j_{\rm c}}^n,h^n)]\left((\theta^n-1)\dfrac{\Delta x}{\Delta t_{n+1}} \right) &\hbox{if } 1/2< \theta^n \leq 1,
\end{cases}
\\
u_{j_{c+1}}^{n+1} &\doteq\begin{cases}
\rss[u_{c \,\Delta}^+(u_{j_{\rm c}}^n,u^n_{j_{\rm c}+1},h^n),u_{j_{\rm c}+1}^n]\left(\theta^n\dfrac{\Delta x}{\Delta t_{n+1}} \right) &\hbox{if }  0\leq \theta^n \leq 1/2,
\\[7pt]
\rss[u_{j_{\rm c}+1}^n,u_{j_{\rm c}+2}^n]\left((\theta^n-1)\dfrac{\Delta x}{\Delta t_{n+1}} \right) &\hbox{if }  1/2< \theta^n \leq 1.
\end{cases}
\end{align*}
\item We update the valve variable as
\[
h^{n+1} \doteq \hc(u_{j_{\rm c}}^n,h^n).
\]
\end{itemize}

\subsection{Numerical simulations with Riemann data in \texorpdfstring{$\mathsf{CH}$}{} and convergence rate estimates} 
\label{s:num2}

In this subsection, we present some numerical simulations in the computational domain $I= [-1,1]$ for the c-Riemann solver $\rsc$ given in \eqref{e:InFlames} by using Riemann data in $\mathsf{CH}$. 

In particular, in Examples~\ref{ex:1} and~\ref{ex:2}, we consider the same left and right states $u_\ell$ and $u_r$, but with $h_o=1$ and $h_o=0$, respectively. For these specific initial data, the resulting solutions are different, highlighting the role played by the initial value $h_o$. In Example~\ref{ex:3}, we choose $u_\ell \in \mathsf{CH}_{\ell,0}^{\scriptscriptstyle\complement}(q_*^+) \setminus \mathsf{CH}_{\ell,1}^{\scriptscriptstyle\complement}(q_*^-,q_*^+)$ and $h_o=1$, so that the initial datum is in $\mathsf{CH}$, to show that also in this case the scheme correctly captures the solution generated by $\rsc$. On the contrary, as we will see in Example~\ref{ex:CHc3}, if we consider the same state $u_\ell$ but $h_o=0$, then the initial datum is not in $\mathsf{CH}$ and then the scheme \emph{fails} to capture the solution generated by $\rsc$. Finally, in Example~\ref{ex:4} we show a case when the solution does not depend on $h_o$. 

These examples cover all possible wave patterns, including rarefaction and shock waves along both the forward and backward Lax curves, as well as the c-shock located at the valve. In these simulations, we use the following parameters:
\begin{align}
\label{e:parameters}
a = 2,&
&q_* = 3,&
&q_*^-=2,&
&q_*^+ = 4,&
&\Delta x = 5 \times 10^{-4},&
&T =0.2.
\end{align}

We remark that, in the following figures, $\rho$ and $q$ denote the first two components of the Riemann solver $\rsc[u_\ell,u_r,h_o]$.

\begin{ex}[Initially open valve]\label{ex:1}
Consider the initial datum $u_\ell =(2.5,2)$, $u_r=(5,-2)$, and $h_o=1$, so that the valve is initially open. In this case $v(u_\ell)<a$, implying that $\overline{Q}(u_\ell) = \overline{q}(u_\ell) \approx 2.7441 \in [q_*^-,q_*^+]$. Therefore $\hc(u_\ell,h_o)=h_o=1$ and $\qc(u_\ell,h_o) = \overline{q}(u_\ell)$. 
It follows then that $\rsc[u_\ell,u_r]$ consists of a $1$-rarefaction connecting the left state $u_\ell$ to $\ucm \approx (1.3729,2.7441)$, a stationary c-wave connecting $\ucm$ to $\ucp \approx (7.3440,2.7441)$, and a $2$-shock connecting $\ucp$ to the right state $u_r$, see \figurename~\ref{fig:ex1}.
\end{ex}

\begin{figure}[!htb]
\centering

\begin{minipage}[c]{0.33\textwidth}
\centering
\includegraphics[trim=0cm 0cm 0cm 0.34cm, clip,width=0.95\linewidth]{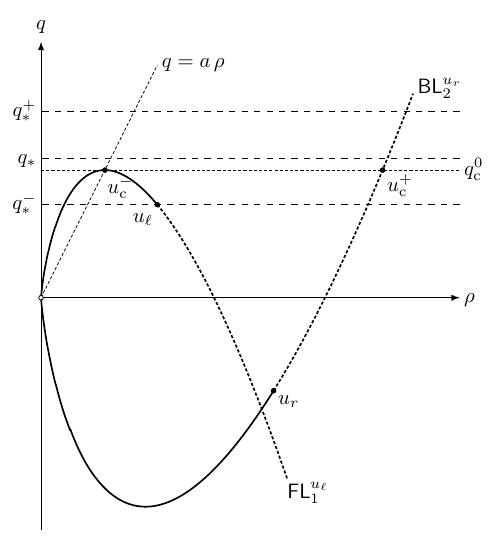}
\end{minipage}%
\hfill
\begin{minipage}[c]{0.33\textwidth}
\centering
\includegraphics[trim=2.5cm 8.3cm 1cm 8.3cm, clip, width=0.95\linewidth]{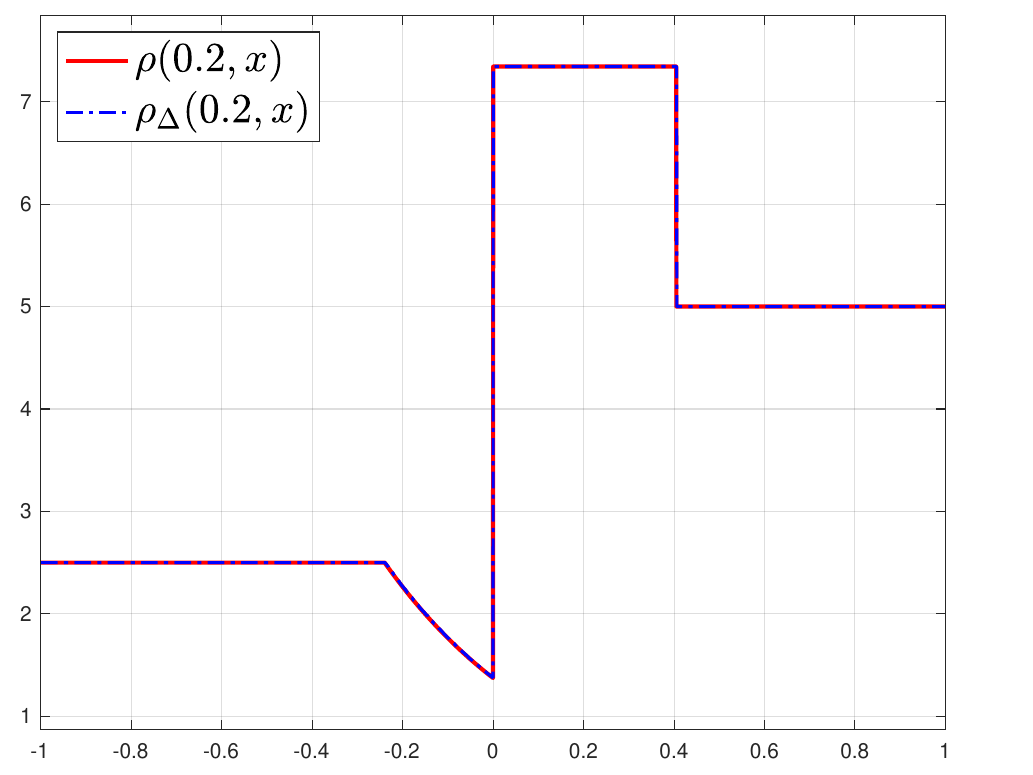}
\end{minipage}%
\hfill
\begin{minipage}[c]{0.33\textwidth}
\centering
\includegraphics[trim=2.5cm 8.3cm 1cm 8.3cm, clip, width=0.95\linewidth]{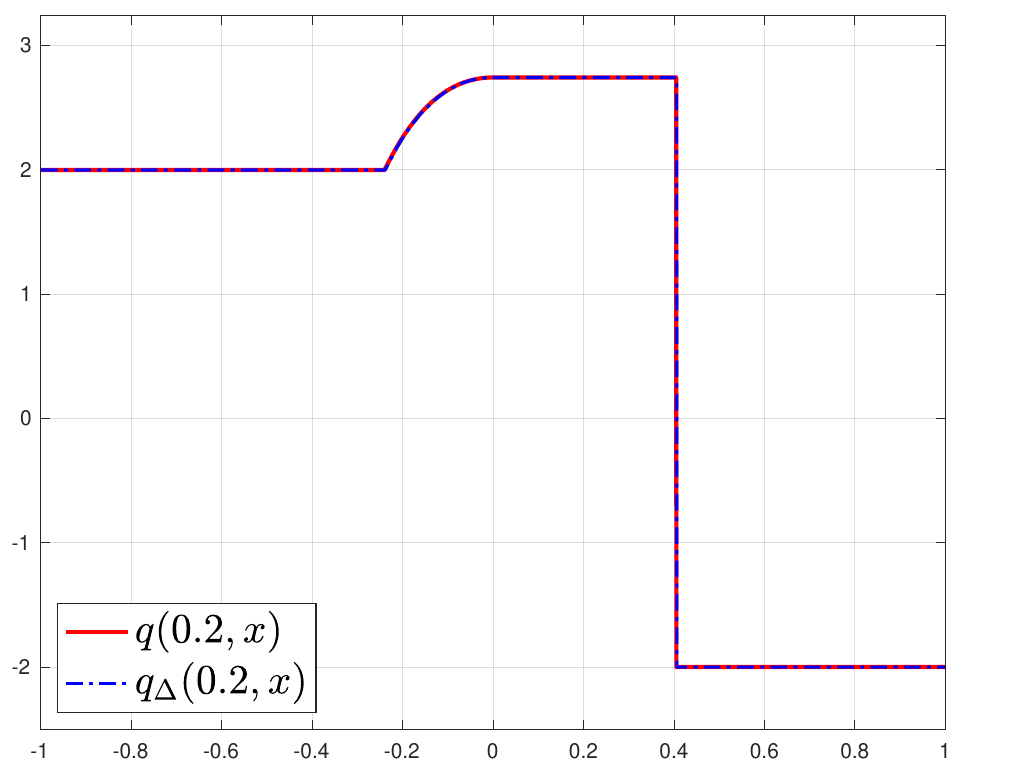} 
\end{minipage}

\caption{A solution with a $1$-rarefaction, a c-wave and a $2$-shock, see Example~\ref{ex:1}. Here and in the following three figures, the construction of the exact solution is depicted on the left, while the comparison between the numerical and exact solution is depicted on the middle and the right.}
\label{fig:ex1}
\end{figure}

\begin{ex}[Initially closed valve]\label{ex:2}
As in Example~\ref{ex:1}, consider the initial states $u_\ell =(2.5,2)$, $u_r=(5,-2)$, but now assume that the valve is initially closed, namely $h_o=0$. As already observed in the previous example, we have $v(u_\ell) < a$, implying that $\overline{Q}(u_\ell) = \overline{q}(u_\ell) \approx 2.7441 \in [q_*^-,q_*^+]$. Therefore, in this case we have $\hc(u_\ell,h_o)=h_o=0$ and $\qc(u_\ell,h_o) = 0$. 
It follows then that $\rsc[u_\ell,u_r]$ consists of a $1$-shock connecting the left state $u_\ell$ to $\ucm \approx (3.7198,0)$, a stationary c-wave connecting $\ucm$ to $\ucp \approx (6.1050,0)$, and a $2$-shock connecting $\ucp$ to the right state $u_r$, see \figurename~\ref{fig:ex2}.
\end{ex}

\begin{figure}[!htb]
\centering

\begin{minipage}[c]{0.33\textwidth}
\centering
\includegraphics[trim=0cm 0cm 0cm 0.34cm, clip,width=0.95\linewidth]{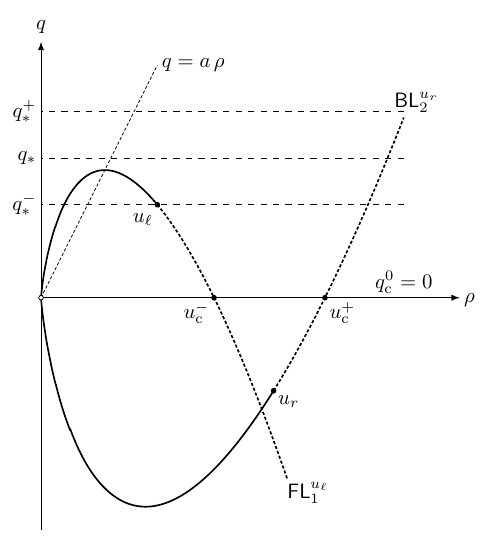}
\end{minipage}%
\hfill
\begin{minipage}{0.33\textwidth}
\centering
\includegraphics[trim=2.5cm 8.3cm 1cm 8.3cm, clip, width=0.95\linewidth]{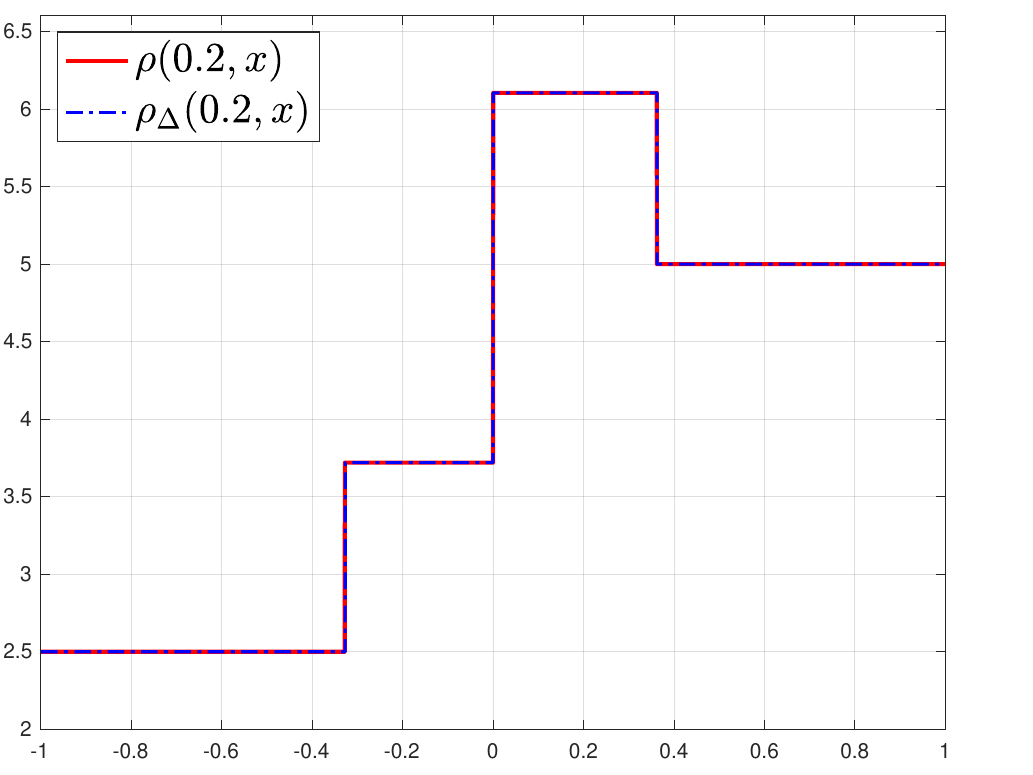}
\end{minipage}%
\hfill
\begin{minipage}{0.33\textwidth}
\includegraphics[trim=2.5cm 8.3cm 1cm 8.3cm, clip, width=0.95\linewidth]{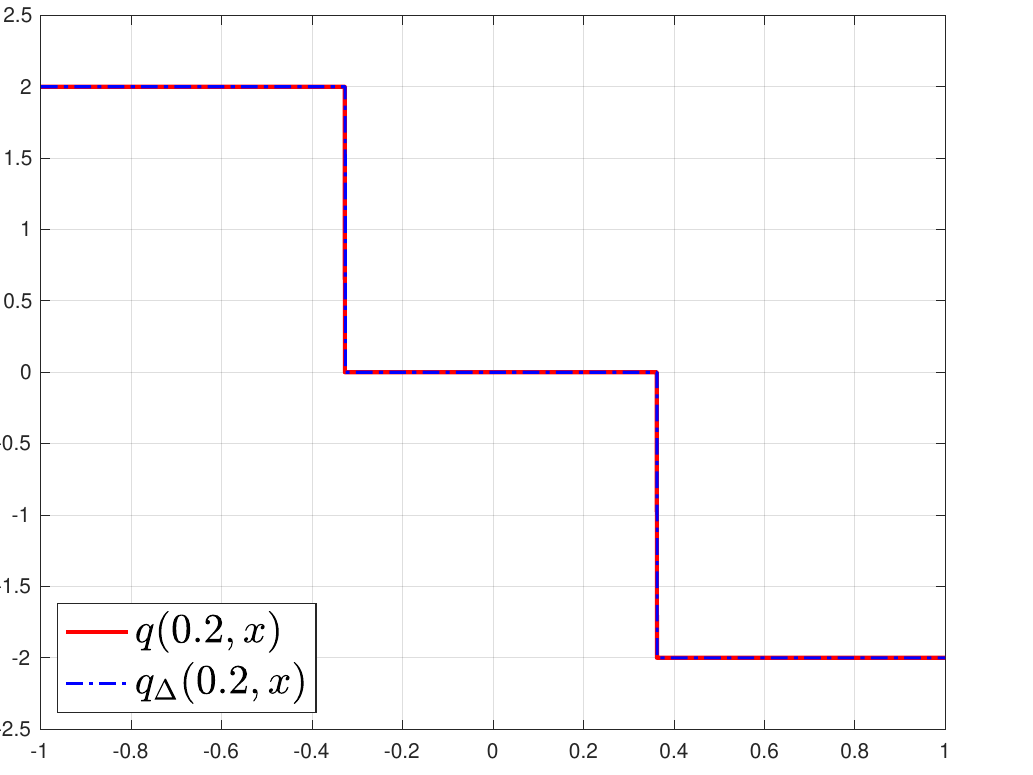}
\end{minipage}
\caption{A solution with a $1$-shock, a c-wave and a $2$-shock, see Example~\ref{ex:2}.}
\label{fig:ex2}
\end{figure}

\begin{ex}\label{ex:3}
Take the initial states $u_\ell=(0.25,2.5)$, $u_r = (2,3.5)$ and $h_o=1$. In this case $v(u_\ell)=10>a$, and then $\overline{Q}(u_\ell) = q_\ell = 2.5 \in [q_*^-,q_*^+].$ Therefore $h_c(u_\ell,h_o)=h_o=1$ and $q_c^0(u_\ell,h_o)= \overline{Q}(u_\ell) = q_\ell = 2.5$. Since again $v(u_\ell)=10>a$, then $\ucm$ coincides with $u_\ell$ because of \eqref{e:ucpm} and it is connected to $\ucp\approx(1.7226, 2.5)$ by a c-wave. In turn, $\ucp$ is connected to $u_r$ with a $2$-rarefaction, see \figurename~\ref{fig:ex3}. 
\end{ex}

\begin{figure}[!htb]
\centering

\begin{minipage}[c]{0.33\textwidth}
\centering
\includegraphics[trim=0cm 0cm 0cm 0cm, clip,width=0.95\linewidth]{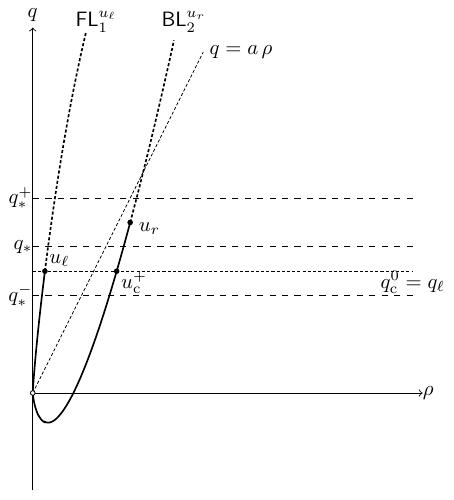}
\end{minipage}%
\hfill
\begin{minipage}{0.33\textwidth}
\centering
\includegraphics[trim=2.5cm 8.3cm 1cm 8.3cm, clip, width=0.95\linewidth]{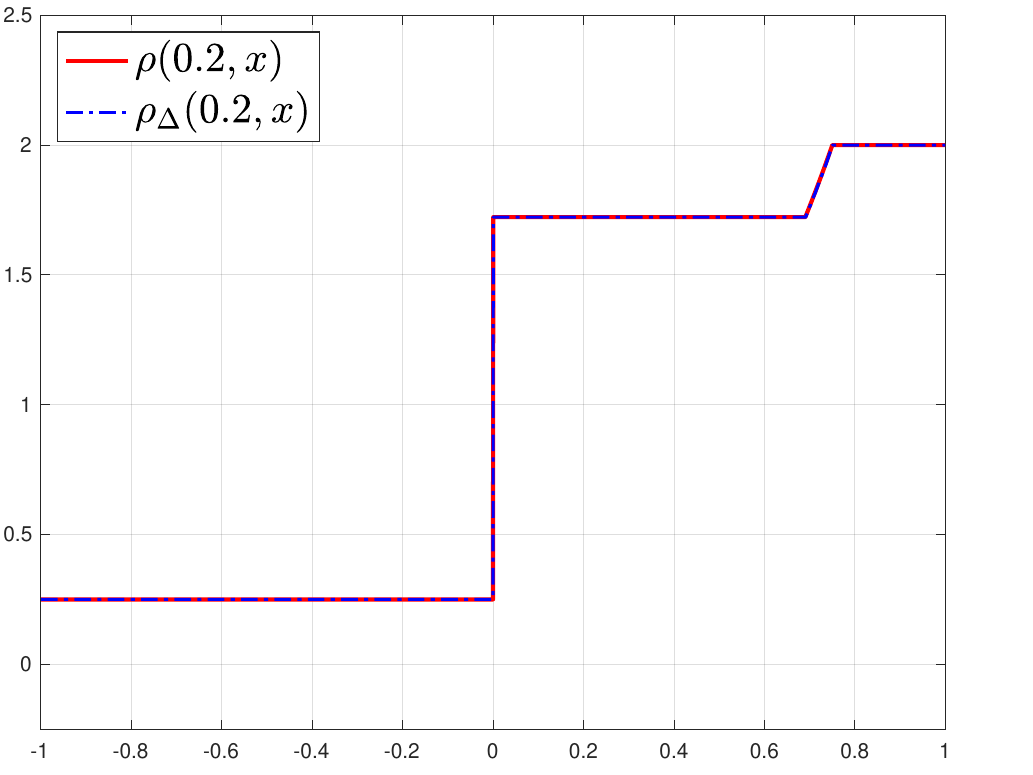}
\end{minipage}%
\hfill
\begin{minipage}{0.33\textwidth}
\includegraphics[trim=2.5cm 8.3cm 1cm 8.3cm, clip, width=0.95\linewidth]{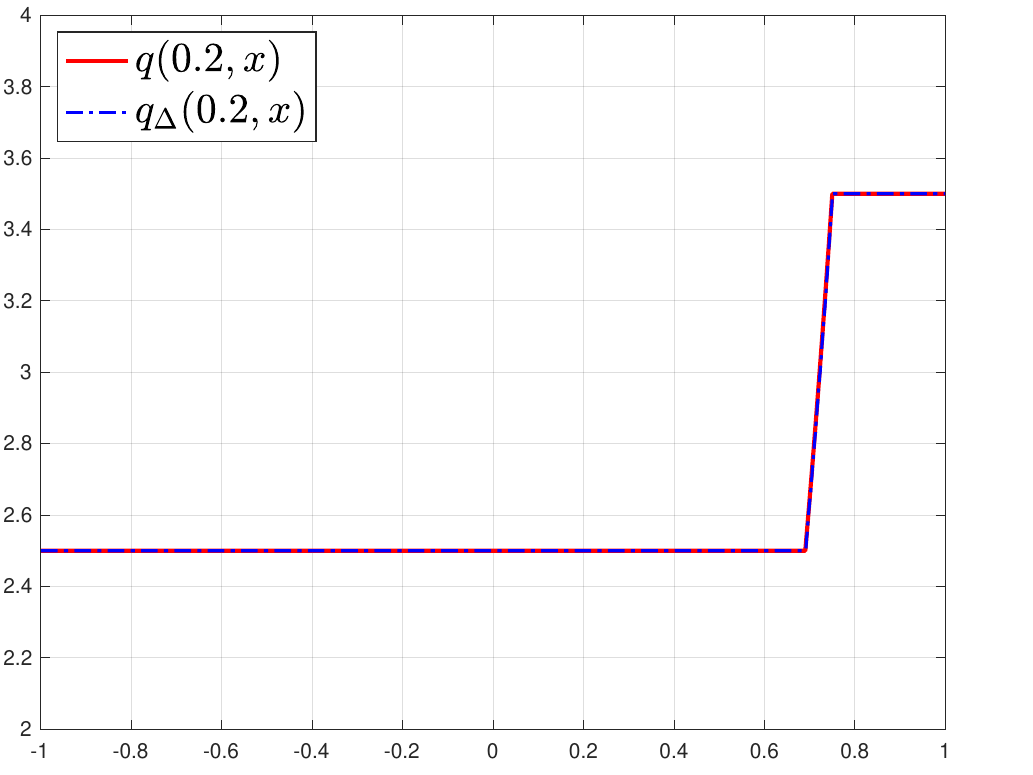}
\end{minipage}
\caption{A solution with a c-wave and a $2$-rarefaction, see Example~\ref{ex:3}.}
\label{fig:ex3}
\end{figure}

\begin{ex}[A solution independent of $h_o$]\label{ex:4}
Consider the initial states $u_\ell=(2.5,-4)$ and $u_r = (3,3)$. In this case $v(u_\ell) < a$, and therefore we have $\overline{Q}(u_\ell) = \overline{q}(u_\ell) \approx 0.8265 <q_*^-$.
As a result, the valve closes for $t>0$, $\hc(u_\ell,h_o)=0$, yielding $\qc(u_\ell,h_o)=0$, independently of the initial value of $h_o$. The left state $u_\ell$ is then connected to $\ucm\approx(1.1233,0)$ with a $1$-rarefaction, $\ucm$ to $\ucp\approx(1.8196,0)$ with a c-wave, and $\ucp$ to $u_r$ with a $2$-rarefaction, see \figurename~\ref{fig:ex4}.
\end{ex}

\begin{figure}[!htb]

\begin{minipage}[c]{0.33\textwidth}
\centering
\includegraphics[trim=0cm 0cm 0cm 0.1cm, clip,width=0.95\linewidth]{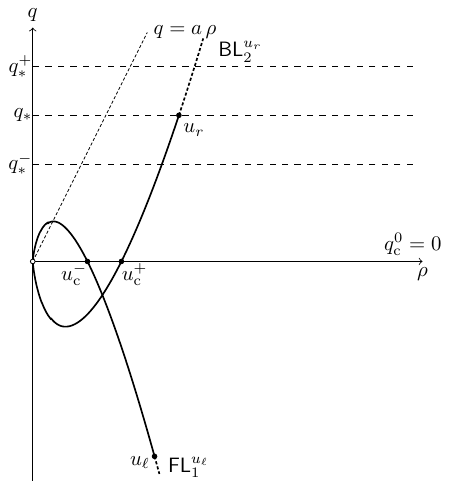}
\end{minipage}%
\hfill
\begin{minipage}{0.33\textwidth}
\centering
\includegraphics[trim=2.5cm 8.3cm 1cm 8.3cm, clip, width=0.95\linewidth]{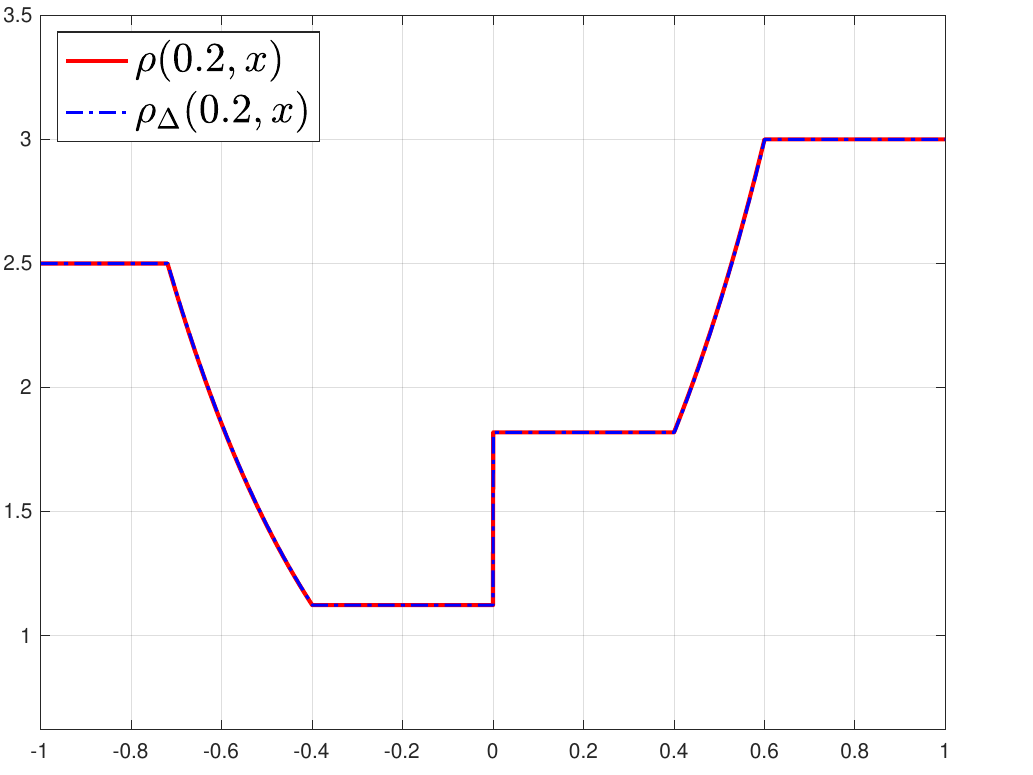}
\end{minipage}%
\hfill
\begin{minipage}{0.33\textwidth}
\includegraphics[trim=2.5cm 8.3cm 1cm 8.3cm, clip, width=0.95\linewidth]{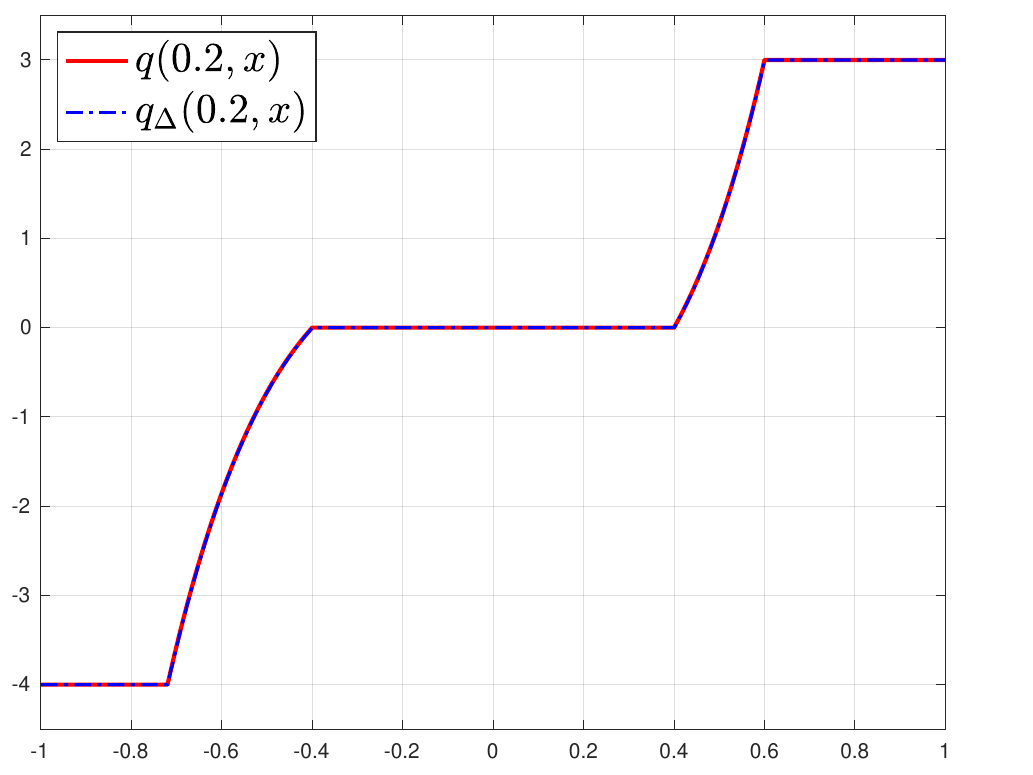}
\end{minipage}

\caption{A solution with a $1$-rarefaction, a c-wave and a $2$-rarefaction, see Example~\ref{ex:4}.}
\label{fig:ex4}
\end{figure}

The \figurename s~\ref{fig:ex1}-\ref{fig:ex4} corresponding to Examples~\ref{ex:1}-\ref{ex:4} demonstrate that the numerical scheme accurately captures the fundamental waves of the solutions to Riemann problems generated by the c-Riemann solver $\rsc$ in the case of initial data in $\mathsf{CH}$. To evaluate the scheme's performance, we introduce the relative $\L1$-error between the numerical solution $u_\Delta$ and the exact solution $u$ at time $t$:
\begin{equation}
\label{e:relerr}
e^t_{\L1}(\Delta x)\doteq \dfrac{\|u_\Delta(t,\,\cdot\,)-u(t,\,\cdot\,)\|_{\L1(I)}}{\|u(t,\,\cdot\,)\|_{\L1(I)}}.
\end{equation}

Figure~\ref{fig:err} plots the relative $\L1$-errors on a log-log scale for the previous examples, confirming the convergence of the scheme with an approximate convergence rate of $1$.

\begin{figure}[!htb]
\centering
\resizebox{\linewidth}{!}{
\begin{minipage}{\linewidth}
\includegraphics[trim=1.5cm 7.5cm 2.2cm 8.3cm, clip, width=0.24\linewidth]{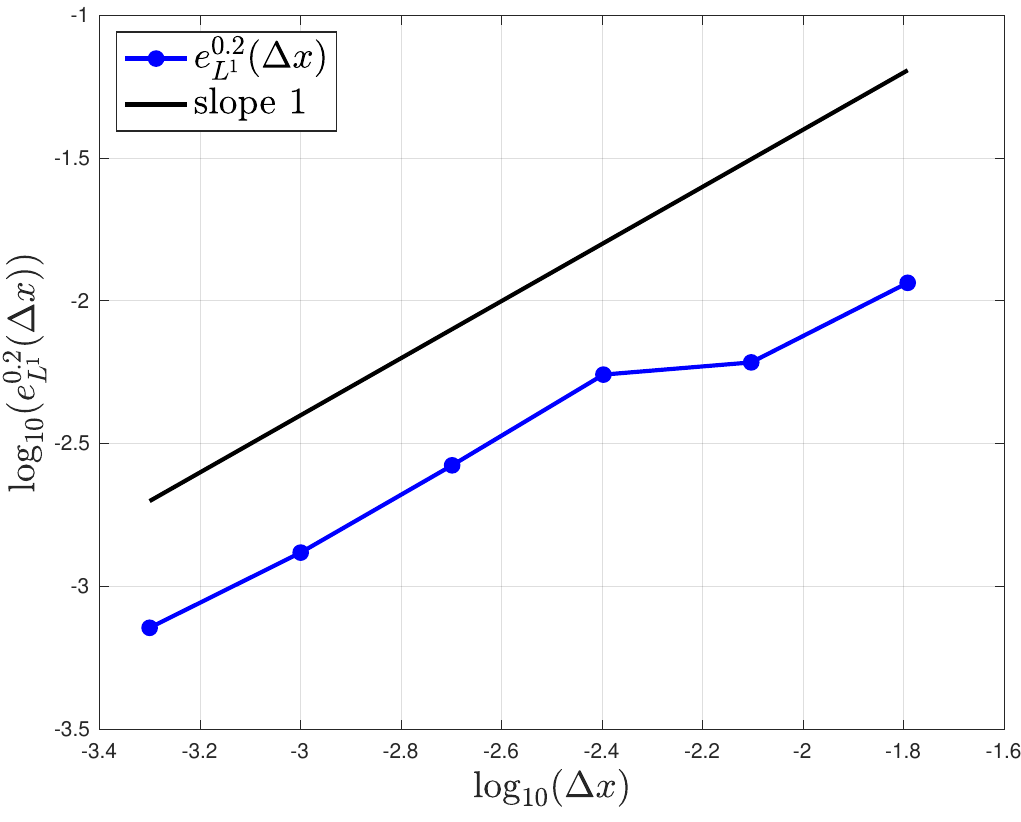}
\includegraphics[trim=1.5cm 7.5cm 2.2cm 8.3cm, clip, width=0.24\linewidth]{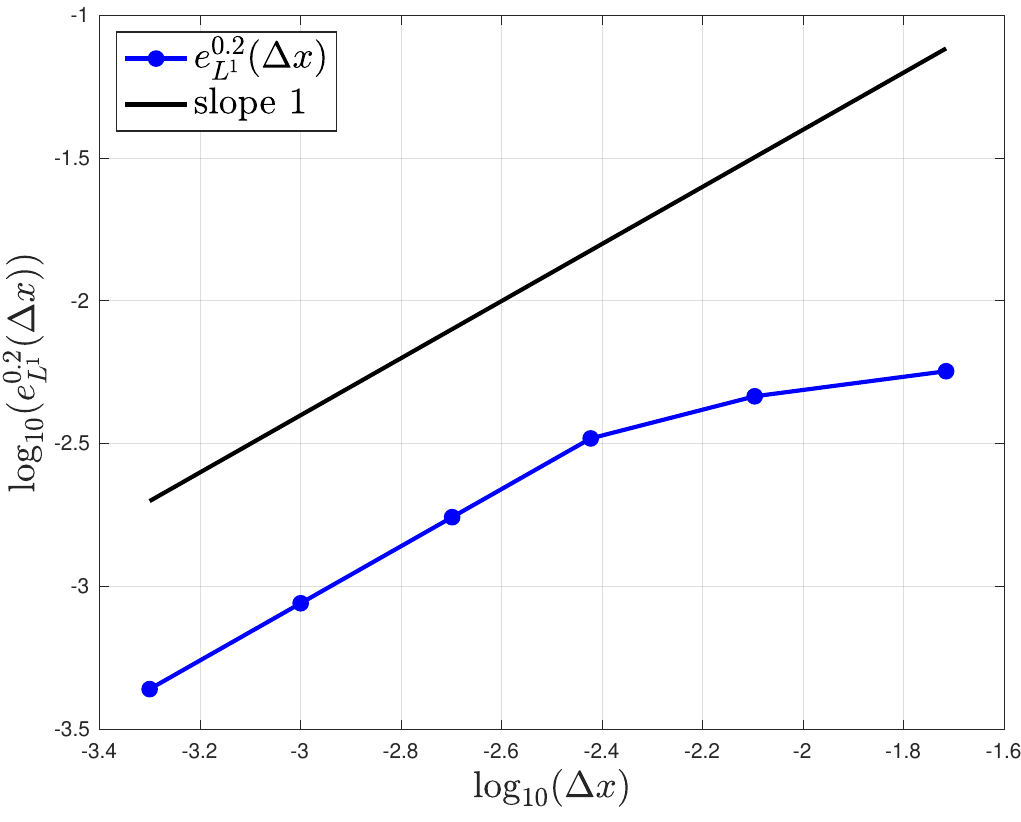}
\includegraphics[trim=1.5cm 7.5cm 2.2cm 8.3cm, clip, width=0.24\linewidth]{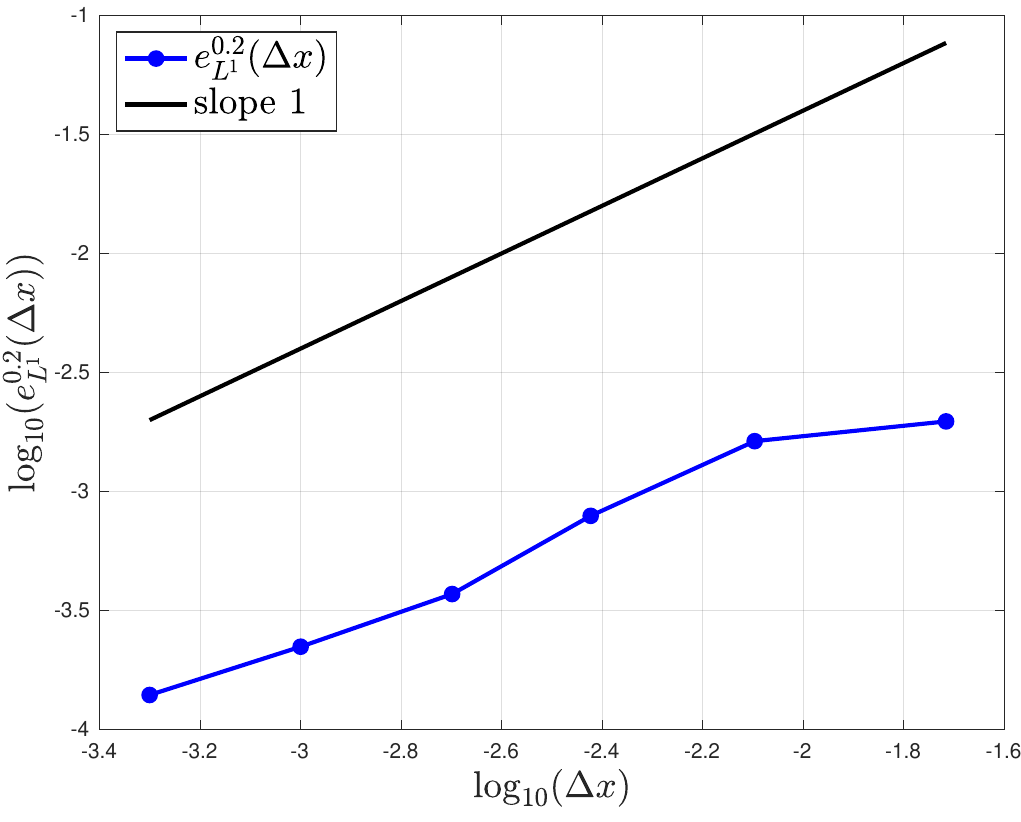}
\includegraphics[trim=1.5cm 7.5cm 2.2cm 8.3cm, clip, width=0.24\linewidth]{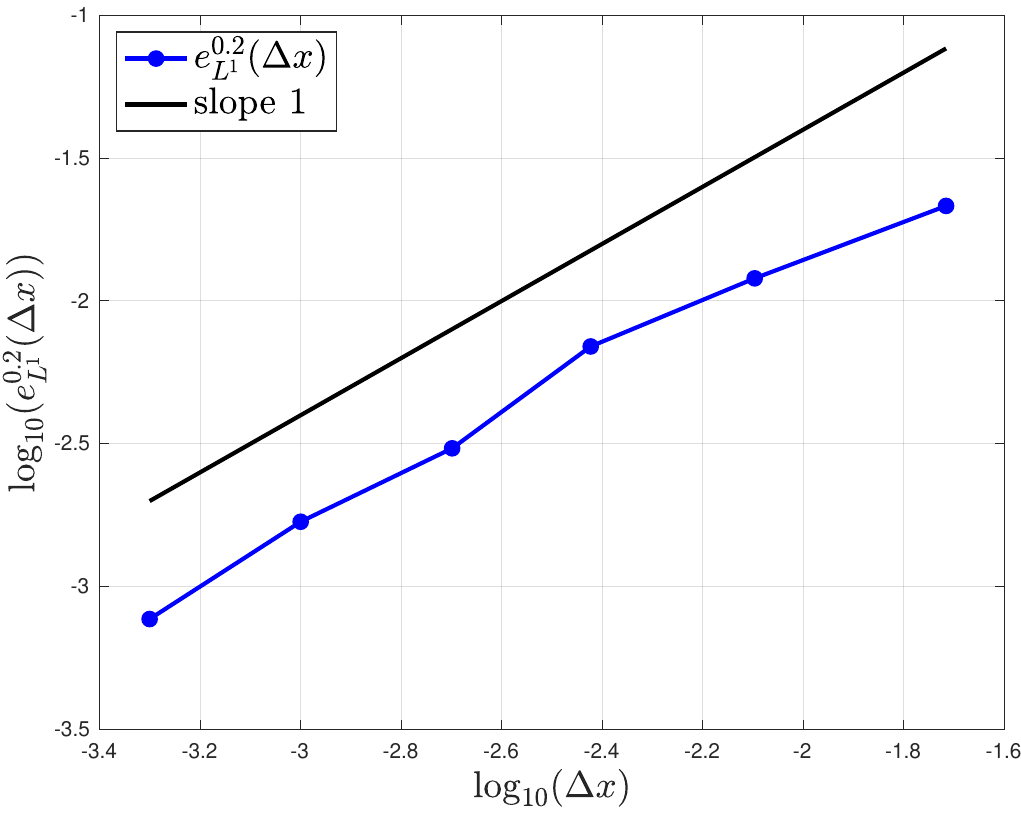}
\end{minipage}}
\caption{Relative error defined in \eqref{e:relerr} in log/log scale for Examples~\ref{ex:1}-\ref{ex:4}.}
\label{fig:err}
\end{figure}

\subsection{Numerical simulations with initial Riemann data in \texorpdfstring{$\mathsf{CH}^{\scriptscriptstyle\complement}$}{}}
\label{s:num3}

Now, we present some examples with Riemann data $(u_\ell,u_r,h_o) \in \mathsf{CH}^{\scriptscriptstyle\complement}$. In these cases,  for some initial time steps, the numerical solution may exhibit instabilities around the valve location $x=0$. For subsequent time steps, such instabilities either can lead to chattering phenomena, producing large oscillations in the $\rho$, $q$, and $h$ components of the approximate solutions $u_\Delta$ for $x>0$, or they may decay, yielding non-oscillatory profiles with, in general, $(u_\Delta(T,0^-),u_\Delta(T,0^+),h_{\Delta}(T)) \not = C_\Delta(u_\ell,u_r,h_o)$, where $T=0.2$. 

More specifically, in Example~\ref{ex:CHc1}, we consider $u_\ell \in \mathsf{CH}_{\ell,1}^{\scriptscriptstyle\complement}(q_*^-,q_*^+) \subset \mathsf{CH}_{\ell,0}^{\scriptscriptstyle\complement}(q_*^+)$, hence with $q_\ell<q_*^-$. Independently of the value of $h_o$, this choice leads to a chattering effect in the numerical solution, which {\em persists} and is {\em amplified} over time, as can still be observed at time $T$. In Example~\ref{ex:CHc2}, we consider the same data $u_r$ as in Example~\ref{ex:CHc1}, but now with $u_\ell \in \mathsf{CH}_{\ell,0}^{\scriptscriptstyle\complement}(q_*^+) \setminus \mathsf{CH}_{\ell,1}^{\scriptscriptstyle\complement}(q_*^-,q_*^+)$, so $q_\ell \in [q_*^-,q_*^+]$, and with $h_o=0$. The instability {\em disappears} by $T=0.2$; however, the numerical scheme {\em fails} to capture the solution generated by $\rsc$. Finally, in Example~\ref{ex:CHc3}, we consider the same data $u_\ell$ and $u_r$ as in Example~\ref{ex:3}, but with $h_o=0$. In this case we have $u_\ell \in \mathsf{CH}_{\ell,0}^{\scriptscriptstyle\complement}(q_*^+) \setminus \mathsf{CH}_{\ell,1}^{\scriptscriptstyle\complement}(q_*^-,q_*^+)$, as in Example \ref{ex:CHc2}, and observe that, although the instability again disappears by $T=0.2$, the numerical scheme {\em fails} to capture the solution generated by $\rsc$, in contrast to the case $h_o=1$ considered in Example~\ref{ex:3}.  See also \figurename~\ref{fig:plotul} for the plots of $u_\ell$ for Examples~\ref{ex:CHc1}-\ref{ex:CHc3}.

In the simulations we again use the parameters given in \eqref{e:parameters}.

\begin{figure}
\centering
\resizebox{.8\linewidth}{!}{
\begin{minipage}{0.3\linewidth}
\centering
\includegraphics[width = 0.99\linewidth]{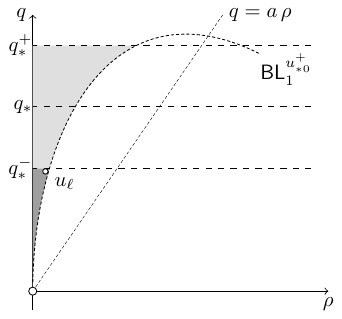}

\end{minipage}
\begin{minipage}{0.3\linewidth}
\centering
\includegraphics[width = 0.99\linewidth]{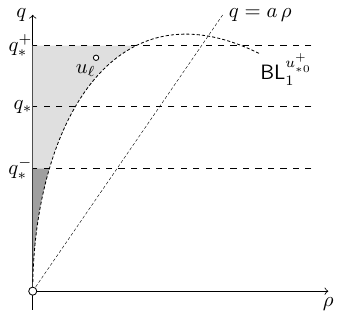}
    
\end{minipage}
\begin{minipage}{0.30\linewidth}
\centering
\includegraphics[width = 0.99\linewidth]{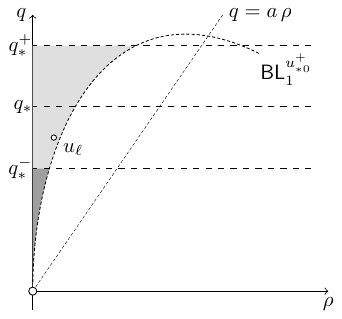}
    
\end{minipage}}
\caption{
The states $u_\ell$ for Examples~\ref{ex:CHc1}-\ref{ex:CHc3}, from left to right. The dark-gray region represents $\mathsf{CH}_{\ell,1}^{\scriptscriptstyle\complement}(q_*^-,q_*^+) \subset \mathsf{CH}_{\ell,0}^{\scriptscriptstyle\complement}(q_*^+)$; the light-gray region represents $\mathsf{CH}_{\ell,0}^{\scriptscriptstyle\complement}(q_*^+) \setminus \mathsf{CH}_{\ell,1}^{\scriptscriptstyle\complement}(q_*^-,q_*^+)$.}
    \label{fig:plotul}
\end{figure}

\begin{ex}
\label{ex:CHc1}
Consider the initial states $u_\ell=(0.18,1.95)$ and $u_r = (1,0.5).$ As $v(u_\ell)>a$, then $\overline{Q}(u_\ell) = q_\ell = 1.95 < q_*^-$ meaning that $\qc(u_\ell,h_o) =0$ both for $h_o=0$ and $h_o=1$; moreover, $u_\ell \in \mathsf{CH}_{\ell,1}^{\scriptscriptstyle\complement}(q_*^-,q_*^+) \subset \mathsf{CH}_{\ell,0}^{\scriptscriptstyle\complement}(q_*^+)$. As a result, {\em independently of the initial value of $h_o$} we get the same numerical solution; such a solution is displayed in \figurename~\ref{fig:chattering1} and exhibits chattering effect. It can be observed that $q_{j_{\rm c}}^n=q_\ell$ after a finite number of time steps, because $q_\ell < q_*^-$. On the other hand, we have $q_{j_{\rm c}+1}^n \neq q_\ell$ because by construction $q_{j_{\rm c}+1}^n$ belongs to $\{0\} \cup [q_*^-,q_*]$ whereas $q_\ell$ does not; therefore the approximate solution does not satisfy the first Rankine-Hugoniot condition \eqref{e:RH1} at $x=0$.
In other words, the scheme does not conserve the mass in this case!
\end{ex}

\begin{figure}[!htb]
\centering
\begin{minipage}{0.32\linewidth}
\centering
\includegraphics[trim=2.2cm 7.5cm 2.2cm 8.3cm, clip, width=\linewidth]{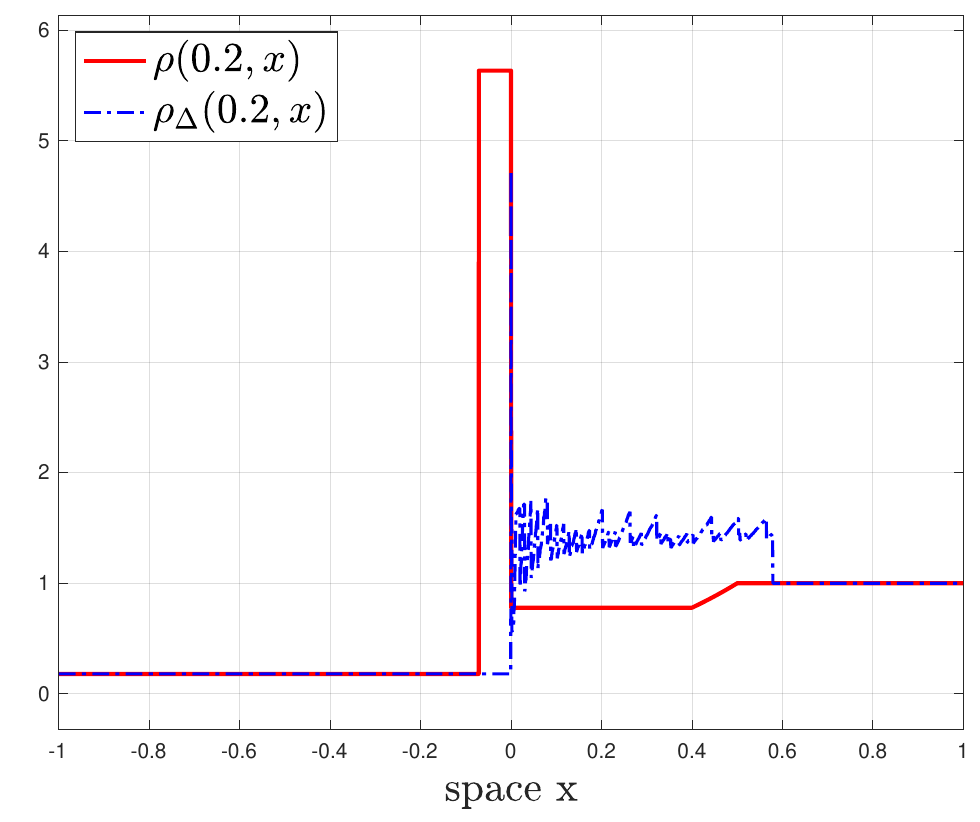}
\end{minipage}\hfill
\begin{minipage}{0.32\linewidth}
\centering
\includegraphics[trim=2.2cm 7.5cm 2.2cm 8.3cm, clip, width=\linewidth]{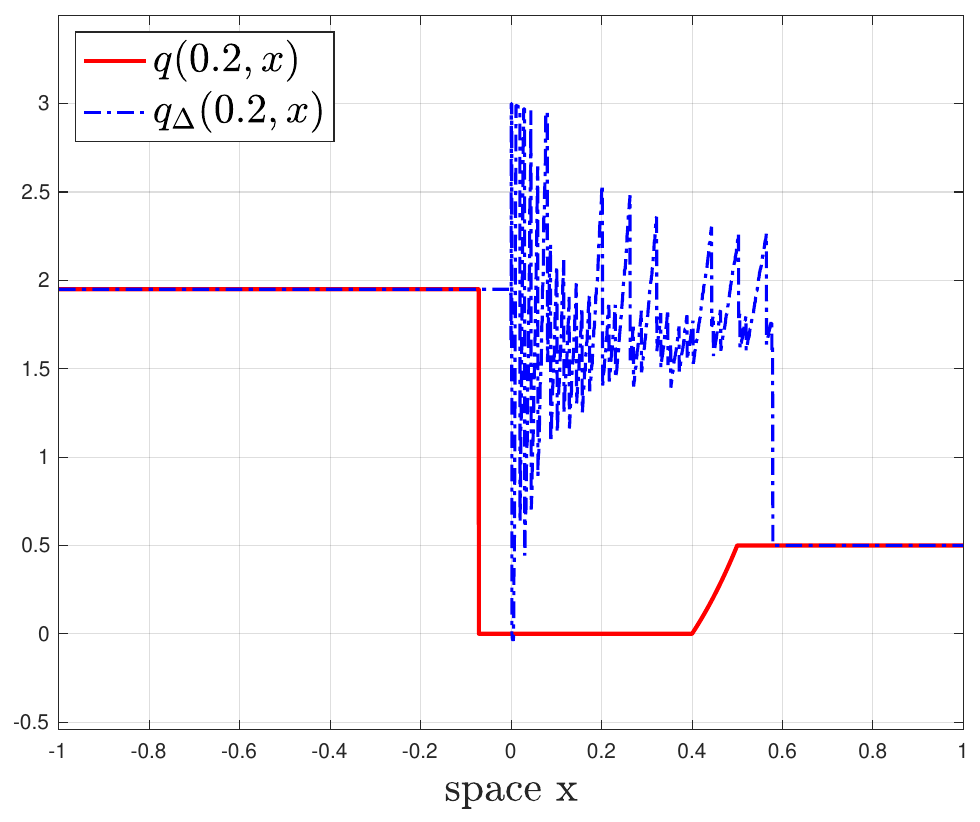}
\end{minipage}\hfill
\begin{minipage}{0.32\linewidth}
\centering
\includegraphics[trim=2.2cm 7.5cm 2.2cm 8.3cm, clip, width=\linewidth]{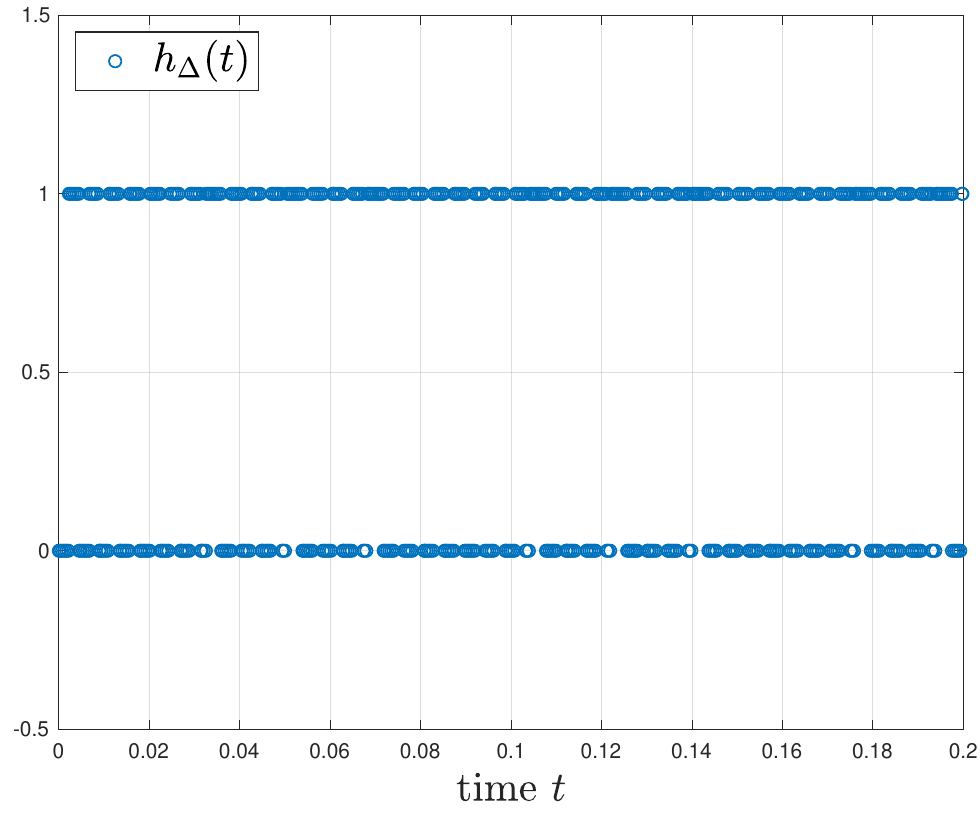}
\end{minipage}
\caption{Chattering at $T=0.2$ in the numerical solution corresponding to the Riemann datum in $\mathsf{CH}^{\scriptscriptstyle\complement}$ with $h_o=1$, as in Example~\ref{ex:CHc1}.}
\label{fig:chattering1}
\end{figure}

\begin{ex}
\label{ex:CHc2}
Consider the initial states $u_\ell = (0.75, 3.8)$, $u_r=(1,0.5)$, with $h_o=0.$ In this case the chattering effect in the numerical solution $u_\Delta$ is no longer present at time $T$, because $q_\ell > q_*$. However, we have $\bigl(u_\Delta(T,0^-),u_\Delta(T,0^+),h_{\Delta}(t)\bigr) \not = C_\Delta(u_\ell,u_r,0)$, as can be seen in \figurename~\ref{fig:chattering0}: in particular, one can verify that $\bigl(u_\Delta(T,0^-),u_\Delta(T,0^+),h_\Delta(T)\bigr)=C_\Delta(u_\ell,u_r,1)$. Moreover, during an initial time interval, as shown in \figurename~\ref{fig:chattering0} on the right, it holds $h_\Delta(t)=0$, and the numerical solutions for $\rho$ and $q$ approximate the first two components of $\rsc[u_\ell,u_r,0]$. Afterwards, we have $h_\Delta(t)=1$, and the numerical solutions for $\rho$ and $q$ capture, with some time delay, the first two components of $\rsc[u_\ell,u_r,1]$. Notice that, we have $q_\ell>q_*^-$ and $u_\ell \in \mathsf{CH}_{\ell,0}^{\scriptscriptstyle\complement}(q_*^+) \setminus \mathsf{CH}_{\ell,1}^{\scriptscriptstyle\complement}(q_*^-,q_*^+)$.

\end{ex}

\begin{figure}[!htb]
\centering
\begin{minipage}{0.32\linewidth}
\centering
\includegraphics[trim=2.2cm 7.5cm 2.2cm 8.3cm, clip, width=\linewidth]{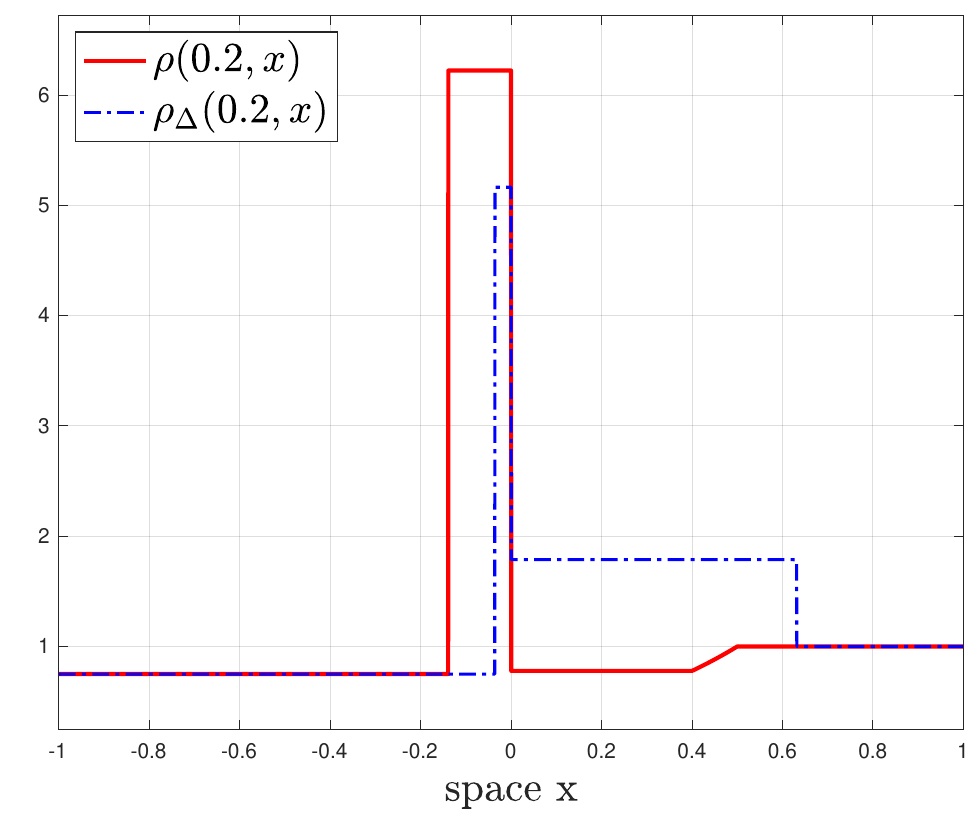}
\end{minipage}\hfill
\begin{minipage}{0.32\linewidth}
\centering
\includegraphics[trim=2.2cm 7.5cm 2.2cm 8.3cm, clip, width=\linewidth]{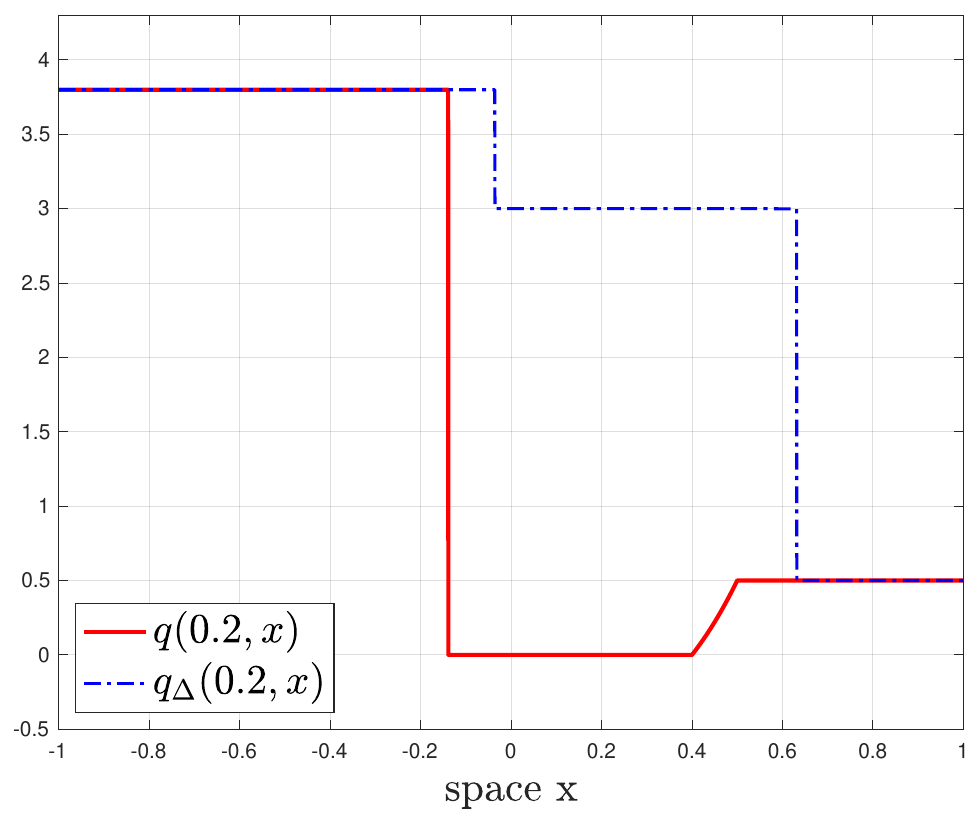}
\end{minipage}\hfill
\begin{minipage}{0.32\linewidth}
\centering
\includegraphics[trim=2.2cm 7.5cm 2.2cm 8.3cm, clip, width=\linewidth]{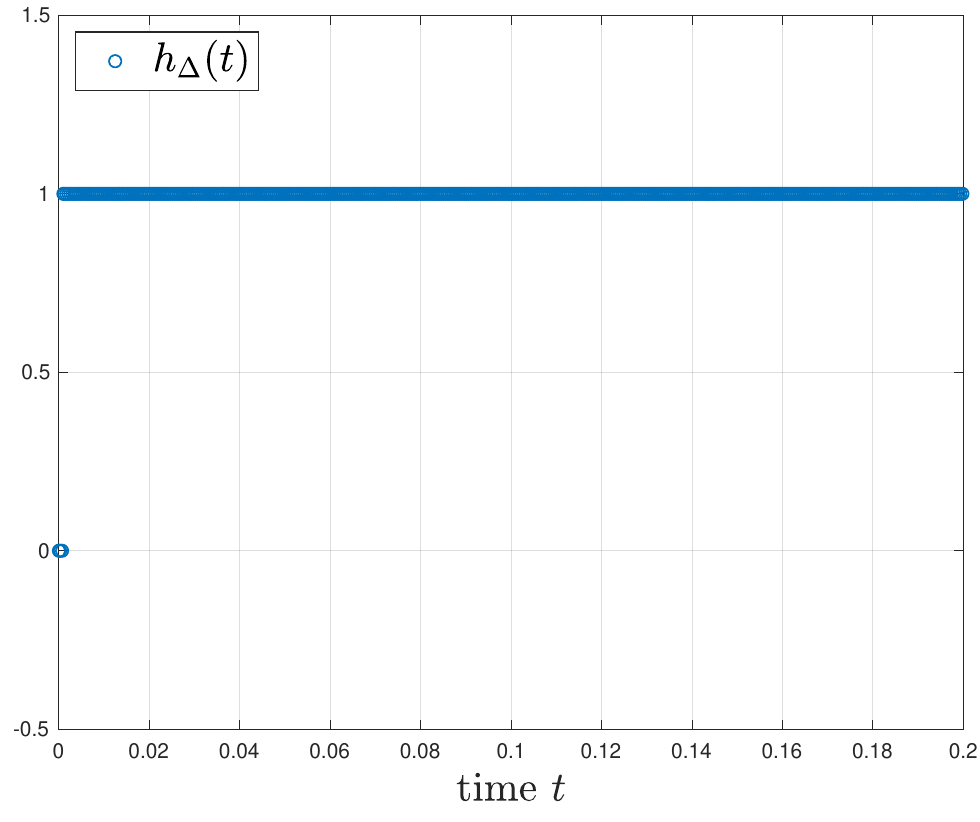}
\end{minipage}
\caption{No chattering at $T=0.2$ in the numerical solution corresponding to the Riemann datum in $\mathsf{CH}^{\scriptscriptstyle\complement}$ with $h_o=0$, as in Example~\ref{ex:CHc2}.}
\label{fig:chattering0}
\end{figure}

\begin{ex}
\label{ex:CHc3}
Consider the same initial states $u_\ell=(0.25,2.5)$ and $u_r = (2,3.5)$ as in Example~\ref{ex:3}, but with $h_o=0$.
Again, as in Example~\ref{ex:CHc2}, the numerical solution does not exhibit chattering at $T$, even if $\bigl(u_\Delta(T,0^-),u_\Delta(T,0^+),h_\Delta(T)\bigr)\neq C_\Delta(u_\ell,u_r,0)$, as shown in \figurename~\ref{fig:no_chattering}. Instead, we have $\bigl(u_\Delta(T,0^-),u_\Delta(T,0^+),h_\Delta(T)\bigr)=C_\Delta(u_\ell,u_r,1)$. Thus, during an initial time interval, the numerical solution approximates the first two components of $\rsc[u_\ell,u_r,h_o]$, and afterwards those of $\rsc[u_\ell,u_r,1]$, with some time delay. Moreover, $u_\ell \in \mathsf{CH}_{\ell,0}^{\scriptscriptstyle\complement}(q_*^+) \setminus \mathsf{CH}_{\ell,1}^{\scriptscriptstyle\complement}(q_*^-,q_*^+)$. 
This example and Example~\ref{ex:3} show that, {\em depending on the choice of $h_o$}, the numerical solution may or may not be able to capture the exact solution, see also Figure~\ref{fig:ex3}.
\end{ex}

\begin{figure}[!htb]
\centering
\begin{minipage}{0.32\linewidth}
\centering
\includegraphics[trim=2.2cm 7.5cm 2.2cm 8.3cm, clip, width=\linewidth]{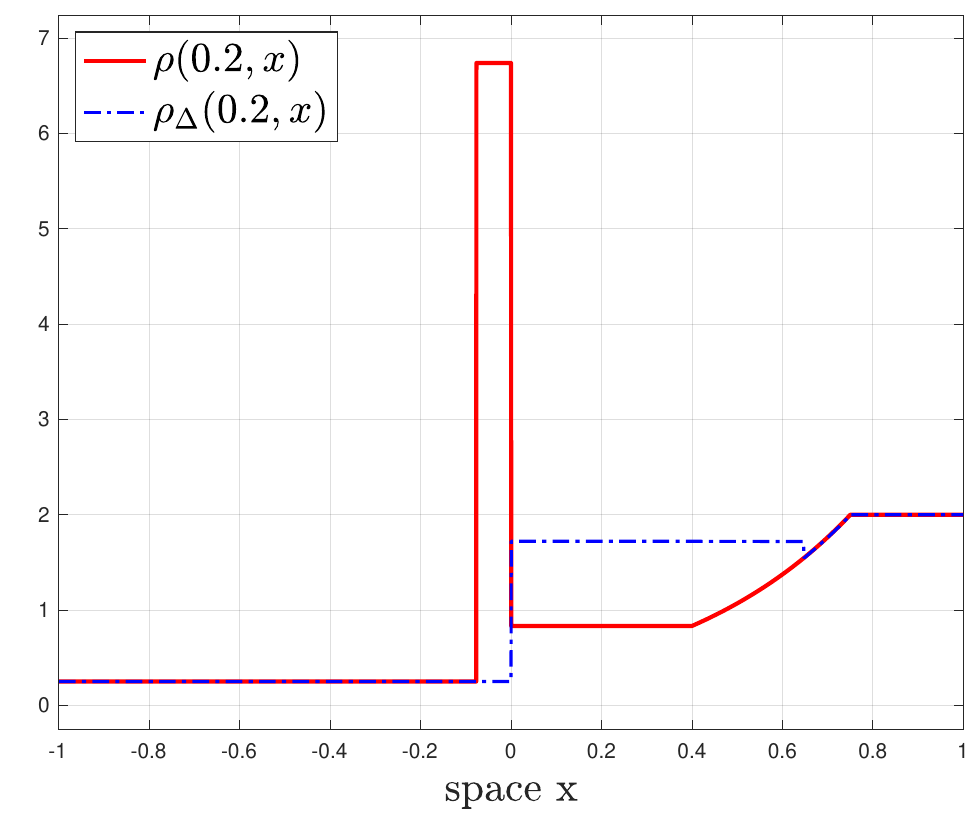}
\end{minipage}\hfill
\begin{minipage}{0.32\linewidth}
\centering
\includegraphics[trim=2.2cm 7.5cm 2.2cm 8.3cm, clip, width=\linewidth]{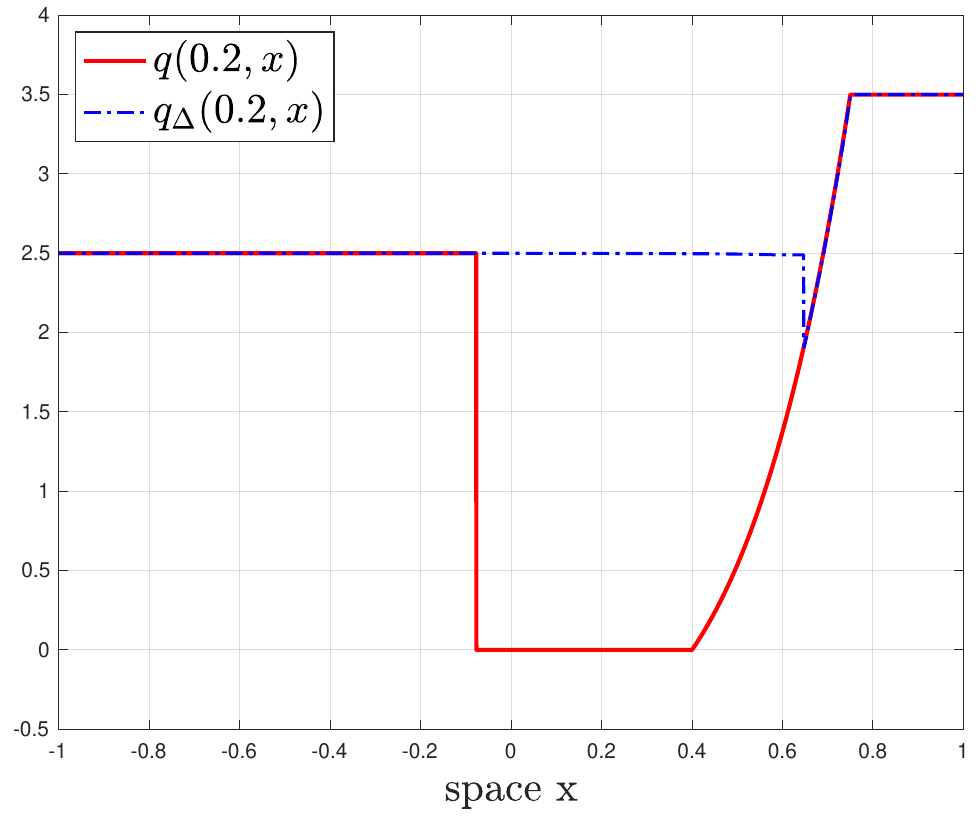}
\end{minipage}\hfill
\begin{minipage}{0.32\linewidth}
\centering
\includegraphics[trim=2.2cm 7.5cm 2.2cm 8.3cm, clip, width=\linewidth]{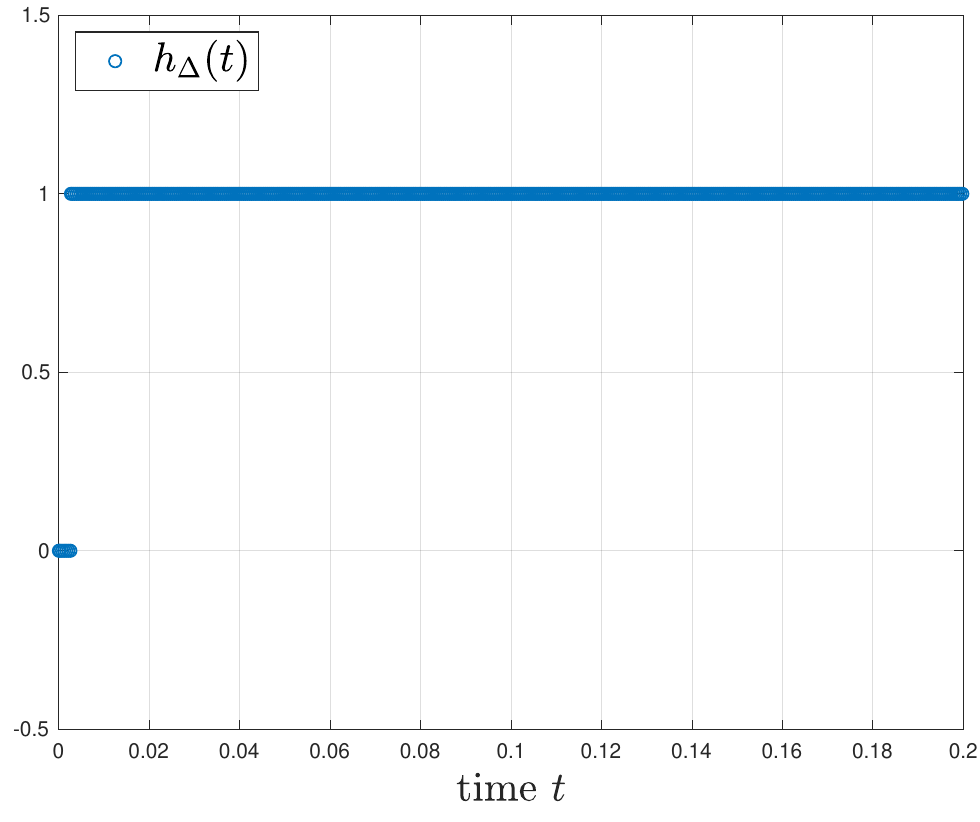}
\end{minipage}
\caption{No chattering at $T=0.2$ in the numerical solution corresponding to the Riemann datum in $\mathsf{CH}^{\scriptscriptstyle\complement}$ with $h_o=0$, as in Example~\ref{ex:CHc3}.}
\label{fig:no_chattering}
\end{figure}

As already observed, Example~\ref{ex:CHc1} suggests that, when $u_\ell \in \mathsf{CH}_{\ell,1}^{\scriptscriptstyle\complement}(q_*^-,q_*^+) \subset \mathsf{CH}_{\ell,0}^{\scriptscriptstyle\complement}(q_*^+)$, the chattering effect of the numerical solution persists and is amplified over time. On the other hand, Examples~\ref{ex:CHc2} and~\ref{ex:CHc3} show that, when $u_\ell \in \mathsf{CH}_{\ell,0}^{\scriptscriptstyle\complement}(q_*^+) \setminus \mathsf{CH}_{\ell,1}^{\scriptscriptstyle\complement}(q_*^-,q_*^+)$ and $h_o=0$, the instability is only present for short times and eventually disappears when $h_\Delta=1$. 
In these cases, the numerical solution seems to reproduce at first the solution associated with $\rsc[u_\ell,u_r,0]$. 
After some time, $h_\Delta$ switches from $0$ to $1$, and the numerical solution at the interface $x=0$ stabilizes.
More precisely, in Examples~\ref{ex:CHc2} and~\ref{ex:CHc3} the numerical solution follows the Riemann solver $\rsc[u_\ell,u_r,1]$, after some time delay.
In Example~\ref{ex:CHc1} the numerical mass is not conserved, as the approximate solution does not satisfy the first Rankine-Hugoniot condition \eqref{e:RH1} at $x=0$; on the contrary, it is conserved in Examples~\ref{ex:CHc2} and~\ref{ex:CHc3}.

\subsection{The regularization effect of hysteresis for strongly oscillatory inflow}
\label{s:num4}

We present some numerical simulations illustrating how different choices of the hysteresis parameters $q_*^-$ and $q_*^+$ affect the resulting solution. We consider an initial constant state $u_o$ at time $t=0$ given by $\rho_o \equiv 3$ and $q_o \equiv 0$, with the valve initially closed, i.e., $h_o=0$. This state is an equilibrium configuration, since $\rsc[u_o,u_o,h_o](\xi)\equiv(u_o,h_o)$.
We set $q_*=3$, $a=2$, the computational domain $I=[-1,1]$, and the final time $T=2$. For $t>0$, we impose a strongly oscillatory inflow boundary condition at the left boundary $x=-1$ by setting
$q(t,-1)=5\sin(200t)+1$ and $\rho(t,-1)=3$. The imposed oscillations propagate, on average, with positive speed and, due to the coupling between the variables, also induce oscillations in the density $\rho$. The resulting waves eventually interact with the valve located at $x=0$. Then, we consider several scenarios corresponding to different choices of $q_*^-$ and $q_*^+$, and analyze how the valve regularizes the oscillations. 

In particular, in Example~\ref{ex:nohis}, we choose $q_*^-=q_*=q_*^+$, corresponding to the case without hysteresis. In Example~\ref{ex:hissp}, we consider $q_*^-=q_*<q_*^+=3.5$, so that only the upper threshold is active. In Example~\ref{ex:hisspsm}, we choose $2=q_*^-<q_*<q_*^+=3.5$, corresponding to the full hysteresis case with both thresholds. Finally, in Example~\ref{ex:hissm}, we consider $2=q_*^-<q_*=q_*^+$, so that only the lower threshold is active.

For all simulations we set $\Delta x=5\times10^{-4}$, as in the previous subsections, and also verify that $(\rho_\Delta(t,x),q_\Delta(t,x),h_\Delta(t))$ is coherent for all $x\in I$ and $t\in[0,T]$, so that the numerical scheme provides a consistent approximation of the exact solution.

\begin{ex}[The case without hysteresis]
\label{ex:nohis}
Take $q_*^-=q_*=q_*^+ = 3$, corresponding to a valve without hysteresis. The oscillatory incoming waves force the valve to open and close frequently, producing strong oscillations in the space domain $[0,1]$ downstream of the valve, see Figure~\ref{fig:exnohis}.
\end{ex}

\begin{figure}[!htb]
\centering
\begin{minipage}{0.32\linewidth}
\centering
\includegraphics[trim=2.2cm 7.5cm 2.2cm 8.3cm, clip, width=\linewidth]{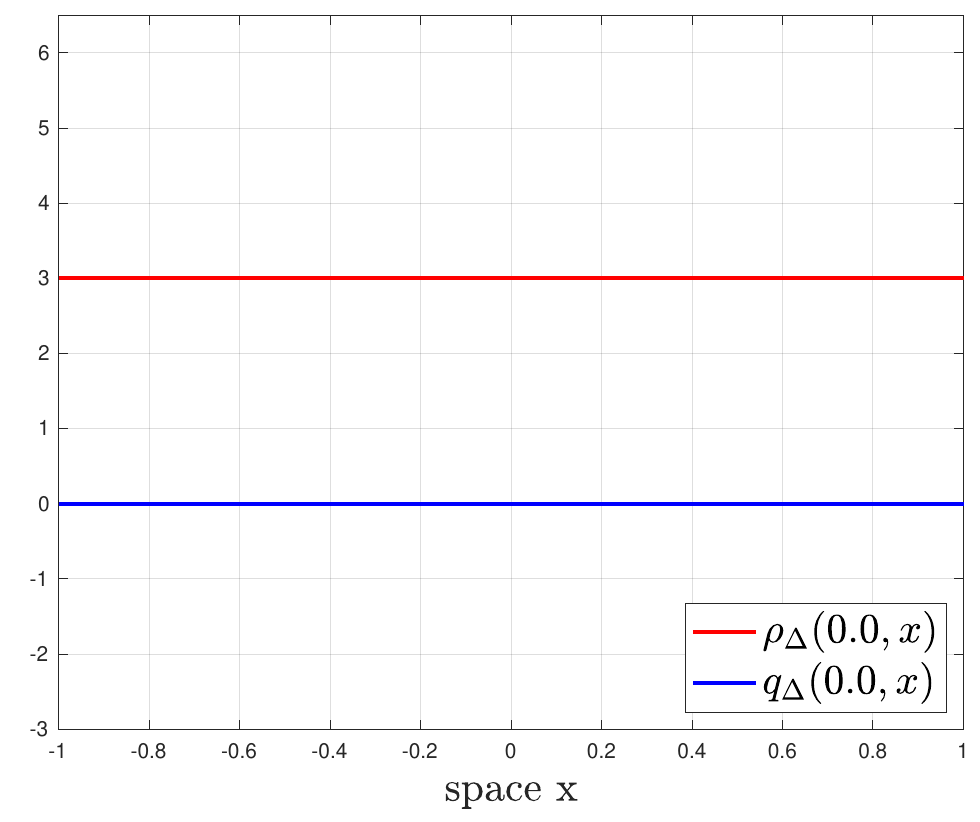}
\end{minipage}\hfill
\begin{minipage}{0.32\linewidth}
\centering
\includegraphics[trim=2.2cm 7.5cm 2.2cm 8.3cm, clip, width=\linewidth]{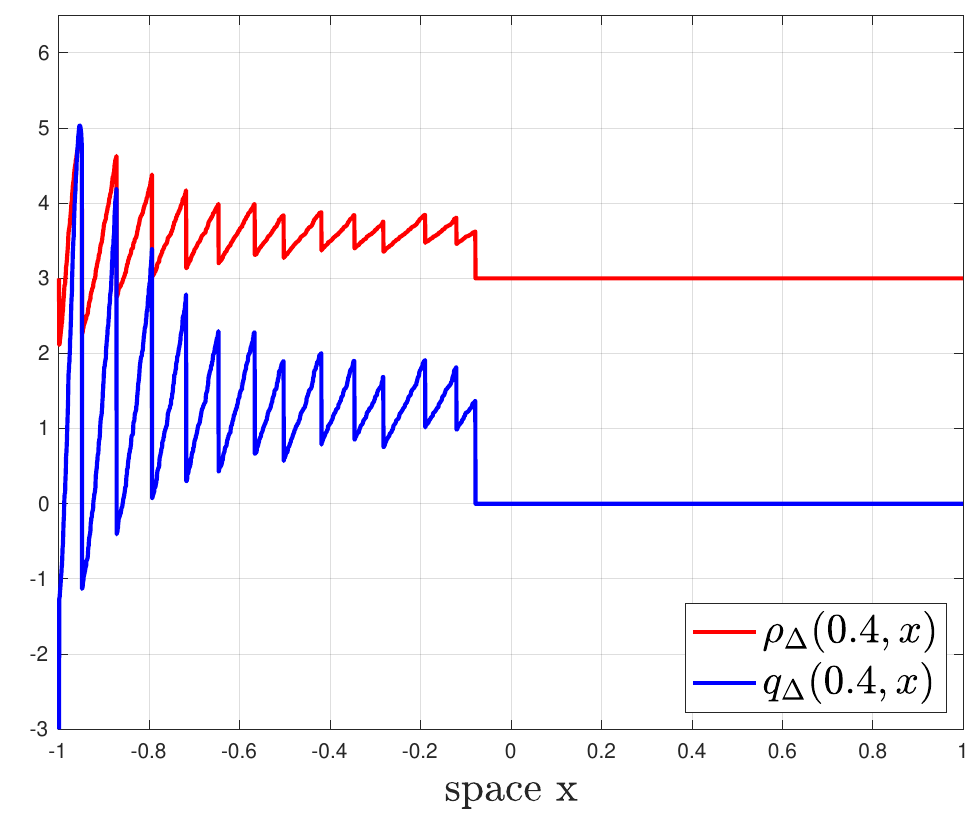}
\end{minipage}\hfill
\begin{minipage}{0.32\linewidth}
\centering
\includegraphics[trim=2.2cm 7.5cm 2.2cm 8.3cm, clip, width=\linewidth]{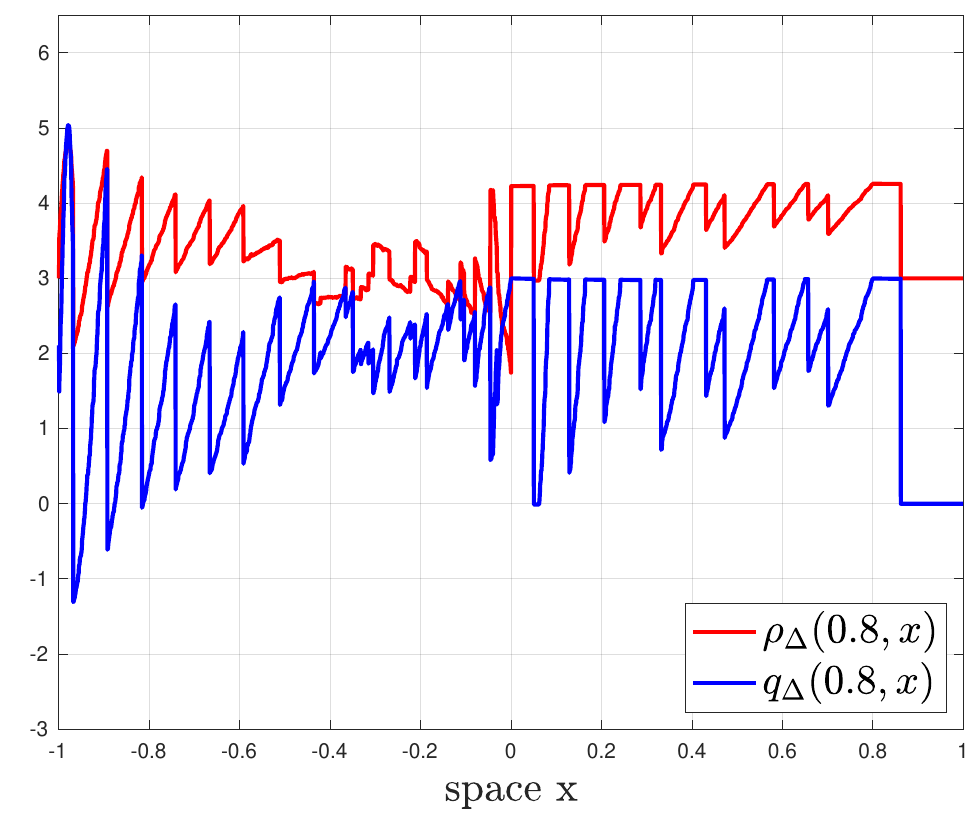}
\end{minipage}
\begin{minipage}{0.32\linewidth}
\centering
\includegraphics[trim=2.2cm 7.5cm 2.2cm 8.3cm, clip, width=\linewidth]{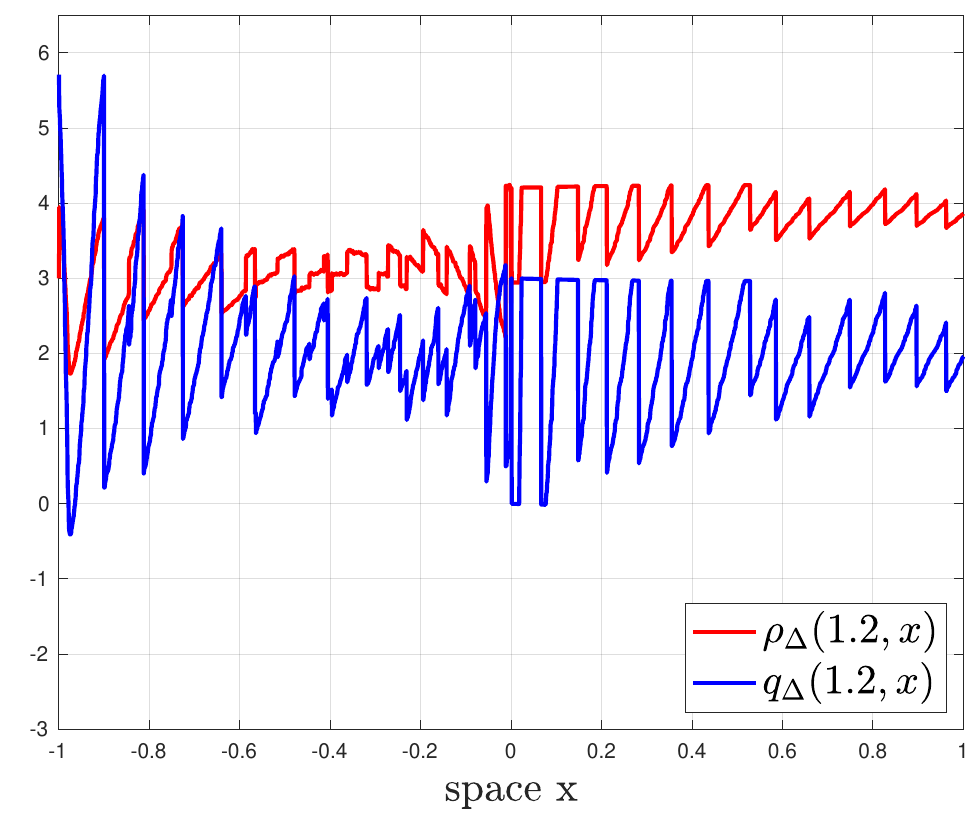}
\end{minipage}\hfill
\begin{minipage}{0.32\linewidth}
\centering
\includegraphics[trim=2.2cm 7.5cm 2.2cm 8.3cm, clip, width=\linewidth]{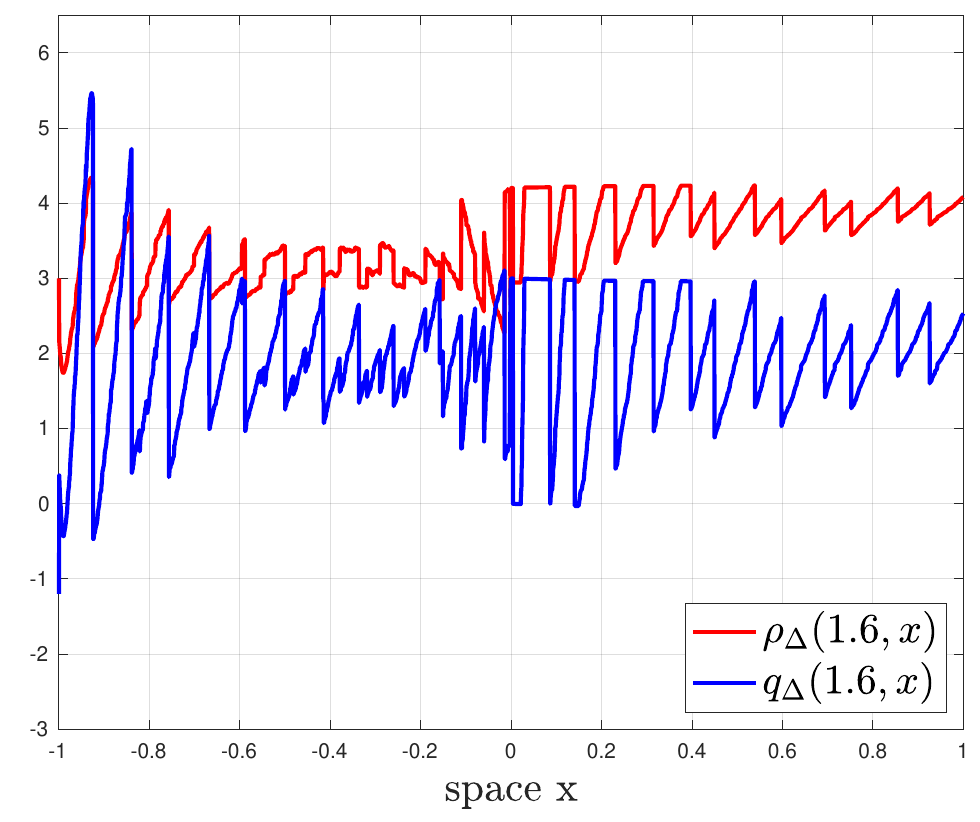}
\end{minipage}\hfill
\begin{minipage}{0.32\linewidth}
\centering
\includegraphics[trim=2.2cm 7.5cm 2.2cm 8.3cm, clip, width=\linewidth]{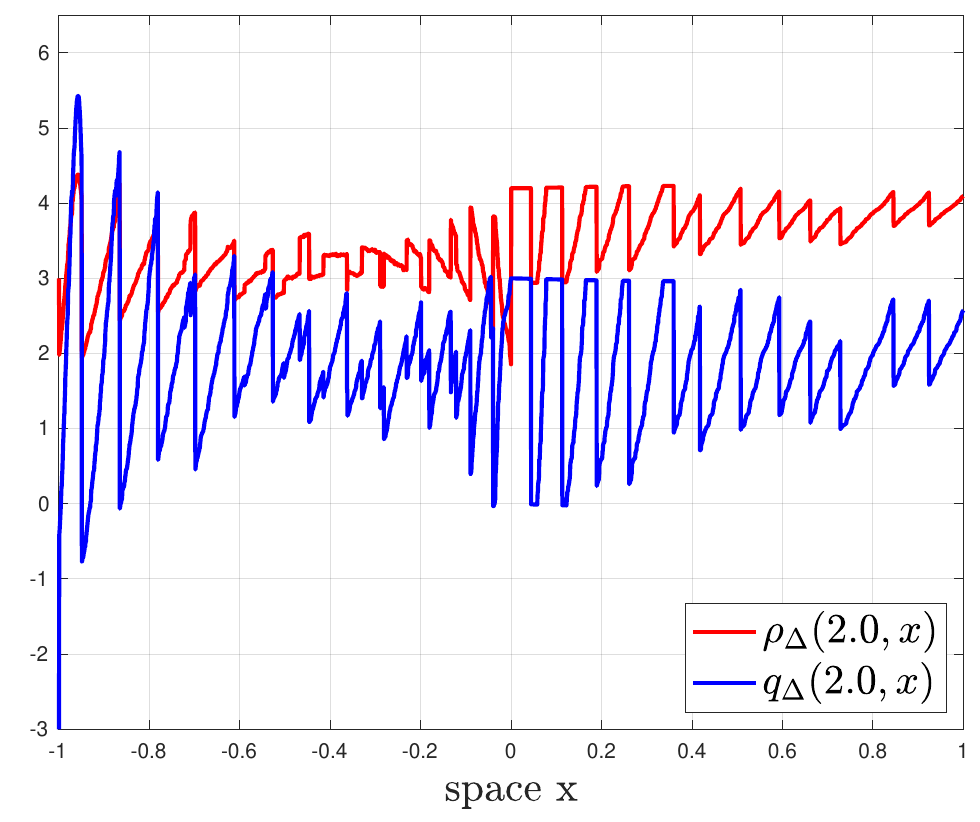}
\end{minipage}

\caption{Strong oscillatory inflow at $x=-1$ when $q_*^-=q_*=q_*^+$, see Example~\ref{ex:nohis}.}
\label{fig:exnohis}
\end{figure}

\begin{ex}[The case when the upper threshold is active]
\label{ex:hissp}
Take $q_*^-=q_*=3$ and $q_*^+=3.5$. Compared to Example~\ref{ex:nohis}, when the valve is closed then the set of states for which it opens is {\em reduced}; however, if the valve is open, then the states that lead to closure are the {\em same} as in Example~\ref{ex:nohis}. As a result, we observe some switches of the valve between the open and closed configurations, which persist with time, see Figure~\ref{fig:exhissp}. Nevertheless, the oscillatory behavior observed in Example~\ref{ex:nohis} is no longer present.
\end{ex}

\begin{figure}[!htb]
\centering
\begin{minipage}{0.32\linewidth}
\centering
\includegraphics[trim=2.2cm 7.5cm 2.2cm 8.3cm, clip, width=\linewidth]{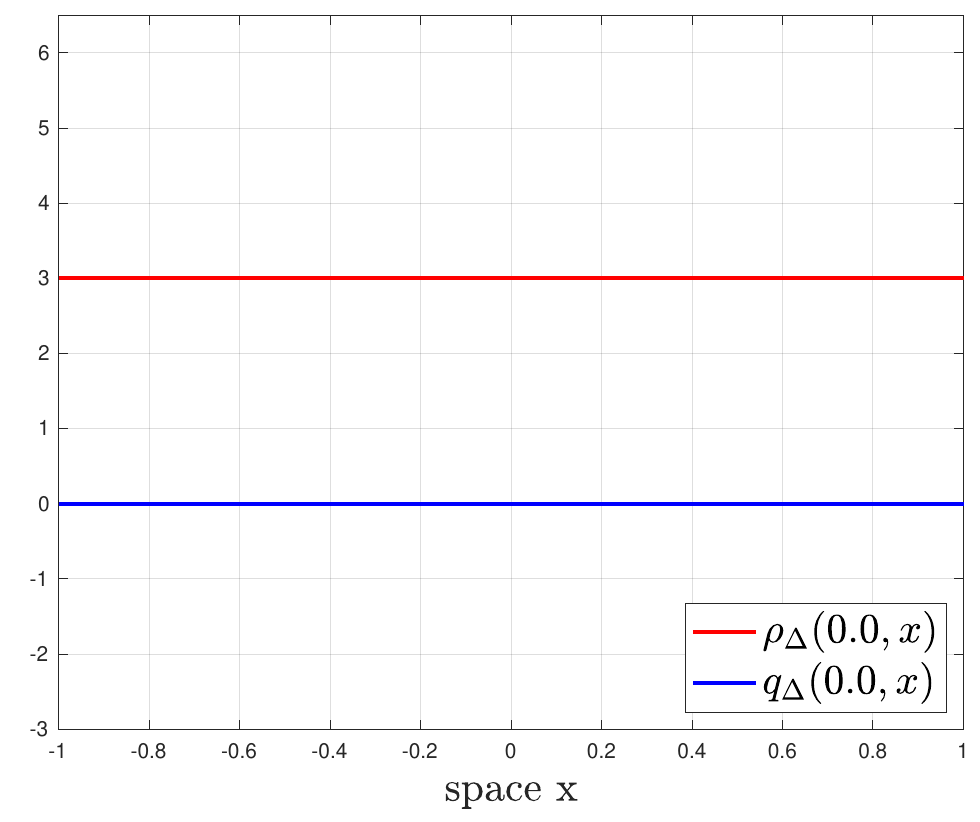}
\end{minipage}\hfill
\begin{minipage}{0.32\linewidth}
\centering
\includegraphics[trim=2.2cm 7.5cm 2.2cm 8.3cm, clip, width=\linewidth]{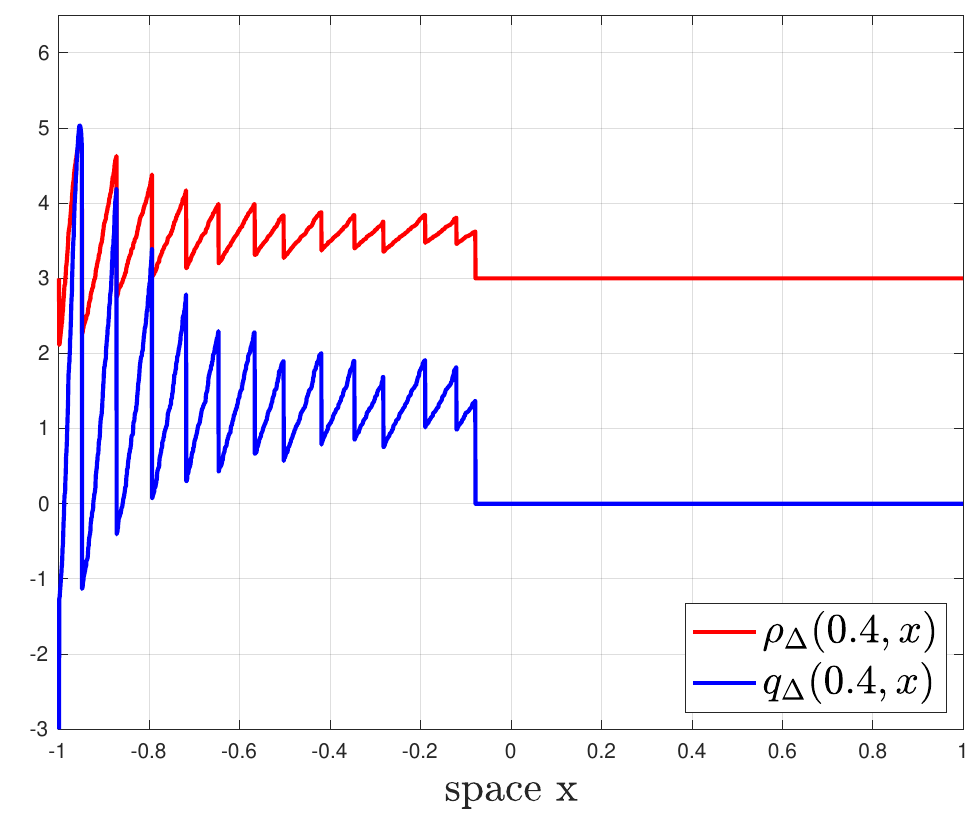}
\end{minipage}\hfill
\begin{minipage}{0.32\linewidth}
\centering
\includegraphics[trim=2.2cm 7.5cm 2.2cm 8.3cm, clip, width=\linewidth]{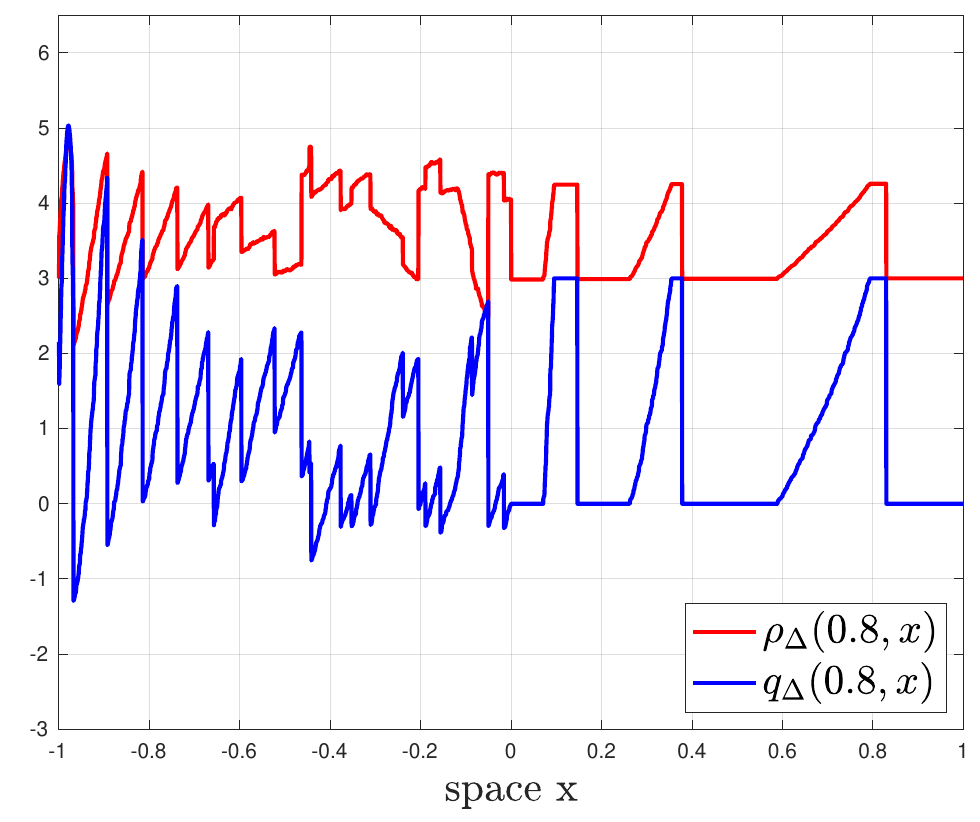}
\end{minipage}
\begin{minipage}{0.32\linewidth}
\centering
\includegraphics[trim=2.2cm 7.5cm 2.2cm 8.3cm, clip, width=\linewidth]{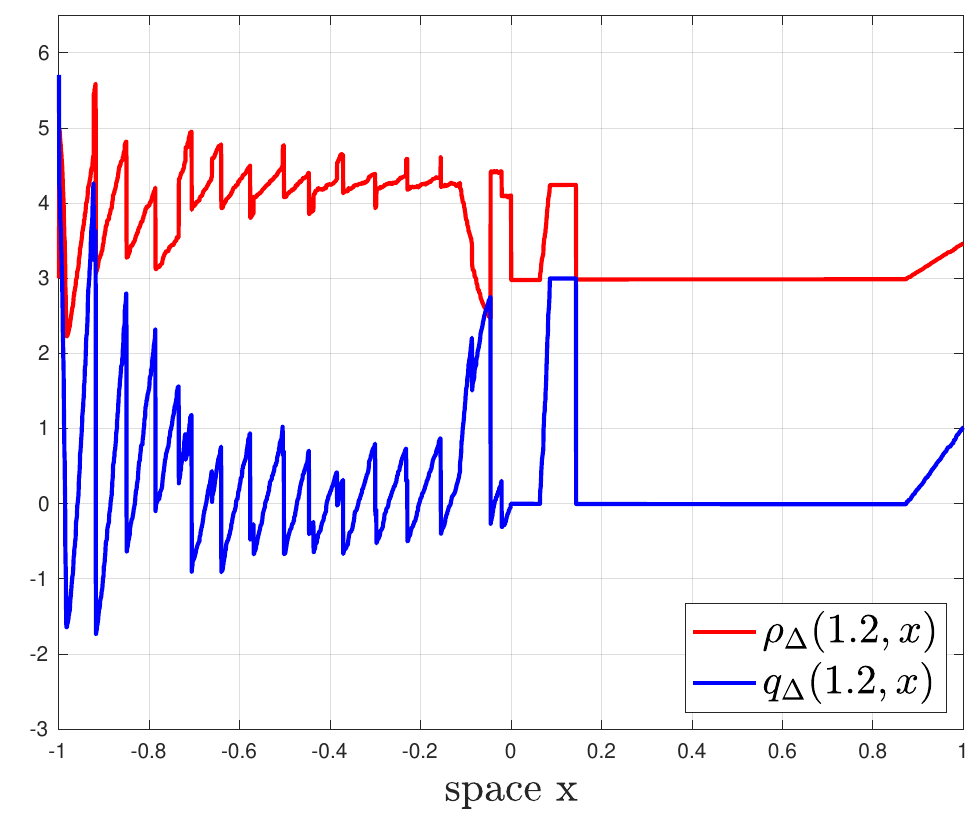}
\end{minipage}\hfill
\begin{minipage}{0.32\linewidth}
\centering
\includegraphics[trim=2.2cm 7.5cm 2.2cm 8.3cm, clip, width=\linewidth]{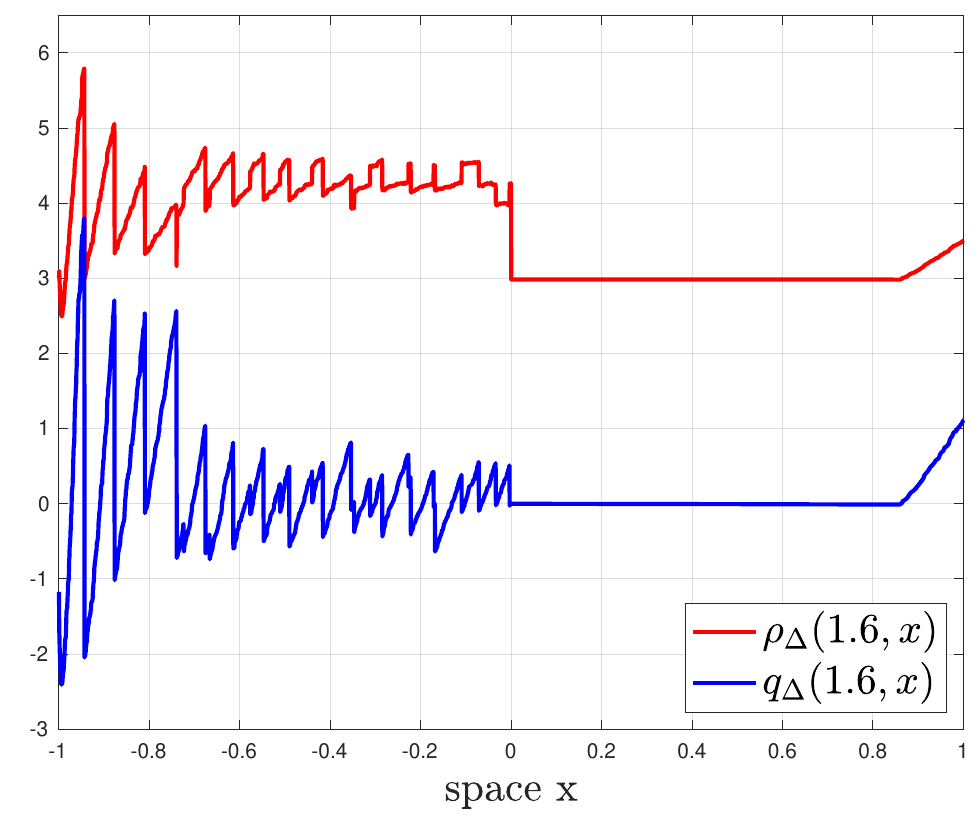}
\end{minipage}\hfill
\begin{minipage}{0.32\linewidth}
\centering
\includegraphics[trim=2.2cm 7.5cm 2.2cm 8.3cm, clip, width=\linewidth]{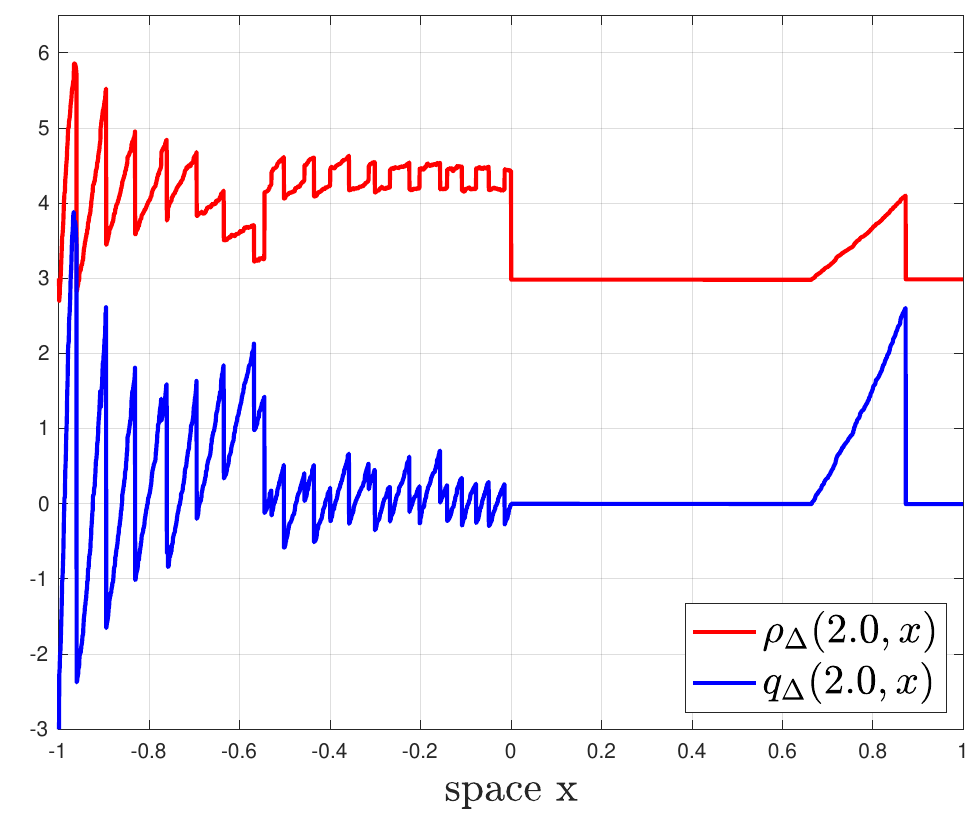}
\end{minipage}

\caption{Strong oscillatory inflow at $x=-1$ when $q_*^-=q_*<q_*^+$, see Example~\ref{ex:hissp}.}
\label{fig:exhissp}
\end{figure}

\begin{ex}[The full hysteresis case]
\label{ex:hisspsm}
Take $q_*^-=2$, $q_*=3$ and $q_*^+=3.5$. Compared to Examples~\ref{ex:nohis} and~\ref{ex:hissp}, both the opening and closing conditions are more restrictive. In particular, after some time $t>0$ the valve eventually opens and then no longer closes. We can still see an oscillatory behavior in $[0,1]$, with a low amplitude, as the trace of $q$ at $x=0^+$ is allowed to take values in the interval $[q_*^-, q_*]$, see Figure~\ref{fig:exhisspsm}.
\end{ex}

\begin{figure}[!htb]
\centering
\begin{minipage}{0.32\linewidth}
\centering
\includegraphics[trim=2.2cm 7.5cm 2.2cm 8.3cm, clip, width=\linewidth]{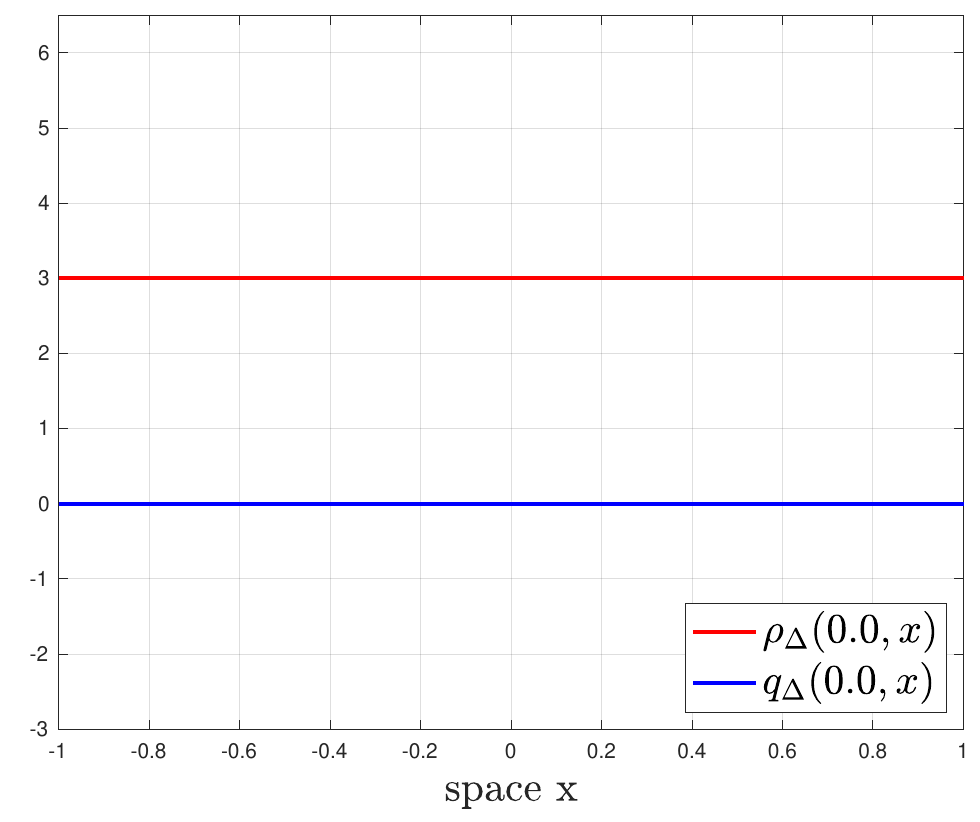}
\end{minipage}\hfill
\begin{minipage}{0.32\linewidth}
\centering
\includegraphics[trim=2.2cm 7.5cm 2.2cm 8.3cm, clip, width=\linewidth]{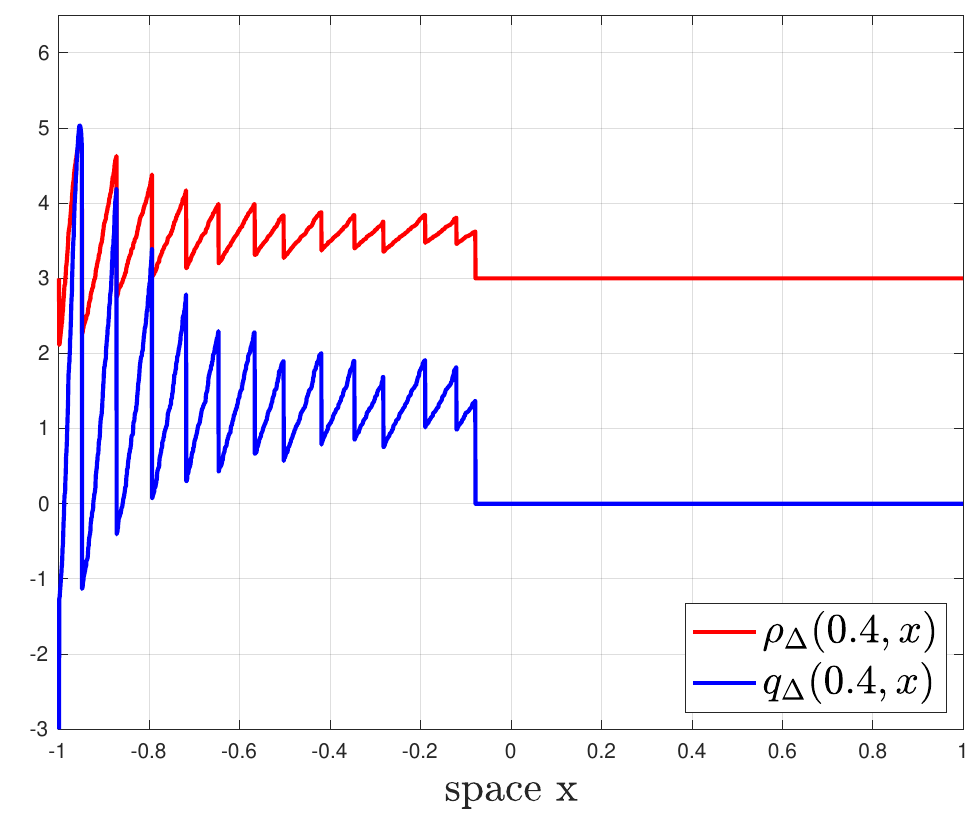}
\end{minipage}\hfill
\begin{minipage}{0.32\linewidth}
\centering
\includegraphics[trim=2.2cm 7.5cm 2.2cm 8.3cm, clip, width=\linewidth]{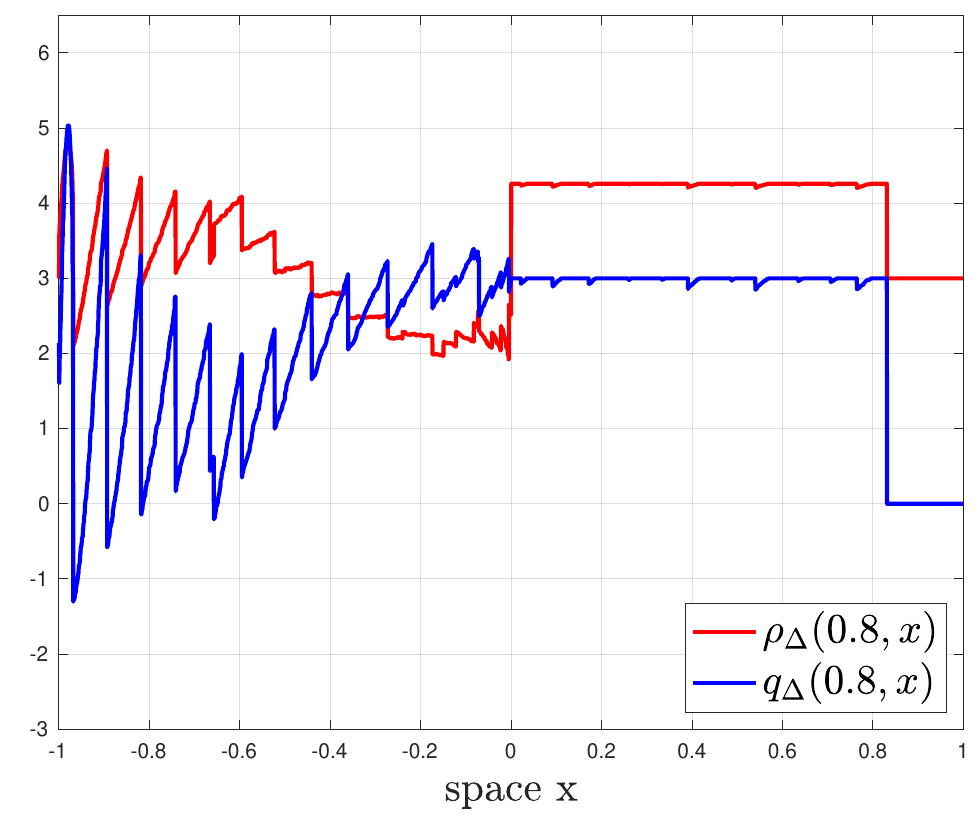}
\end{minipage}
\begin{minipage}{0.32\linewidth}
\centering
\includegraphics[trim=2.2cm 7.5cm 2.2cm 8.3cm, clip, width=\linewidth]{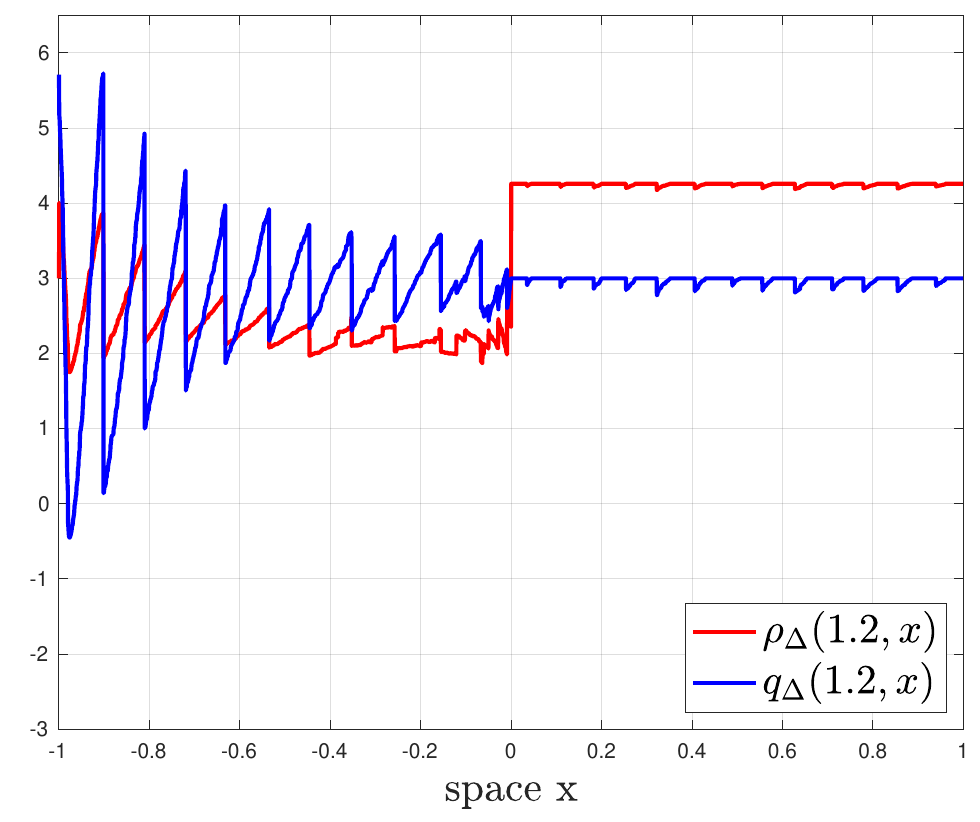}
\end{minipage}\hfill
\begin{minipage}{0.32\linewidth}
\centering
\includegraphics[trim=2.2cm 7.5cm 2.2cm 8.3cm, clip, width=\linewidth]{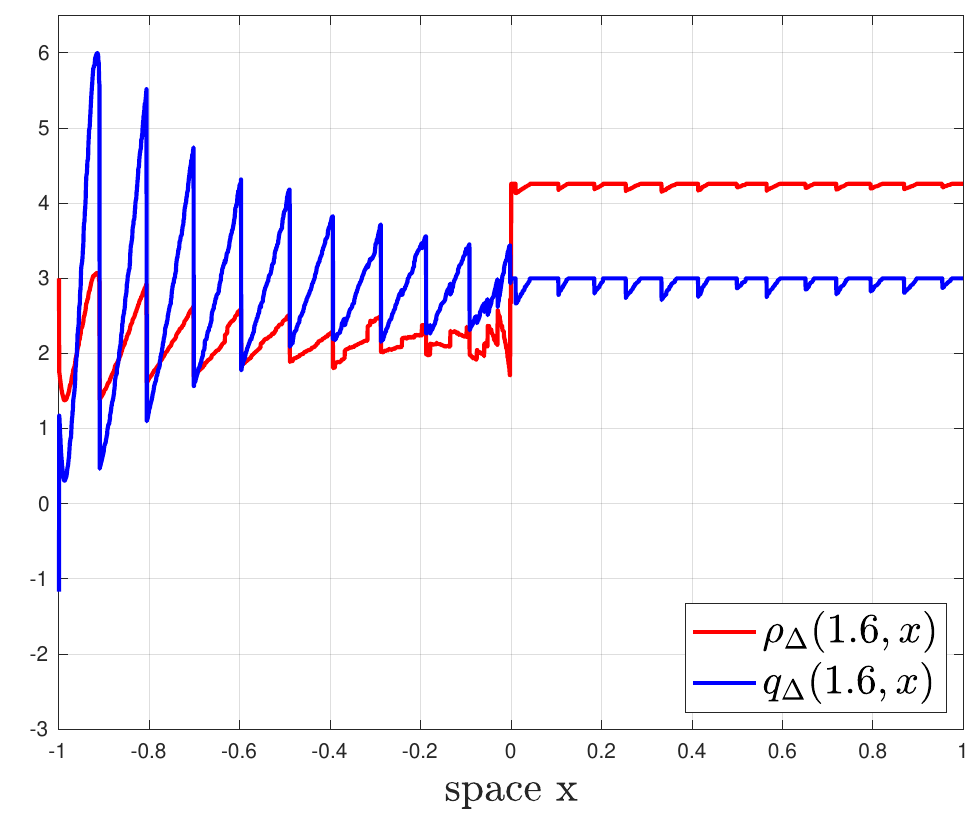}
\end{minipage}\hfill
\begin{minipage}{0.32\linewidth}
\centering
\includegraphics[trim=2.2cm 7.5cm 2.2cm 8.3cm, clip, width=\linewidth]{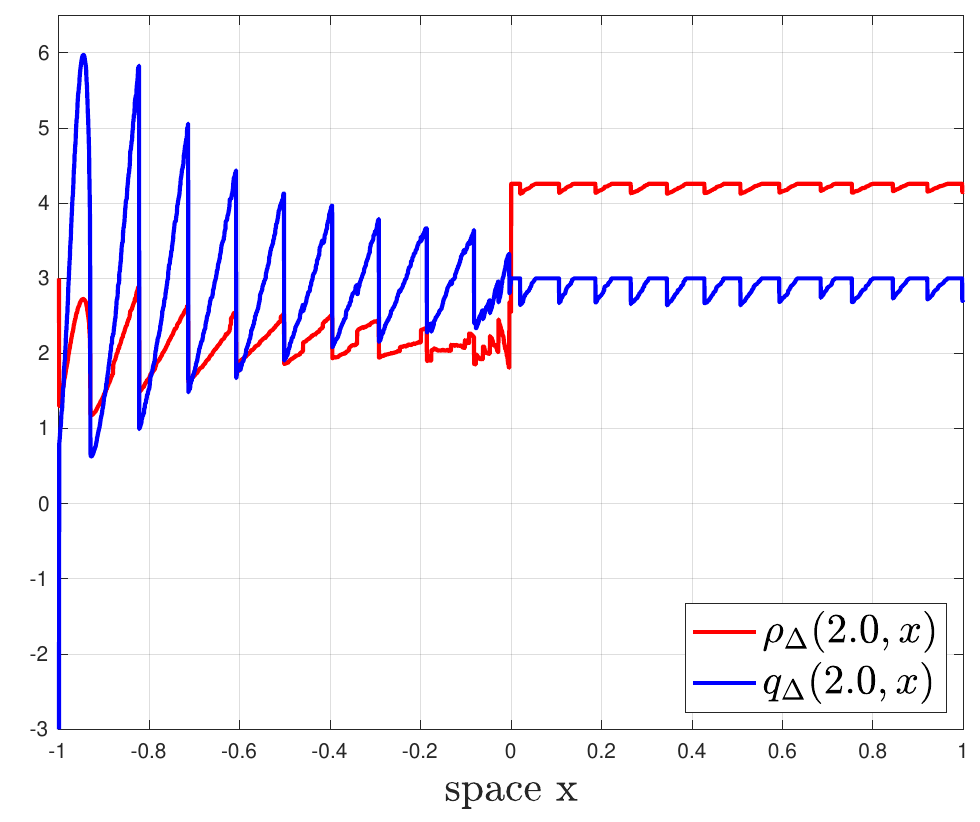}
\end{minipage}

\caption{Strong oscillatory inflow at $x=-1$ when $q_*^-<q_*<q_*^+$, see Example~\ref{ex:hisspsm}.}
\label{fig:exhisspsm}
\end{figure}

\begin{ex}[The case when the lower threshold is active]
\label{ex:hissm}
Take $q_*^-=2$ and $q_*^+=q_*=3$. This example is similar to Example~\ref{ex:hisspsm}. The main difference is that the valve opens earlier, since the threshold $q_*^+$ to pass from closed to open is lower. As before, once the valve opens, it remains open for all subsequent times, see Figure~\ref{fig:exhissm}.
\end{ex}

\begin{figure}[!htb]
\centering
\begin{minipage}{0.32\linewidth}
\centering
\includegraphics[trim=2.2cm 7.5cm 2.2cm 8.3cm, clip, width=\linewidth]{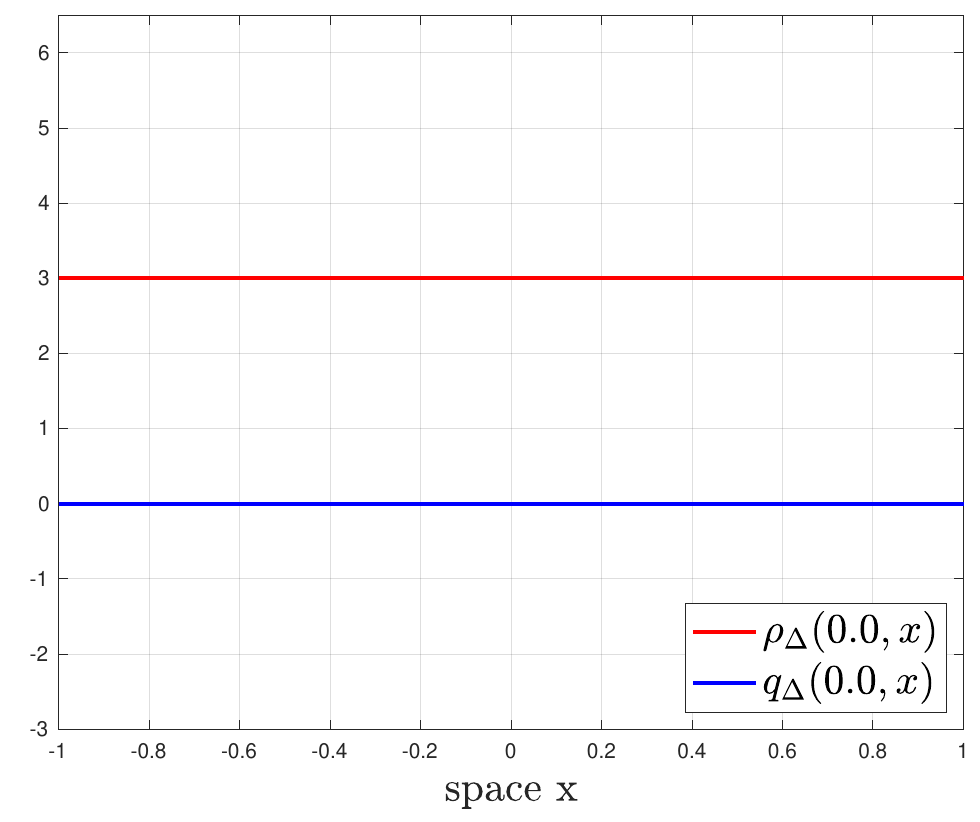}
\end{minipage}\hfill
\begin{minipage}{0.32\linewidth}
\centering
\includegraphics[trim=2.2cm 7.5cm 2.2cm 8.3cm, clip, width=\linewidth]{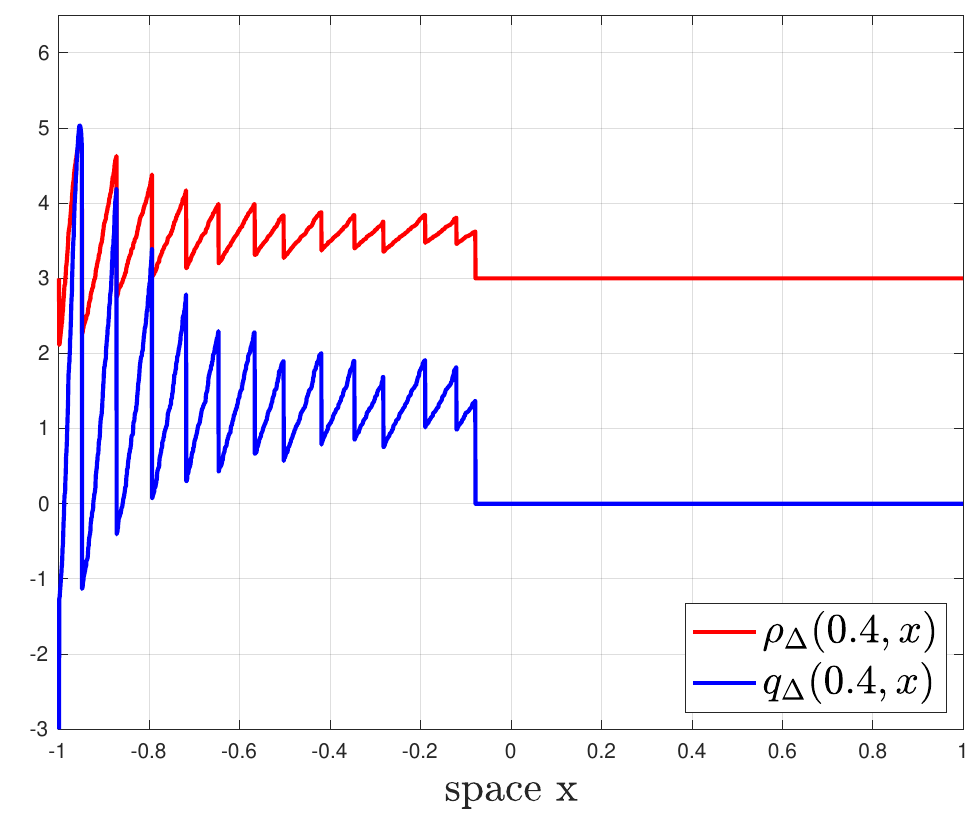}
\end{minipage}\hfill
\begin{minipage}{0.32\linewidth}
\centering
\includegraphics[trim=2.2cm 7.5cm 2.2cm 8.3cm, clip, width=\linewidth]{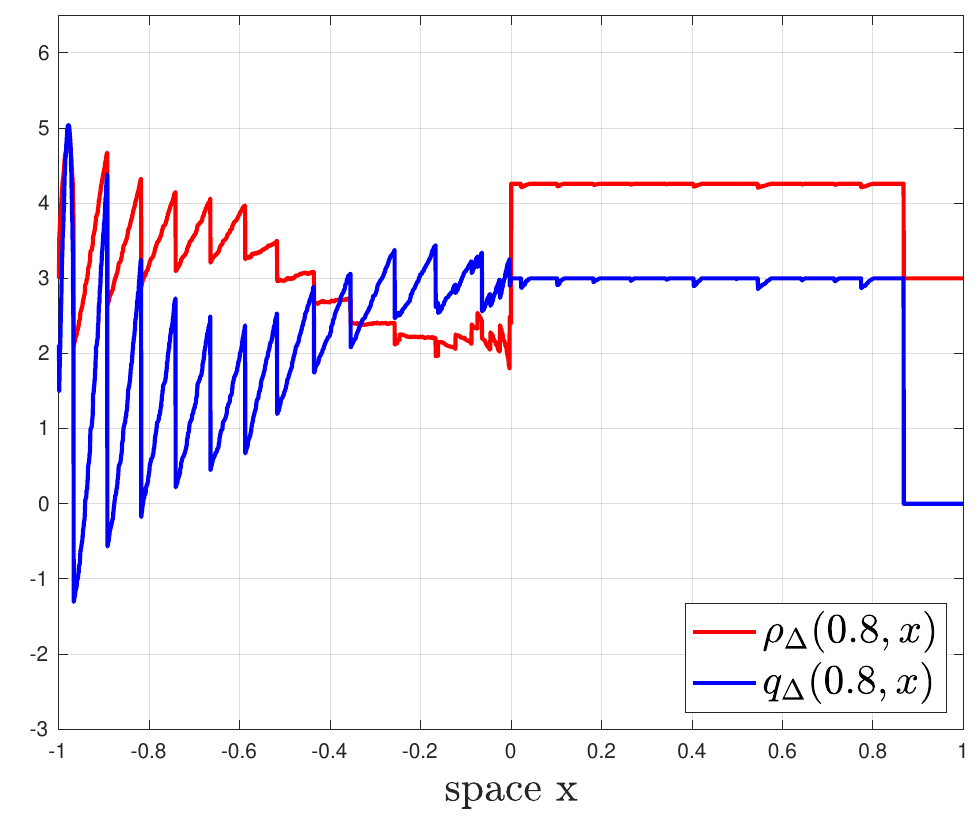}
\end{minipage}
\begin{minipage}{0.32\linewidth}
\centering
\includegraphics[trim=2.2cm 7.5cm 2.2cm 8.3cm, clip, width=\linewidth]{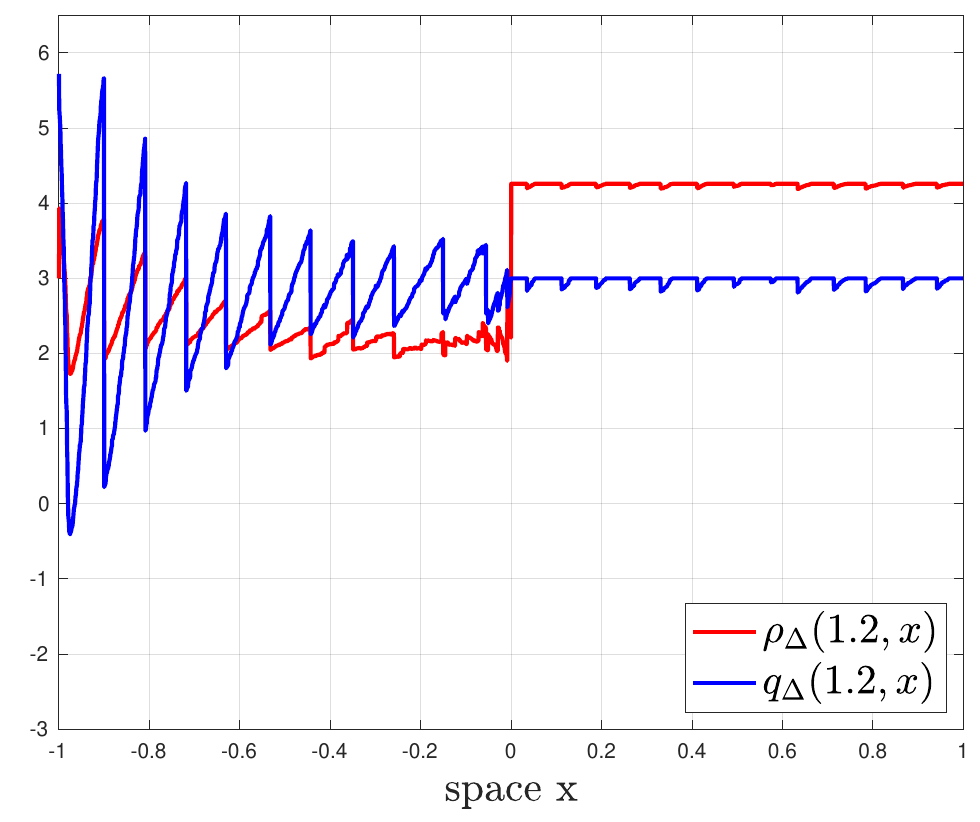}
\end{minipage}\hfill
\begin{minipage}{0.32\linewidth}
\centering
\includegraphics[trim=2.2cm 7.5cm 2.2cm 8.3cm, clip, width=\linewidth]{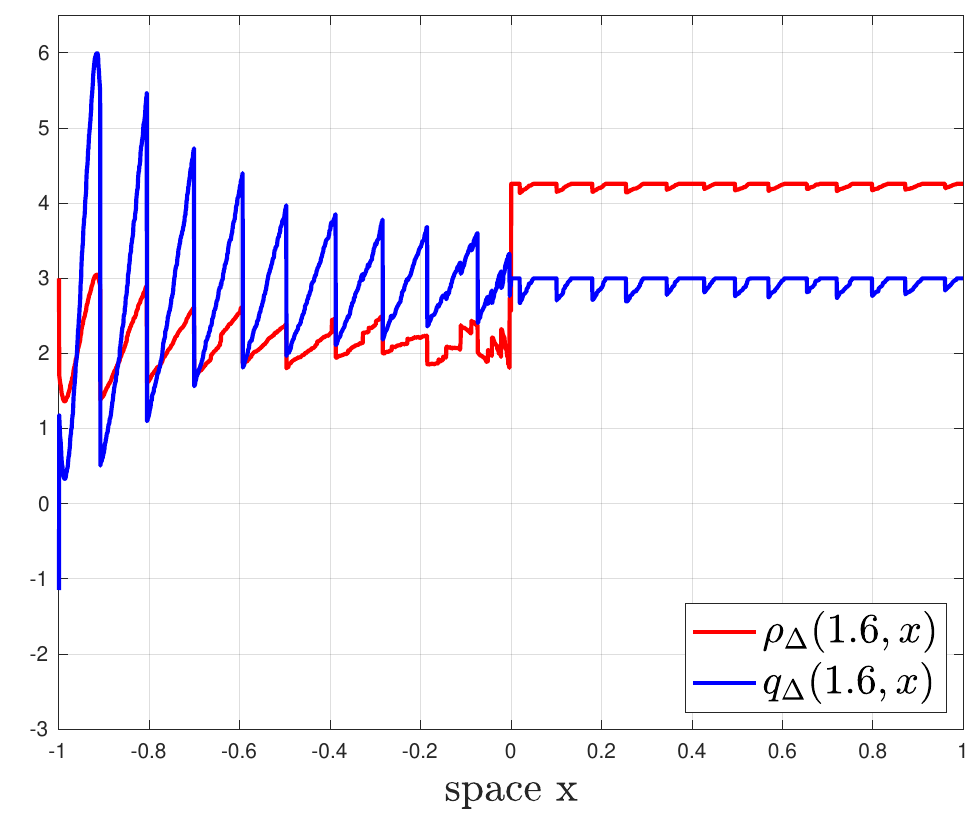}
\end{minipage}\hfill
\begin{minipage}{0.32\linewidth}
\centering
\includegraphics[trim=2.2cm 7.5cm 2.2cm 8.3cm, clip, width=\linewidth]{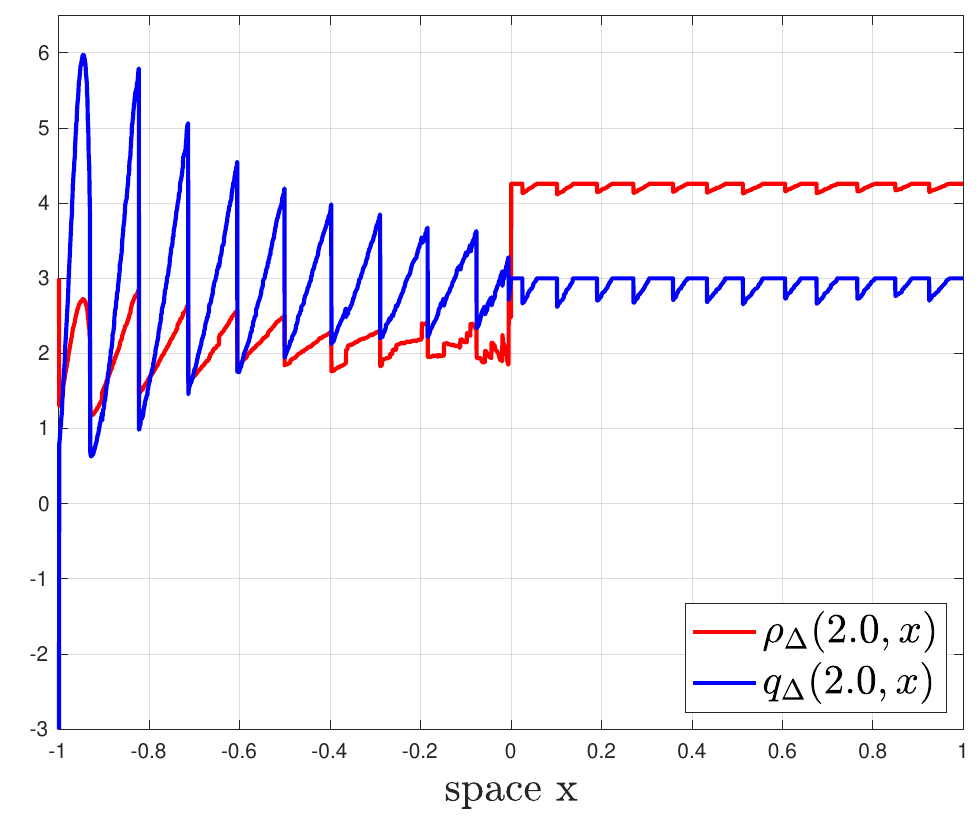}
\end{minipage}

\caption{Strong oscillatory inflow at $x=-1$ when $q_*^-<q_*=q_*^+$, see Example~\ref{ex:hissm}.}
\label{fig:exhissm}
\end{figure}

We showed in the previous plots that if we choose the parameters $q_*^-$ and $q_*^+$ different from $q_*$, then we have a regularizing effect on the outflow. This effect results both in a {\em smaller} spatial total variation of the variables $\rho$ and $q$ (both upstream and downstream of the valve), and in {\em fewer} switches between the open and closed states; they are measured by $\mathrm{TV}_{[0,t]}(h_\Delta)$. To further quantify the latter effect, we introduce the piecewise constant function $q_{\rm{c}\,\Delta}^0$, defined by $q_{\rm{c}\,\Delta}^0 \doteq q_{\rm{c}}^{0,n}$ on $[t^n,t^{n+1})$, where $q_{\rm{c}}^{0,n} \doteq \qc(u_{j_c}^n,h^n)$ is computed via \eqref{e:hc} and \eqref{e:qc}. We then evaluate $\mathrm{TV}_{[0,t]}(q_{\rm{c}\,\Delta}^0)$, which measures the oscillations of the flux imposed by the valve. We provide in Table~\ref{tab:example} the total variations at subsequent times of the states for the Examples~\ref{ex:nohis}-\ref{ex:hissm}.

\begin{table}[!htb]
\centering
\begin{tabular}{|l|c|| c c c c c|}
\hline
\multicolumn{2}{|l||}{\textbf{Example}}  &  $t=0.4$ & $t=0.8$ & $t=1.2$ & $t=1.6$ & $t=2.0$ \\
\hline

\multirow{3}{*}{\shortstack{Example~\ref{ex:nohis}:\\ \,$q_*^-=q_*=q_*^+$}} 
& $\mathrm{TV}_{[0,t]}(h_\Delta), \mathrm{TV}_{[0,t]}(q_{\rm{c}\,\Delta}^0)$ 
& 0, 0 & 23, 69 & 51, 153 & 77, 231 & 101, 303 \\

& $\mathrm{TV}_{[-1,0]}(\rho_\Delta(t,\cdot)), \mathrm{TV}_{[0,1]}(\rho_\Delta(t,\cdot))$
& 24.87, 0 & 32.68, 16.11 & 35.57, 21.10 & 35.35, 20.08 & 35.61, 19.22\\

& $\mathrm{TV}_{[-1,0]}(q_\Delta(t,\cdot)), \mathrm{TV}_{[0,1]}(q_\Delta(t,\cdot))$
& 59.28, 0 & 71.76, 39.85 & 79.86, 51.83 & 80.19, 49.18 & 83.27, 46.81  \\
\hline
\multirow{3}{*}{\shortstack{Example~\ref{ex:hissp}:\\ \,$q_*^-=q_*<q_*^+$}} 
& $\mathrm{TV}_{[0,t]}(h_\Delta), \mathrm{TV}_{[0,t]}(q_{\rm{c}\,\Delta}^0)$ 
& 0, 0 & 6, 18 & 8, 24 & 8, 24 & 10, 30 \\

& $\mathrm{TV}_{[-1,0]}(\rho_\Delta(t,\cdot)), \mathrm{TV}_{[0,1]}(\rho_\Delta(t,\cdot))$
& 24.87, 0 & 32.51, 7.56 & 34.95, 3.00 & 28.96, 0.52 & 30.84, 2.23\\

& $\mathrm{TV}_{[-1,0]}(q_\Delta(t,\cdot)), \mathrm{TV}_{[0,1]}(q_\Delta(t,\cdot))$
& 59.28, 0 & 74.86, 18.00 & 74.62, 7.02 & 65.80, 1.12 & 69.17, 5.22  \\
\hline

\multirow{3}{*}{\shortstack{Example~\ref{ex:hisspsm}:\\ \,$q_*^-<q_*<q_*^+$}} 
& $\mathrm{TV}_{[0,t]}(h_\Delta), \mathrm{TV}_{[0,t]}(q_{\rm{c}\,\Delta}^0)$ 
& 0, 0 & 1, 4.82 & 1, 9.19 & 1, 15.57 & 1, 25.01 \\

& $\mathrm{TV}_{[-1,0]}(\rho_\Delta(t,\cdot)), \mathrm{TV}_{[0,1]}(\rho_\Delta(t,\cdot))$
& 24.87, 0 & 26.50, 1.88 & 22.27, 1.43 & 21.33, 2.08 & 21.45, 2.99\\

& $\mathrm{TV}_{[-1,0]}(q_\Delta(t,\cdot)), \mathrm{TV}_{[0,1]}(q_\Delta(t,\cdot))$
& 59.28, 0 & 60.35, 4.69 & 56.98, 3.85 & 56.85, 5.59 & 59.77, 8.00 \\
\hline

\multirow{3}{*}{\shortstack{Example~\ref{ex:hissm}:\\ \,$q_*^-<q_*=q_*^+$}} 
& $\mathrm{TV}_{[0,t]}(h_\Delta), \mathrm{TV}_{[0,t]}(q_{\rm{c}\,\Delta}^0)$ 
& 0, 0 & 1, 4.76 & 1, 8.98 & 1, 16.76 & 1, 25.73 \\

& $\mathrm{TV}_{[-1,0]}(\rho_\Delta(t,\cdot)), \mathrm{TV}_{[0,1]}(\rho_\Delta(t,\cdot))$
& 24.87, 0 & 25.92, 1.83 & 21.71, 1.40 & 20.78, 2.36 & 21.19, 2.76\\

& $\mathrm{TV}_{[-1,0]}(q_\Delta(t,\cdot)), \mathrm{TV}_{[0,1]}(q_\Delta(t,\cdot))$
& 59.28, 0 & 58.76, 4.54 & 55.86, 3.77 & 56.26, 6.33 & 59.99, 7.39  \\
\hline

\end{tabular}
\caption{Total variations of the states in the Examples~\ref{ex:nohis}-\ref{ex:hissm}.}
\label{tab:example}
\end{table}

First, notice that we have $\mathrm{TV}_{[0,t]}(q_{\rm{c}\,\Delta}^0)=q_*\,\mathrm{TV}_{[0,t]}(h_\Delta)$ for Examples~\ref{ex:nohis} and \ref{ex:hissp}, since the flux either jumps between $0$ and $q_*$ or remains constant. The data show that Example~\ref{ex:nohis} exhibits the largest number of switches and the highest variation in the flux imposed by the valve, followed by Example~\ref{ex:hissp}. On the other hand, Examples~\ref{ex:hisspsm} and \ref{ex:hissm} exhibit only one switch, as already observed in Figures~\ref{fig:exhisspsm} and~\ref{fig:exhissm}. 
Although in these two cases the fluxes at the valve, when the valve is open, are allowed to vary in the interval $[q_*^-,q_*]$, nevertheless we obtain values of $\mathrm{TV}_{[0,t]}(q_{\rm{c}\,\Delta}^0)$ smaller than in the two previous examples, and in particular smaller than in Example~\ref{ex:hissp}. 
This is due to the smaller number of switches. 
Regarding the downstream total variation of the variables, namely $\mathrm{TV}_{[0,1]}(\rho_\Delta(t,\cdot))$ and $\mathrm{TV}_{[0,1]}(q_\Delta(t,\cdot))$, Examples~\ref{ex:hisspsm} and \ref{ex:hissm} perform better at $t=0.8$ and $t=1.2$, whereas Example~\ref{ex:hissp} performs better at $t=1.6$ and $t=2$. The latter occurs when the valve in that example is closed for a certain period of time. However, at these times the upstream total variation, namely $\mathrm{TV}_{[-1,0]}(\rho_\Delta(t,\cdot))$ and $\mathrm{TV}_{[-1,0]}(q_\Delta(t,\cdot))$, is larger in Example~\ref{ex:hissp} than in Examples~\ref{ex:hisspsm} and \ref{ex:hissm}, since the closed valve forces the formation of stronger waves upstream. Finally, there is no significant qualitative difference between Examples~\ref{ex:hisspsm} and \ref{ex:hissm}; however we point out that, as already mentioned, the valve opens earlier in Example~\ref{ex:hisspsm}.

As suggested by the previous examples, the choice of $q_*^-<q_*$ has the effect of reducing both the number of valve switches and the total variation of the flux imposed by the valve. However, this choice allows $q_{\rm{c}\,\Delta}^0$ to take values in $[q_*^-,q_*]\cup\{0\}$. From a practical point of view, this means that, when the valve is open, the flux is allowed to differ from the desired value $q_*$. To quantify this distance, we introduce the functional
\[
F(t, u_{\Delta})=\int_0^t h(s)\,|q_*-q_{\rm{c}\,\Delta}^0(s)|\,\d s = \sum_{n=1}^N \Delta t_{n+1} h^n \, | q_* - q_{\rm{c}}^{0,n}|,
\]
which measures the $\L1$-distance of the flux at the valve from the desired flux while the valve is open. Clearly, we have $F(t, u_\Delta)\equiv0$ for Examples~\ref{ex:nohis} and \ref{ex:hissp}, since $q_*^-=q_*$. The values of $F(t,u_\Delta)$ for the remaining two examples are reported in Table~\ref{tab:funzionale}. Once again, no significant difference can be observed between the two cases.

\begin{table}[!htb]
\centering
\begin{tabular}{|l|c|| c c c c c|}
\hline
\multicolumn{2}{|l||}{\textbf{Example}}  &  $t=0.4$ & $t=0.8$ & $t=1.2$ & $t=1.6$ & $t=2.0$ \\
\hline

{Example~\ref{ex:hisspsm}: $q_*^-<q_*<q_*^+$}
& $F(t, u_{\Delta})$ 
& 0 & 0.0031 & 0.0113 & 0.0261 & 0.0587 \\

\hline

Example~\ref{ex:hissm}: $q_*^-<q_*=q_*^+$
& $F(t, u_{\Delta})$ 
& 0 & 0.0029 & 0.0115 & 0.0330 &  0.0621 \\

\hline

\end{tabular}
\caption{$F(t, u_{\Delta})$ for Examples~\ref{ex:hisspsm} and \ref{ex:hissm}.}
\label{tab:funzionale}
\end{table}

Therefore, on the one hand, choosing $q_*^-<q_*$ reduces the number of switches of the valve while, on the other hand, it increases the deviation of the flux from its desired value. Investigating how to balance these competing effects by selecting optimal hysteresis parameters naturally leads to an interesting optimal control problem to be considered in a forthcoming paper.

\section{Final comments}
\label{s:final}
This final section provides some concluding remarks on the topic of this paper; then we present our main conclusions.

The input-output relation in \eqref{e:hrelay} is the \lq\lq closed version\rq\rq\ of the so-called {\em delayed relay} (compare with the definition of $h_{\rho}^*$ on page~102 of \cite{vis}), where the switch of $h$ occurs only when the input flow $\overline Q(u_\ell^0)$ bypasses the corresponding thresholds $q_*^\pm$. If the input reaches the corresponding threshold but then reverses its behavior, then no switch occurs. The \lq\lq open version\rq\rq\ of the delayed relay is obtained by replacing $\cal L$ given in \eqref{e:L} with
\[
{\cal L}^\circ=\left([0,q_*^+)\times\{0\}\right)\cup\left((q_*^-, +\infty)\times\{1\}\right)\subseteq\R^2,
\]
and $\mathcal{T}^-_t,\mathcal{T}_t^+$ given in \eqref{e:TT} respectively with
\begin{align*}
\left(\mathcal{T}_t^-\right)^\circ &= \left\{\tau\in[0,t] : \overline{Q}\left(u_\ell^0(\tau)\right) \le q_*^-\right\},&
\left(\mathcal{T}_t^+\right)^\circ &= \left\{\tau\in[0,t] : \overline{Q}\left(u_\ell^0(\tau)\right) \ge q_*^+\right\}.
\end{align*}
In this way, the switch occurs exactly when the input $\overline Q(u_\ell^0)$ reaches the corresponding threshold. This version of the problem can be treated analogously. 

Another way to close the delayed relay is to leave the output free to switch on or off when the input is on the corresponding threshold; moreover, if it switches, it cannot switch back unless the input reaches the other threshold, and it must also switch if the input bypasses the corresponding threshold. This version of the delayed relay gives multiplicity of the outputs but in many cases, related to convergence and approximation in differential equations, it is the suitable version of the delayed relay to be considered in the limit process. Actually, it may be seen as the closure of the graph of the delayed relay as a functional between suitable function spaces (especially with respect to the uniform convergence of the inputs), see definition of $\overline{h}_\rho$ at page 153 and Theorem 1.1 at page 154 of \cite{vis}. In particular, this version of multivalued delayed relay may be suitable to be considered when studying the convergence of the numerical scheme. 

\smallskip

Finally, we briefly draw our main conclusions; we refer to the Introduction for the statement of the problem.  

\begin{itemize}

\item From a theoretical point of view, we have proposed a model for the gas flow through a valve incorporating hysteresis effects. We completely and explicitly characterized the incoherent regions; these are regions where a kind of instability is hidden and are understood to be associated to the chattering of the valve. These regions strongly depend on the initial state of the valve (open or closed).

\item From a numerical point of view, we have introduced an essentially first-order scheme that captures solutions associated to coherent states. Incoherent states show up different time evolutions of the solution: chattering can persist and amplify over time or disappear. A theoretical explanation is provided for each case and how the numerical solution differs from the exact one.

\item Hysteresis plays a major role when strongly oscillatory flows are concerned. Numerical simulations show that the upper-threshold $q_*^+$ is responsible for killing the periodicity of the flow, keeping however sporadic and \lq\lq random\rq\rq\ switches of the valve. The lower-threshold $q_*^-$, on the contrary, preserves the periodicity of the flow but reduces its amplitude. The combined effect of the two thresholds reduces the total variation of the solution.  

\end{itemize}

\section*{Acknowledgments}
The authors are members of the {\em Gruppo Nazionale per l'Analisi Matematica, la Probabilit\`{a} e le loro Applicazioni} (GNAMPA) of the {\em Istituto Nazionale di Alta Matematica} (INdAM) and acknowledge financial support from this institution. A.C.\ and M.D.R.\ were supported by the PRIN 2022 project \emph{Modeling, Control and Games through Partial Differential Equations}, funded by the European Union - Next Generation EU, Mission 4, Component 1, CUP D53D23005620006. F.B., A.C. and S.M. were supported by the GNAMPA 2026 Project {\em Analisi e controllo per alcuni problemi di evoluzione (CUP E53C25002010001)}.

{\small
\bibliography{biblio.bib}

@book{LeVeque-book,
	address = {Basel},
	author = {LeVeque, Randall J.},
	isbn = {3-7643-2464-3},
	mrclass = {65Mxx (35L65 35L67 65P05)},
	mrnumber = {1077828 (91j:65142)},
	mrreviewer = {Carlos A. de Moura},
	optseries = {Lectures in Mathematics ETH Z\"urich},
	pages = {x+214},
	publisher = {Birkh\"auser Verlag},
	title = {Numerical methods for conservation laws},
	year = {1990}}

@article {AMP,
    AUTHOR = {Agrawal, Jyoti and Moudgalya, Kannan M. and Pani, Amiya K.},
     TITLE = {Sliding motion and stability of a class of discontinuous
              dynamical systems},
   JOURNAL = {Nonlinear Dynam.},
  FJOURNAL = {Nonlinear Dynamics. An International Journal of Nonlinear
              Dynamics and Chaos in Engineering Systems},
    VOLUME = {37},
      YEAR = {2004},
    NUMBER = {2},
     PAGES = {151--168},
      ISSN = {0924-090X,1573-269X},
   MRCLASS = {76T99 (34A36 34D20 76E99)},
  MRNUMBER = {2098828},
       DOI = {10.1023/B:NODY.0000042910.39609.f6},
       URL = {https://doi.org/10.1023/B:NODY.0000042910.39609.f6},
}

@article {bagmor,
    AUTHOR = {Bagagiolo, Fabio and Moreti, Stefan},
     TITLE = {Wave-front tracking for a quasi-linear scalar conservation law
              with hysteresis},
   JOURNAL = {J. Math. Anal. Appl.},
  FJOURNAL = {Journal of Mathematical Analysis and Applications},
    VOLUME = {543},
      YEAR = {2025},
    NUMBER = {2, part 1},
     PAGES = {Paper No. 128900, 25},
      ISSN = {0022-247X,1096-0813},
   MRCLASS = {35L65 (35L03)},
  MRNUMBER = {4803782},
MRREVIEWER = {Richard\ De la cruz},
       DOI = {10.1016/j.jmaa.2024.128900},
       URL = {https://doi.org/10.1016/j.jmaa.2024.128900},
}

@article {bagmor2,
    AUTHOR = {Bagagiolo, Fabio and Moreti, Stefan},
     TITLE = {Wave-front tracking for a quasi-linear scalar conservation law
              with hysteresis {II}: {T}he case of {P}reisach},
   JOURNAL = {J. Math. Anal. Appl.},
  FJOURNAL = {Journal of Mathematical Analysis and Applications},
    VOLUME = {558},
      YEAR = {2026},
    NUMBER = {1},
     PAGES = {29},
      ISSN = {0022-247X,1096-0813},
   MRCLASS = {35L65},
  MRNUMBER = {5005579},
       DOI = {10.1016/j.jmaa.2025.130345},
       URL = {https://doi.org/10.1016/j.jmaa.2025.130345},
}

@article {Bazso-Champneys-Hos,
    AUTHOR = {Bazs\'{o}, Csaba and Champneys, Alan R. and H\H{o}s, Csaba J.},
     TITLE = {Bifurcation analysis of a simplified model of a pressure
              relief valve attached to a pipe},
   JOURNAL = {SIAM J. Appl. Dyn. Syst.},
  FJOURNAL = {SIAM Journal on Applied Dynamical Systems},
    VOLUME = {13},
      YEAR = {2014},
    NUMBER = {2},
     PAGES = {704--721},
      ISSN = {1536-0040},
   MRCLASS = {37G15 (34C60 37M20 70B15)},
  MRNUMBER = {3199407},
MRREVIEWER = {Weinian\ Zhang},
       DOI = {10.1137/130922598},
       URL = {https://doi.org/10.1137/130922598},
}

@article {BHJ,
    AUTHOR = {Bi\'{a}k, Martin and Hanus, Tom\'{a}\v{s} and Janovsk\'{a},
              Drahoslava},
     TITLE = {Some applications of {F}ilippov's dynamical systems},
   JOURNAL = {J. Comput. Appl. Math.},
  FJOURNAL = {Journal of Computational and Applied Mathematics},
    VOLUME = {254},
      YEAR = {2013},
     PAGES = {132--143},
      ISSN = {0377-0427,1879-1778},
   MRCLASS = {34A36 (34C23 37M05 65P30)},
  MRNUMBER = {3061071},
MRREVIEWER = {Inmaculada\ Baldom\'{a}},
       DOI = {10.1016/j.cam.2013.03.034},
       URL = {https://doi.org/10.1016/j.cam.2013.03.034},
}

@book {Bressan-book,
    AUTHOR = {Bressan, Alberto},
     TITLE = {Hyperbolic systems of conservation laws},
    SERIES = {Oxford Lecture Series in Mathematics and its Applications},
    VOLUME = {20},
      NOTE = {The one-dimensional Cauchy problem},
 PUBLISHER = {Oxford University Press, Oxford},
      YEAR = {2000},
     PAGES = {xii+250},
      ISBN = {0-19-850700-3},
   MRCLASS = {35-02 (35B35 35L65)},
  MRNUMBER = {1816648},
MRREVIEWER = {Denis\ Serre},
}

@book {brospre,
    AUTHOR = {Brokate, Martin and Sprekels, J\"{u}rgen},
     TITLE = {Hysteresis and phase transitions},
    SERIES = {Applied Mathematical Sciences},
    VOLUME = {121},
 PUBLISHER = {Springer-Verlag, New York},
      YEAR = {1996},
     PAGES = {x+357},
      ISBN = {0-387-94763-9},
   MRCLASS = {35Qxx (00A69 47N20 49S05 73B05 73B30 80A22)},
  MRNUMBER = {1411908},
MRREVIEWER = {Pavel\ Krej\v{c}\'{\i}},
       DOI = {10.1007/978-1-4612-4048-8},
       URL = {https://doi.org/10.1007/978-1-4612-4048-8},
}

@article{Colombo-Goatin,
	author = {Colombo, Rinaldo M. and Goatin, Paola},
	coden = {JDEQAK},
	doi = {10.1016/j.jde.2006.10.014},
	fjournal = {Journal of Differential Equations},
	issn = {0022-0396},
	journal = {J. Differential Equations},
	mrclass = {35L65 (35B30 35L45 90B20)},
	mrnumber = {2300671 (2007k:35317)},
	mrreviewer = {Michael Herty},
	number = {2},
	pages = {654--675},
	title = {A well posed conservation law with a variable unilateral constraint},
	url = {http://dx.doi.org/10.1016/j.jde.2006.10.014},
	volume = {234},
	year = {2007}}

@article{Colombo-Garavello_notions,
	author = {Colombo, Rinaldo M. and Garavello, Mauro},
	fjournal = {Discrete and Continuous Dynamical Systems. Series A},
	isbn = {978-1-60133-011-6; 1-60133-011-1},
	issn = {1078-0947},
	journal = {Discrete Contin. Dyn. Syst.},
	mrclass = {35L65 (76N10)},
	mrnumber = {2641394},
	optnumber = {Dynamical systems, differential equations and applications. 7th AIMS Conference, suppl.},
	pages = {181--190},
	title = {Comparison among different notions of solution for the {$p$}-system at a junction},
	volume = {Supplement},
	year = {2009}}

@article {Corli-Fan_2019,
    AUTHOR = {Corli, Andrea and Fan, Haitao},
     TITLE = {Hysteresis and stop-and-go waves in traffic flows},
   JOURNAL = {Math. Models Methods Appl. Sci.},
  FJOURNAL = {Mathematical Models and Methods in Applied Sciences},
    VOLUME = {29},
      YEAR = {2019},
    NUMBER = {14},
     PAGES = {2637--2678},
      ISSN = {0218-2025},
   MRCLASS = {35L65 (35C07 76A30 90B20)},
  MRNUMBER = {4053240},
MRREVIEWER = {Jean-Fran\c{c}ois Coulombel},
       DOI = {10.1142/S0218202519500568},
       URL = {https://doi.org/10.1142/S0218202519500568},
}

@article {C1,
    AUTHOR = {Corli, Andrea and Razafison, Ulrich and Rosini, Massimiliano
              D.},
     TITLE = {Coherence and flow-maximization of a one-way valve},
   JOURNAL = {ESAIM Math. Model. Numer. Anal.},
  FJOURNAL = {ESAIM. Mathematical Modelling and Numerical Analysis},
    VOLUME = {56},
      YEAR = {2022},
    NUMBER = {5},
     PAGES = {1715--1739},
      ISSN = {2822-7840,2804-7214},
   MRCLASS = {76B75 (35L65 35L67)},
  MRNUMBER = {4454160},
MRREVIEWER = {Piotr\ Biler},
       DOI = {10.1051/m2an/2022053},
       URL = {https://doi.org/10.1051/m2an/2022053},
}

@article {C2,
    AUTHOR = {Corli, Andrea and Rosini, Massimiliano D.},
     TITLE = {Coherence and chattering of a one-way valve},
   JOURNAL = {ZAMM Z. Angew. Math. Mech.},
  FJOURNAL = {ZAMM. Zeitschrift f\"{u}r Angewandte Mathematik und Mechanik.
              Journal of Applied Mathematics and Mechanics},
    VOLUME = {99},
      YEAR = {2019},
    NUMBER = {6},
     PAGES = {e201800250, 25},
      ISSN = {0044-2267,1521-4001},
   MRCLASS = {35L65 (35L67 74F10 76N15)},
  MRNUMBER = {3973090},
       DOI = {10.1002/zamm.201800250},
       URL = {https://doi.org/10.1002/zamm.201800250},
}

@article {C3,
    AUTHOR = {Corli, Andrea and Figiel, Magdalena and Futa, Anna and Rosini,
              Massimiliano D.},
     TITLE = {Coupling conditions for isothermal gas flow and applications
              to valves},
   JOURNAL = {Nonlinear Anal. Real World Appl.},
  FJOURNAL = {Nonlinear Analysis. Real World Applications. An International
              Multidisciplinary Journal},
    VOLUME = {40},
      YEAR = {2018},
     PAGES = {403--427},
      ISSN = {1468-1218,1878-5719},
   MRCLASS = {76N15},
  MRNUMBER = {3718991},
MRREVIEWER = {Bernard\ Ducomet},
       DOI = {10.1016/j.nonrwa.2017.09.005},
       URL = {https://doi.org/10.1016/j.nonrwa.2017.09.005},
}

@book {Filippov,
    AUTHOR = {Filippov, A. F.},
     TITLE = {Differential equations with discontinuous righthand sides},
    SERIES = {Mathematics and its Applications (Soviet Series)},
    VOLUME = {18},
      NOTE = {Translated from the Russian},
 PUBLISHER = {Kluwer Academic Publishers Group, Dordrecht},
      YEAR = {1988},
     PAGES = {x+304},
      ISBN = {90-277-2699-X},
   MRCLASS = {34-02 (58F99)},
  MRNUMBER = {1028776},
       DOI = {10.1007/978-94-015-7793-9},
       URL = {https://doi.org/10.1007/978-94-015-7793-9},
}

@article{goamor,
       title   = {Well-posedness and {Godunov} approximation for nonlinear conservation laws with hysteresis},
      author={Paola Goatin and Stefan Moreti},
      year={2026},
       journal = {Nonlinear Differential Equations and Applications NoDEA},
       year    = {2026},
       volume  = {33},
       number  = {6},
       pages   = {131},
       doi     = {10.1007/s00030-026-01271-7},
       url     = {https://doi.org/10.1007/s00030-026-01271-7}
}

@article{Gugat-Herty,
	author = {Gugat, Martin and Herty, Micha\"{e}l},
	doi = {10.1051/cocv/2009035},
	fjournal = {ESAIM. Control, Optimisation and Calculus of Variations},
	issn = {1292-8119},
	journal = {ESAIM Control Optim. Calc. Var.},
	mrclass = {93C20 (35A09 35Q35 76N25 93D15)},
	mrnumber = {2775185},
	mrreviewer = {Farid Ammar Khodja},
	number = {1},
	pages = {28--51},
	title = {Existence of classical solutions and feedback stabilization for the flow in gas networks},
	url = {https://doi.org/10.1051/cocv/2009035},
	volume = {17},
	year = {2011}}

@article{Gugat-Herty-Muller,
	author = {Gugat, Martin and Herty, Michael and M\"{u}ller, Siegfried},
	doi = {10.3934/nhm.2017016},
	fjournal = {Networks and Heterogeneous Media},
	issn = {1556-1801},
	journal = {Netw. Heterog. Media},
	mrclass = {35L65 (35L60 35Q35 35R02 76J20)},
	mrnumber = {3714975},
	mrreviewer = {Alberto Valli},
	number = {3},
	pages = {371--380},
	title = {Coupling conditions for the transition from supersonic to subsonic fluid states},
	url = {https://doi.org/10.3934/nhm.2017016},
	volume = {12},
	year = {2017}}

@article{Herty_compressors,
	author = {Herty, Michael},
	doi = {10.3934/nhm.2007.2.81},
	fjournal = {Networks and Heterogeneous Media},
	issn = {1556-1801},
	journal = {Netw. Heterog. Media},
	mrclass = {76N15 (35L65 49N90 76N25)},
	mrnumber = {2291813},
	number = {1},
	pages = {81--97},
	title = {Modeling, simulation and optimization of gas networks with compressors},
	url = {https://doi.org/10.3934/nhm.2007.2.81},
	volume = {2},
	year = {2007}}

@article {Hos-Bazso-Champneys,
    AUTHOR = {H\H{o}s, Csaba and Bazs\'{o}, Csaba and Champneys, Alan},
     TITLE = {Model reduction of a direct spring-loaded pressure relief
              valve with upstream pipe},
   JOURNAL = {IMA J. Appl. Math.},
  FJOURNAL = {IMA Journal of Applied Mathematics},
    VOLUME = {80},
      YEAR = {2015},
    NUMBER = {4},
     PAGES = {1009--1024},
      ISSN = {0272-4960,1464-3634},
   MRCLASS = {70J35 (76M20 76N15)},
  MRNUMBER = {3384427},
       DOI = {10.1093/imamat/hxu034},
       URL = {https://doi.org/10.1093/imamat/hxu034},
}

@incollection {Krejci-SPP,
    AUTHOR = {Krej\v{c}\'{\i}, P.},
     TITLE = {Hysteresis in singularly perturbed problems},
 BOOKTITLE = {Singular perturbations and hysteresis},
     PAGES = {73--100},
 PUBLISHER = {SIAM, Philadelphia, PA},
      YEAR = {2005},
      ISBN = {0-89871-597-0},
   MRCLASS = {34C55 (34C20 34C60 34D15 47J40)},
  MRNUMBER = {2171578},
MRREVIEWER = {Chein-Shan\ Liu},
       DOI = {10.1137/1.9780898717860.ch3},
       URL = {https://doi.org/10.1137/1.9780898717860.ch3},
}

@article{Modesti-Pirozzoli,
	author = {Modesti, Davide and Pirozzoli, Sergio},
	journal = {Int. J. Heat Fluid Flow},
	pages = {100--112},
	title = {Direct numerical simulation of supersonic pipe flow at moderate {R}eynolds number},
	volume = {76},
	year = {2019}}

@article {Recupero,
    AUTHOR = {Recupero, Vincenzo},
     TITLE = {B{V}-extension of rate independent operators},
   JOURNAL = {Math. Nachr.},
  FJOURNAL = {Mathematische Nachrichten},
    VOLUME = {282},
      YEAR = {2009},
    NUMBER = {1},
     PAGES = {86--98},
      ISSN = {0025-584X,1522-2616},
   MRCLASS = {47H30 (34C55 47J30 47N20 74N30)},
  MRNUMBER = {2473132},
MRREVIEWER = {P.\ P.\ Zabre\u{\i}ko},
       DOI = {10.1002/mana.200610723},
       URL = {https://doi.org/10.1002/mana.200610723},
}

@article {TPN,
    AUTHOR = {Theodosiou, C. and Pournaras, A. and Natsiavas, S.},
     TITLE = {On periodic steady state response and stability of
              {F}ilippov-type mechanical models},
   JOURNAL = {Nonlinear Dynam.},
  FJOURNAL = {Nonlinear Dynamics. An International Journal of Nonlinear
              Dynamics and Chaos in Engineering Systems},
    VOLUME = {66},
      YEAR = {2011},
    NUMBER = {3},
     PAGES = {355--376},
      ISSN = {0924-090X,1573-269X},
   MRCLASS = {70E18 (34A36 34C25)},
  MRNUMBER = {2846603},
       DOI = {10.1007/s11071-011-0080-9},
       URL = {https://doi.org/10.1007/s11071-011-0080-9},
}

@book {vis,
    AUTHOR = {Visintin, Augusto},
     TITLE = {Differential models of hysteresis},
    SERIES = {Applied Mathematical Sciences},
    VOLUME = {111},
 PUBLISHER = {Springer-Verlag, Berlin},
      YEAR = {1994},
     PAGES = {xii+407},
      ISBN = {3-540-54793-2},
   MRCLASS = {47-02 (34A60 34G20 35R99 47H30 58-02 73E50)},
  MRNUMBER = {1329094},
MRREVIEWER = {Martin\ Brokate},
       DOI = {10.1007/978-3-662-11557-2},
       URL = {https://doi.org/10.1007/978-3-662-11557-2},
}

@article {CRR-MMAS-2025,
    AUTHOR = {Corli, Andrea and Razafison, Ulrich and Rosini, Massimiliano
              D.},
     TITLE = {Coherence of coupling conditions for the isothermal {E}uler
              system},
   JOURNAL = {Math. Methods Appl. Sci.},
  FJOURNAL = {Mathematical Methods in the Applied Sciences},
    VOLUME = {48},
      YEAR = {2025},
    NUMBER = {12},
     PAGES = {11565--11591},
      ISSN = {0170-4214,1099-1476},
   MRCLASS = {76N15 (35Q31)},
  MRNUMBER = {4930794},
       DOI = {10.1002/mma.10847},
       URL = {https://doi.org/10.1002/mma.10847},
}

@article{DGRR,
  author  = {Donadello, Carlotta and Gauthier, Alice and Razafison, Ulrich and Rosini, Massimiliano D.},
  title   = {Mathematical Modelling of One-Way Pressure-Activated Valves for Gas},
  journal = {ESAIM: Mathematical Modelling and Numerical Analysis},
  year    = {2026},
  note    = {To appear},
}
\bibliographystyle{abbrv2}
}

\end{document}